\documentclass[11pt]{article}

\usepackage{amsfonts,amsmath,amssymb,amsthm,graphicx}
\usepackage{fullpage}
\usepackage{changepage}
\usepackage{enumerate}
\usepackage{scalerel}
\usepackage{accents}
\usepackage[margin=1in]{geometry}
\usepackage{float}
\usepackage[caption = false]{subfig}
\usepackage{hyperref}
\usepackage{enumitem}

\usepackage{thm-restate}

\usepackage{natbib}
\usepackage{csquotes}
\usepackage{comment}

\usepackage{bbm}

\usepackage{cleveref}
\crefname{claim}{claim}{claims}
\Crefname{algocf}{Algorithm}{Algorithms}

\crefname{appendix}{appendix}{appendices}
\Crefname{appendix}{Appendix}{Appendices}

\let\oldappendix\appendix
\renewcommand{\appendix}{\oldappendix
  \crefalias{section}{appendix}}

\usepackage{tikz}
\usetikzlibrary{matrix,calc}
\usepackage{array,longtable}
\usepackage{booktabs}
\usepackage{pgfplots}
\pgfplotsset{compat=1.18}
\usetikzlibrary{patterns}
\usepgfplotslibrary{fillbetween}
\usetikzlibrary{intersections}

\usepackage[ruled,vlined,linesnumbered]{algorithm2e}
\Crefname{algocf}{Algorithm}{Algorithms}
\allowdisplaybreaks

\usetikzlibrary{arrows, decorations.markings}

\tikzstyle{vecArrow} = [thick, decoration={markings,mark=at position
   1 with {\arrow[semithick]{open triangle 60}}},
   double distance=1.4pt, shorten >= 5.5pt,
   preaction = {decorate},
   postaction = {draw,line width=1.4pt, white,shorten >= 4.5pt}]
\tikzstyle{innerWhite} = [semithick, white,line width=1.4pt, shorten >= 4.5pt]

\allowdisplaybreaks
\theoremstyle{plain}
\newtheorem{theorem}{Theorem}[section]
\newtheorem{lemma}[theorem]{Lemma}
\newtheorem{claim}[theorem]{Claim}

\newtheorem{proposition}[theorem]{Proposition}
\newtheorem{corollary}[theorem]{Corollary}

\theoremstyle{plain}
\newtheorem{definition}{Definition}[section] \newtheorem{example}[definition]{Example}
\newtheorem{remark}[definition]{Remark}

\theoremstyle{plain}

\usepackage{xfrac}

\newcommand{\xhdr}[1]{\vspace{2mm} \noindent{\bf #1}}

\newcommand{\primed}{^\dagger}
\newcommand{\doubleprimed}{^\ddagger}

\newcommand{\GFT}[2][]{\texttt{GFT}\ifthenelse{\not\equal{}{#1}}{_{#1}}{}\!\left[{\def\givenn{\middle|}#2}\right]}

\newcommand{\reals}{\mathbb{R}}
\newcommand{\eps}{\varepsilon}

\newcommand{\alloc}{x}
\newcommand{\price}{p}

\newcommand{\val}{v}
\newcommand{\buyerdist}{F}

\newcommand{\supp}{{\sf supp}}
\newcommand{\allocs}{\boldsymbol{\alloc}}
\newcommand{\prices}{\boldsymbol{\price}}
\newcommand{\revcurve}{R}

\newcommand{\buyercdf}{\buyerdist}
\newcommand{\buyerpdf}{f}

\newcommand{\threshold}{\epsilon}

\newcommand{\reserve}{r}
\newcommand{\optreserve}{\reserve^*}

\newcommand{\auxfunc}{\psi}

\newcommand{\naturals}{\mathbb{N}}

\newcommand{\osdensity}{\eta}

\newcommand{\quant}{q}

\newcommand{\expDists}{\mathcal{F}_{\textsc{TEXP}}}

\newcommand{\SLVCGAuction}{{\sf Supply-Limiting VCG Auction}}
\newcommand{\VCGAuction}{{\sf VCG Auction}}
\newcommand{\BayesianOptimalMech}{{\sf Bayesian-Optimal Mechanism}}

\newcommand{\RevVCG}{{\sf RevVCG}}
\newcommand{\RevOPT}{{\sf RevOPT}}
\newcommand{\RevSLVCG}{{\sf RevVCG}}

\newcommand{\ccomplexity}{\texttt{CC}}

\newcommand{\ccomplexityInfty}{\texttt{CC}_{(\infty)}}
\newcommand{\ccomplexitySL}{\ccomplexity^{\texttt{SL}}}
\newcommand{\ccomplexitySLInfty}{\ccomplexity_{(\infty)}^{\texttt{SL}}}

\newcommand{\MHRDists}{\mathcal{F}_{\text{MHR}}}
\newcommand{\RegDists}{\mathcal{F}_{\text{Reg}}}

\newcommand{\DistClass}{\mathcal{F}}
\newcommand{\ccapproxratio}{\Gamma}
\newcommand{\optquant}{\quant^*}
\newcommand{\supply}{s}

\newcommand{\quantSC}{\quant\primed}
\newcommand{\valSC}{\val\primed}

\newcommand{\imbalanceratio}{\alpha}
\newcommand{\Reglevel}{\lambda}

\newcommand{\WorstRegDist}{\buyerdist_{(\optreserve,\Reglevel)}}

\newcommand{\WorstRevcurve}{\revcurve_{(\optreserve,\Reglevel)}}

\newcommand{\CriticalQuantile}{\quant_0}
\newcommand{\shiftCriticalQuantile}{\abs{\quant-\CriticalQuantile}}
\newcommand{\addBuyers}{t_n}
\newcommand{\revLB}{\revcurve_{\min}}

\newcommand{\BetaFun}[2]{B\!\left(#1,#2\right)}

\newcommand{\BetaDist}[2]{Beta\!\left(#1,#2\right)}
\newcommand{\indexi}{i}
\newcommand{\indexj}{j}
\newcommand{\coefA}{a}
\newcommand{\AltBinSum}{S}
\newcommand{\digammaf}{\psi}
\newcommand{\Harm}[1]{H_{#1}}

\newcommand{\coefOld}{S_{\text{left}}}
\newcommand{\coefNew}{S_{\text{right}}}
\newcommand{\intOld}{I_1}
\newcommand{\intNew}{I_2}  
\newcommand{\binEntropy}{H}

\newcommand{\Auxfunc}{H}
\newcommand{\ExpectX}{\mu}

\newcommand{\funcC}{c}

\newcommand{\variablex}{x}
\newcommand{\variableu}{u}
\newcommand{\variablet}{y}
\newcommand{\pdf}[1]{p\left(#1\right)}

\newcommand{\rprob}{p_+(\optquant)}
\newcommand{\drprob}{p_+^\prime(\optquant)}
\newcommand{\lnexpect}{h(\optquant)}
\newcommand{\dlnexpect}{h^\prime(\optquant)}
\newcommand{\VCGrpart}{F(\optquant)}
\newcommand{\dVCGrpart}{F^\prime(\optquant)}
\newcommand{\IF}{\mathbf{1}}

\newcommand{\kratio}{\theta}

\newcommand{\qLeft}{q_L}
\newcommand{\qRight}{q_R}

\newcommand{\valvec}{\boldsymbol{v}}

\newcommand{\hisval}[1]{\val_{#1}}

\newcommand{\devval}[1]{\bid} 

\newcommand{\alloci}[2]{\alloc_{#1}\!\left(#2\right)}
\newcommand{\pricei}[2]{\price_{#1}\!\left(#2\right)}

\newcommand{\truncGPDists}{\mathcal{F}_{(\lambda)}^*}
\newcommand{\lambdaDists}{\mathcal{F}_{(\lambda\text{-Reg})}}
\newcommand{\distLambdaReserve}{F_{(\lambda,\reserve)}}
\newcommand{\distExpReserve}{F_{(0,\reserve)}}

\newcommand{\DiracDelta}{\delta}
\newcommand{\bids}{\boldsymbol{b}}
\newcommand{\bid}{b}

	\DeclareMathOperator{\argmax}{argmax}

\newcommand{\condition}{\,\mid\,}

\newcommand{\prob}[2][]{\text{Pr}\ifthenelse{\not\equal{}{#1}}{_{#1}}{}\!\left[{\def\givenn{\middle|}#2}\right]}
\newcommand{\expect}[2][]{\mathbb{E}\ifthenelse{\not\equal{}{#1}}{_{#1}}{}\!\left[{\def\givenn{\middle|}#2}\right]}

\newcommand{\tparen}{\big}
\newcommand{\tprob}[2][]{\text{Pr}\ifthenelse{\not\equal{}{#1}}{_{#1}}{}\tparen[{\def\given{\tparen|}#2}\tparen]}
\newcommand{\texpect}[2][]{\mathbb{E}\ifthenelse{\not\equal{}{#1}}{_{#1}}{}\tparen[{\def\given{\tparen|}#2}\tparen]}

\newcommand{\sprob}[2][]{\text{Pr}\ifthenelse{\not\equal{}{#1}}{_{#1}}{}[#2]}
\newcommand{\sexpect}[2][]{\mathbb{E}\ifthenelse{\not\equal{}{#1}}{_{#1}}{}[#2]}

\newcommand{\dd}{{\mathrm d}}

\newcommand{\indicator}[1]{{\mathbbm{1}\left\{ #1 \right\}}}

\newcommand{\plus}[1]{{\left( #1 \right)^+}}

\newcommand{\abs}[1]{{\left| #1 \right|}}

\title{Strengthening Bulow-Klemperer-Style Results\\ for Multi-Unit Auctions}

\author{Moshe Babaioff\thanks{Hebrew University of Jerusalem. Email: {\tt moshe.babaioff@mail.huji.ac.il}} 
\and Yiding Feng\thanks{Hong Kong University of Science and Technology. Email: {\tt ydfeng@ust.hk}}
\and Zihan Luo\thanks{Chinese Academy of Sciences. Email: {\tt luozihan23s@ict.ac.cn}}}

\date{}

\begin{document}

\maketitle
\begin{abstract}
    The classic result of \citet{BK-96} shows that in multi-unit auctions with $m$ units and $n\geq m$ buyers whose values are sampled i.i.d. from a regular distribution, the revenue of the VCG auction with $m$ additional buyers is at least as large as the optimal revenue. Unfortunately, for regular distributions, adding $m$ additional buyers is sometimes indeed necessary, so the  ``competition complexity'' of the VCG auction is $m$. We seek proving better competition complexity results in two dimensions.

First, under stronger distributional assumptions---specifically, $\Reglevel$-regularity, which interpolates between regularity and monotone hazard rate (MHR)---the competition complexity of VCG auction drops dramatically. In balanced markets (where $m=n$) with MHR distributions, it is sufficient to only add $(e^{1/e} - 1 + o(1))n \approx 0.4447n$ additional buyers to match the optimal revenue---less than half the number that is necessary under regularity---and this bound is asymptotically tight. We provide both exact finite-market results for small value of $n$, and closed-form asymptotic formulas for general market with any $m\leq n$, and any target fraction of the optimal revenue.

Second, we analyze a supply-limiting variant of VCG auction that caps the number of units sold in a prior-independent way. Whenever the goal is to achieve almost the optimal revenue---say, 90\% or 95\%---this mechanism strictly improves upon standard VCG auction, requiring significantly fewer additional buyers.

Together, our results show that both stronger distributional assumptions, as well as a simple prior-independent refinement to the VCG auction, can each substantially reduce the number of additional buyers that is sufficient to achieve (near-)optimal revenue. Our analysis hinges on a unified worst-case reduction to truncated generalized Pareto distributions, enabling both numerical computation and analytical tractability. \end{abstract}

\thispagestyle{empty}
\newpage

\section{Introduction}
\label{sec:intro}

One of the most fundamental goals in mechanism design is revenue maximization.  
Optimizing over the complex space of mechanisms to maximize the revenue is challenging, but for the single-item setting, \citet{mye-81} has presented his celebrate mechanism, which maximizes the seller's expected revenue in the Bayesian environment. However, this mechanism relies on the knowledge of the buyers' prior value distribution---a requirement that may be impractical in many real-world settings.
To address this limitation, \citet{BK-96} showed that for a single-item auction, when buyer values are drawn i.i.d.\ from a regular distribution, it suffices to recruit just one additional buyer from the same distribution and run the {\VCGAuction} (i.e., the second-price auction). This simple, prior-independent mechanism---requiring no knowledge of the prior---guarantees expected revenue at least as high as that of Myerson's optimal auction. This result suggests that, rather than optimizing over the complex, prior-dependent mechanism-design space to maximize revenue, a firm may instead benefit from market expansion and the use of simple and practical prior-independent mechanisms, such as the second-price auction.

In many practical applications, the seller may have multiple units of a homogeneous good to allocate among $n$ buyers. Prominent examples include online advertising (where an ad platform sells multiple impressions of the same ad slot), spectrum auctions (where a regulator allocates several identical frequency blocks), cloud computing markets (where a provider offers multiple instances of a virtual machine type), and emerging AI service platforms (where providers sell access to multiple identical queries or compute tokens from a large language model). As a natural extension of the single-unit setting, \citet{BK-96} show that when the number of units is $m \leq n$ and buyers' values are drawn i.i.d.\ from a regular distribution, recruiting $m$ additional buyers and then running the {\VCGAuction}\footnote{{\VCGAuction} is a natural extension of the second-price auction. See \Cref{def:vcg auction}.} \citep{vic-61} achieves weakly higher revenue than the {\BayesianOptimalMech}.

Remarkably, the result in \citet{BK-96} is tight: the following folklore example 
exhibits a market with a regular valuation distribution for which recruiting only $m-1$ additional buyers and running the {\VCGAuction} yields strictly lower revenue than the {\BayesianOptimalMech}.\footnote{The formal analysis of \Cref{example:intro:multi-unit bk:regular} can be found in \Cref{apx:intro example analysis}.}
\begin{example}
\label{example:intro:multi-unit bk:regular}
    Fix any $n \geq 1$, and let $H \triangleq 3n - 1$. Consider an $n$-unit, $n$-buyer (balanced) market where buyers' values are drawn i.i.d.\ from a regular distribution $\buyerdist$ defined as follows: its support is $\supp(\buyerdist) = [0, H]$, and its cumulative distribution function satisfies $\buyercdf(\val) = 1 - \frac{1}{\val + 1}$ for $\val < H$, with an atom of mass $\frac{1}{H+1}$ at $\val = H$.

    In this market, the {\BayesianOptimalMech} posts a take-it-or-leave-it price of $H$ to all buyers and achieves an expected revenue of $n \cdot \left(1 - \frac{1}{H+1}\right) = n - \frac{1}{3}$.
    In contrast, the revenue of the {\VCGAuction} with $n-1$ additional buyers (i.e., $2n - 1$ total buyers) is 
strictly less than $n - \frac{1}{2}$.
\end{example}

While \Cref{example:intro:multi-unit bk:regular} illustrates the necessity of recruiting $n$ additional buyers (which, in this instance, amounts to duplicating the entire market), one might argue that the valuation distribution used is somewhat unrealistic: {even in a large market ($n\gg 0$),
it forces the {\BayesianOptimalMech} to
serve less than 1 buyer in expectation}---specifically, $1/(H+1) = 1/(3n)$---who happen to have extremely high valuations. For instance, when $n = 10{,}000$, 
the mechanism never serves any buyer with a value that is not in the top $0.01\%$ percentile.
We thus see that the family of regular distributions is premise enough to allow for such phenomenon. This motivates the following question regarding further restricting the family of distributions:
\begin{displayquote}
   (Question 1) 
\emph{Under standard distributional assumptions stronger than regularity, is it sufficient to add much fewer buyers to guarantee that the {\VCGAuction} always outperforms the {\BayesianOptimalMech}?}
\end{displayquote}
In a nutshell, our work shows that significantly smaller number of additional buyers is sufficient under the stronger conditions of the monotone hazard rate (MHR) and $\Reglevel$-regularity.

Taking a step back, one can view ``Bulow–Klemperer-style results'' as aiming to minimize the number of additional buyers that is sufficient and necessary to match or exceed the revenue of the {\BayesianOptimalMech} by \emph{some prior-independent mechanism}. 
It is not at all clear that the {\VCGAuction} is the best prior-independent mechanism to consider when the goal is to  establish such a result.  
The core issue lies in a fundamental mismatch of objectives: the {\VCGAuction} maximizes social welfare and thus allocates units as long as there is demand whereas the {\BayesianOptimalMech} strategically restricts supply to extract high revenue from the most valuable buyers, even at the cost of efficiency.

This tension is starkly illustrated in \Cref{example:intro:multi-unit bk:regular}. There, the {\BayesianOptimalMech} serves only a tiny fraction of buyers---only about $1/3$ buyer in expectation---but charges the maximum possible price $H$, yielding high revenue. In contrast, the {\VCGAuction} serves all $n$ units, but at much lower prices, resulting in strictly lower total revenue. 
The revenue gap arises not from insufficient competition, but from the {\VCGAuction}'s inherent inability to replicate the revenue-driven supply discipline of the Bayesian optimum.

This observation suggests that a prior-independent mechanism which internalizes this strategic supply restriction---without knowing the distribution---could significantly reduce the number of additional buyers that is sufficient to beat the {\BayesianOptimalMech}. A natural candidate is a \emph{supply-limiting VCG auction}, which caps the number of allocated units in a prior-independent way, based solely on the market size (e.g., the number of buyers), thereby emulating the cautious allocation behavior of the {\BayesianOptimalMech} in high-value-tail scenarios. 
Motivated by this, we pose the following question:
\begin{displayquote}
   (Question 2) 
\emph{By considering prior-independent variants of the {\VCGAuction}---such as optimally limiting its supply without knowing the distribution---is it possible to show that a much smaller number of additional buyers is always sufficient for such a mechanism to outperform, or nearly outperform, the {\BayesianOptimalMech}?}
\end{displayquote}
Indeed, our work shows that to approximately outperform the optimal revenue, it is sufficient for a prior-independent supply-limiting variant of {\VCGAuction} to add significantly fewer additional buyers.

\subsection{Our Contributions}
\label{sec:intro:results}
In this work, we provide compelling answers to both questions. Below we explain our contributions. Throughout, we consider a Bayesian environment where a seller sells $m$ units of a good to $n\geq m$ buyers with values drawn i.i.d.\ from a common valuation distribution $\buyerdist$. For ease of presentation, we follow recent literature on Bulow–Klemperer-style results, refer to the number of additional buyers that is sufficient and necessary as the \emph{competition complexity} \citep{EFFTW-17} (see formal definition in \Cref{def:competition complexity}).

\xhdr{Balanced Markets with MHR Distributions (\Cref{sec:bk mhr}).}
We begin our investigation by considering a balanced market (i.e., $n = m$) and refer to $n$ as the market size. Recall that \Cref{example:intro:multi-unit bk:regular} shows that $n$ additional buyers are indeed sometimes necessary for the {\VCGAuction} to achieve revenue at least as high as the Bayesian optimal revenue in the original balanced market when the valuation distribution $\buyerdist$ is regular.

Our first set of results demonstrates that under the stronger distributional assumption of monotone hazard rate (MHR)---a canonical condition widely studied in the literature---it is sufficient to add much fewer buyers to outperform {\BayesianOptimalMech}.
Specifically, for all market sizes $n \geq 3$, strictly fewer than $n$ additional buyers suffice, 
and this number quickly drops relative to $n$, and for $n \geq 23$ it is already below half of $n$.

To establish these improved guarantees, we first characterize a subclass of MHR distributions that contains the worst case for competition complexity. In \Cref{prop:bk vcg:mhr worst case}, we show that restricting attention to the class of truncated exponential distributions (\Cref{def:truncated exponential distribution})---a natural subclass of MHR distributions---does not reduce the competition complexity of the {\VCGAuction}. This class admits a simple closed-form cumulative distribution function (CDF), parameterized by a single truncation parameter $\reserve \in [1, e]$.

This worst-case characterization serves two purposes. First, it provides structural insight into the nature of the hardest instances for the {\VCGAuction}. Second, it enables both numerical and analytical refinements of the competition complexity. 

Specifically, leveraging this reduction, we can compute the exact competition complexity for any finite market size $n$ via a computer-aided grid search over discretized values of $\reserve$, combined with a Lipschitz continuity argument to bound the error for intermediate values. Using this approach, we determine the exact competition complexity for all $n \leq 593$ and report the results in \Cref{sec:numerical experiment}.

Furthermore, by again exploiting the worst-case characterization, we analytically derive a uniform upper bound on the competition complexity for all $n \geq 594$ (\Cref{thm:bk vcg:mhr finite market}). This bound is less than $0.456 \cdot n$ for all such $n$, and it essentially converges to $(e^{1/e} - 1) \cdot n \approx 0.4447 \cdot n$ as $n \to \infty$. 
Within the same theorem,
we also complement this upper bound with an asymptotically matching lower bound.

\xhdr{General markets with $\Reglevel$-regular distributions (\Cref{sec:bk vcg}).}
We next generalize our results for balanced markets with MHR distributions along three natural dimensions:  
(i) a broader class of $\Reglevel$-regular distributions, which smoothly interpolates between regularity ($\Reglevel = 1$) and MHR ($\Reglevel = 0$);  
(ii) general (possibly imbalanced) markets with $n$ buyers and $m \leq n$ units, capturing a wider range of practical scenarios---from scarce goods ($m \ll n$) to abundant supply ($m \approx n$); and  
(iii) relaxing the requirement to fully match the optimal Bayesian revenue to achieving a $\ccapproxratio$-fraction of it.\footnote{In practice, recruiting additional buyers may be costly. Allowing a target approximation ratio $\ccapproxratio < 1$ helps quantify the trade-off between the cost of recruiting extra buyers and the quality of the revenue guarantee.}

For this more general setting, we show that the worst-case characterization established for balanced markets under MHR (\Cref{prop:bk vcg:mhr worst case}) admits a natural extension. Specifically, as formalized in \Cref{prop:bk vcg:worst case}, for any $\Reglevel \in [0,1]$, any market size with $m \leq n$, and any target approximation ratio $\ccapproxratio \in (0,1]$, the worst-case distribution belongs to the class of truncated $\Reglevel$-generalized Pareto distributions (\Cref{def:truncated generalized Pareto distribution}). This class admits a closed-form CDF and is parameterized by a single truncation parameter $\reserve$, whose range is bounded for every $\Reglevel \in [0,1)$. Remarkably, while the extension is conceptually natural, the analysis becomes significantly more intricate compared to the balanced-market case; see further discussion in \Cref{sec:intro:techniques}.

Leveraging this worst-case characterization, we shed light on the competition complexity in this general setting. To ensure tractability and obtain clean asymptotic results, we consider the following large-market regime: let $\imbalanceratio \in [0,1]$ denote the \emph{supply-to-demand ratio}, i.e., the asymptotic fraction of units per buyer. We study the limit as $n \to \infty$ with $m = \lceil \imbalanceratio \cdot n \rceil$ units. Additionally, we assume that the target approximation ratio $\ccapproxratio$ is an absolute constant strictly less than $1$.
In \Cref{thm:bk vcg:large market}, for every $\imbalanceratio \in [0,1]$ and $\ccapproxratio \in (0,1)$, we derive a closed-form expression for the \emph{asymptotic competition complexity}---defined as the limit of the competition complexity in an $\lceil \imbalanceratio \cdot n \rceil$-unit, $n$-buyer market, normalized by $n$, as $n \to \infty$.

\begin{figure}
    \centering

\begin{tikzpicture}
\begin{axis}[
width=10cm,height=7cm,
xmin=0.0,xmax=1.05,
ymin=-0.015,ymax=1.05,
xtick={0, 0.25, 0.367, 0.5, 1},   
xticklabels = {0, 0.25, $1/e$, 0.5, 1},
tick label style={font=\footnotesize}, scaled y ticks=false,                 yticklabel style={/pgf/number format/fixed}, ytick={0, 0.05, 0.15, 0.5, 1}, 
minor tick num=0,
axis line style=gray,
axis x line=middle,
axis y line=left,
tick style={black},
legend style={draw=none,fill=white,font=\small,at={(0.98,0.98)},anchor=north east},
xlabel={\footnotesize $\imbalanceratio$},
ylabel={\footnotesize asymptotic CC},
]

\def\b{1.0}
\def\ee{2.718281828} 

\pgfplotsset{
  myLine/.style={line width=2pt}
}

\addplot[black, myLine, domain=0:1, samples=400]
  {max(0, \b + x - 1)};

\addplot[black!45, myLine, domain=0:1, samples=400]
  {ifthenelse(x<0.25, 0, max(0, x*(0.25*(\b/x) + 1)^2 - 1))};

\addplot[black!20, myLine, domain=0:1, samples=400]
  {ifthenelse(x<1/\ee, 0, max(0, x*exp(\b/(\ee*x)) - 1))};

\addplot[dotted](1, 0) -- (1, 1) -- (0, 1);
\addplot[dotted](1, 0.4447) -- (0, 0.4447);
\addplot[dotted](1, 0.5625) -- (0, 0.5625);
\addplot[dotted](0.5, 0) -- (0.5, 0.5) -- (0, 0.5);
\addplot[dotted](0.5, 0.125) -- (0, 0.125);
\addplot[dotted](0.5, 0.0435) -- (0, 0.0435);

\end{axis}
\end{tikzpicture} \caption{The asymptotic competition complexity (CC) (to achieve 99.99\% of the Bayesian optimal revenue)
as a function of supply-demand ratio $\imbalanceratio\in[0, 1]$. The black, dark-gray, and light-gray curves correspond to regular $\RegDists$, $0.5$-regular $\mathcal{F}_{\text{0.5-Reg}}$, and MHR distributions $\MHRDists$, respectively.}
\label{fig:bk vcg:large market}
\end{figure}
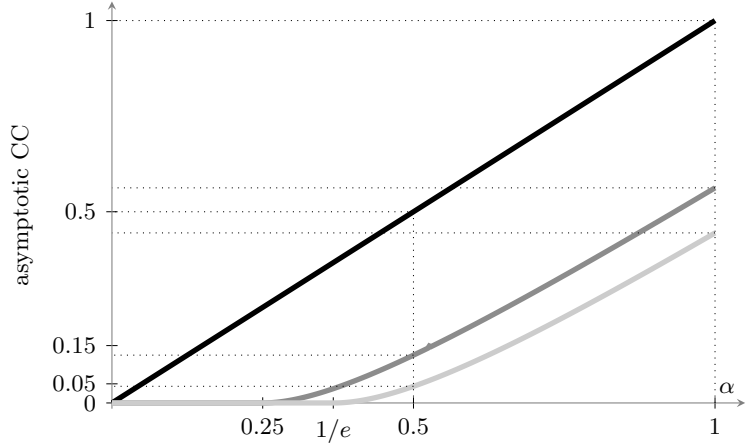

The resulting bound is illustrated in \Cref{fig:bk vcg:large market} for the goal of achieving 99.99\% of the Bayesian optimal revenue. Remarkably:
\begin{itemize}
    \item In a balanced market ($\imbalanceratio = 1$), it suffices to add only 44.47\% (MHR) or 56.25\% ($0.5$-regular) of the original number of buyers.
    \item In a market with half as many units as buyers ($\imbalanceratio = 0.5$), just 4.35\% (MHR) or 12.5\% ($0.5$-regular) additional buyers are always sufficient.
\item By contrast, for regular distributions ($\Reglevel = 1$), one must add 100\% (balanced) or 50\% (half as many units as buyers) additional buyers to achieve the same guarantee.
\end{itemize}
These results highlight how stronger regularity assumptions---specifically, moving from regular to MHR distributions---dramatically reduce the competition complexity of the {\VCGAuction}.

\xhdr{Supply-limiting VCG Auction (\Cref{sec:bk supply limiting}).} Finally, we study a natural supply-limiting variant of the {\VCGAuction}, where the seller choice a supply $\supply\leq m$ and then implement the {\VCGAuction} to sell $\supply$ units. Importantly, the supply $\supply$ is chosen without the knowledge of the valuation distribution (i.e., in a prior-independent manner), to minimize the number of additional buyers for the worst-case distribution.

First, we show in \Cref{prop:bk supply limiting:worst case}, that for any choice of supply, the worst-case distribution still belongs to the class of truncated $\Reglevel$-generalized Pareto distributions (\Cref{def:truncated generalized Pareto distribution}). This results strictly generalize the worst-case characterization in \Cref{prop:bk vcg:mhr worst case,prop:bk vcg:worst case}, since the {\VCGAuction} is a special case of {\SLVCGAuction} with $\supply = m$. 

\begin{figure}
    \centering

\begin{tikzpicture}
\begin{axis}[
width=10cm,height=7cm,
xmin=0.0,xmax=1.05,
ymin=-0.015,ymax=1.05,
xtick distance=1,
ytick distance=1,
minor tick num=0,
axis line style=gray,
axis x line=middle,
axis y line=left,
tick style={black},
legend style={draw=none,fill=white,font=\small,at={(0.98,0.98)},anchor=north east},
xlabel={\footnotesize $\ccapproxratio$},
ylabel={\footnotesize asymptotic CC},
]

\def\a{1.0}
\def\ee{2.718281828} 

\pgfplotsset{
  solidStyle/.style={line width=1.3pt},
  dashedStyle/.style={line width=2pt, dashed, dash pattern=on 4pt off 3pt}
}

\addplot[black, solidStyle, domain=0:1, samples=400] {max(0, \a + x - 1)};

\addplot[black, dashedStyle, domain=0:1, samples=400] {max(0, \a * x + x - 1)};

\addplot[black!45, solidStyle, domain=0:1, samples=400]
  {max(0, \a*(0.25*(x/\a) + 1)^2 - 1)};

\addplot[black!45, dashedStyle, domain=0:1, samples=400]
  {max(0, \a*x*(0.25*(1/\a) + 1)^2 - 1)};

\addplot[black!20, solidStyle, domain=0:1, samples=400]
  {max(0, \a*exp(x/(\ee*\a)) - 1)};

\addplot[black!20, dashedStyle, domain=0:1, samples=400]
  {max(0, \a*x*exp(1/(\ee*\a)) - 1)};
\end{axis}
\end{tikzpicture}

 \caption{
The asymptotic competition complexity (CC) of the {\VCGAuction} (solid curve) and the asymptotic optimal competition complexity (CC) of the {\SLVCGAuction} (dashed curve) as a function of target approximation ratio $\ccapproxratio\in(0, 1)$ in the balanced market (i.e., $\imbalanceratio=1$). The black, dark-gray, and light-gray curves correspond to regular $\RegDists$, $0.5$-regular $\mathcal{F}_{\text{0.5-Reg}}$, and MHR distributions $\MHRDists$, respectively.}
\label{fig:VCG-vs-SLVCG}
\end{figure}
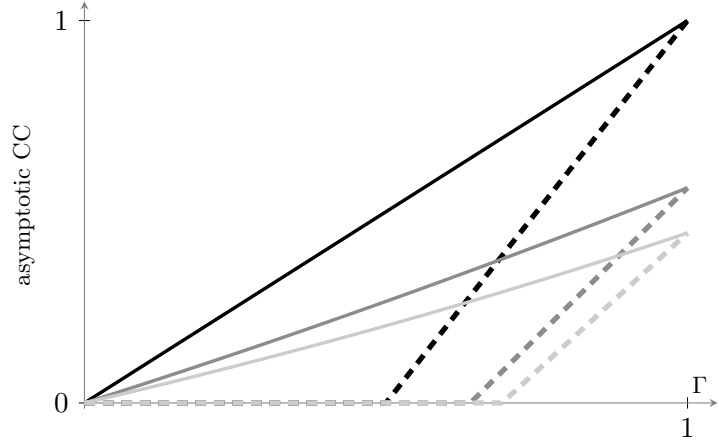

Next, we characterize the asymptotic competition complexity of the {\SLVCGAuction} in \Cref{thm:bk supply limiting:large market} for every supply-to-demand ratio $\imbalanceratio \in (0,1]$, target approximation ratio $\ccapproxratio \in (0,1)$, and $\Reglevel$-regularity level $\Reglevel \in [0,1]$. In particular, in this asymptotic large-market regime, we show that an optimal choice of supply is to sell approximately a $\ccapproxratio$-fraction of the total units.

The resulting bound is illustrated in \Cref{fig:VCG-vs-SLVCG} for balanced markets ($\imbalanceratio = 1$). Remarkably:
\begin{itemize}
    \item As the target approximation ratio $\ccapproxratio$ approaches 1, the asymptotic competition complexity of the {\SLVCGAuction} with optimally chosen supply $\supply$ converges to that of the standard {\VCGAuction}. In other words, when the goal is exact optimality ($\ccapproxratio = 1$), prior-independent supply limiting offers no advantage over the standard {\VCGAuction}.

    \item However, when $\ccapproxratio < 1$, the {\SLVCGAuction} strictly improves upon the standard {\VCGAuction}. For example, when $\ccapproxratio = 0.8$, the {\SLVCGAuction} 
it is sufficient to add only $15.6\%$, $25\%$, and $60\%$ additional buyers for $\MHRDists$, $\mathcal{F}_{\text{0.5-Reg}}$, and $\RegDists$, respectively---compared to $34.2\%$, $44\%$, and $80\%$ for the {\VCGAuction}. When $\ccapproxratio = 0.6$, it is sufficient for the {\SLVCGAuction} to only add a \emph{sublinear} number of additional buyers for both $\MHRDists$ and $\mathcal{F}_{\text{0.5-Reg}}$, and only $20\%$ for $\RegDists$, whereas it is still necessary for the {\VCGAuction} to add $24.7\%$, $32.3\%$, and $60\%$, respectively.
\end{itemize}
To understand the rationale behind the optimal supply choice and the above observations, consider the following trade-off inherent to the {\SLVCGAuction}:  
As the supply $\supply$ decreases, competition among buyers intensifies, improving performance on high-value-tail distributions (e.g., as in \Cref{example:intro:multi-unit bk:regular}). However, reducing $\supply$ simultaneously degrades performance on the opposite extreme---a single-point mass distribution at value 1---where revenue is proportional to the number of units sold. In fact, the optimal supply is precisely the smallest $\supply$ such that the target approximation ratio $\ccapproxratio$ is achieved under this deterministic (point-mass) instance.
\subsection{Our Techniques}
\label{sec:intro:techniques}
In most prior work on Bulow–Klemperer-style results for regular distributions \citep[e.g.,][]{BK-96,HR-09,RTY-20,FLR-19}, the analysis typically relies on charging arguments or virtual surplus comparisons, and an explicit characterization of  worst-case distributions is unnecessary. In contrast, to obtain tight or nearly tight characterizations of competition complexity under stronger distributional assumptions or alternative mechanisms, we require a more refined approach.

A key ingredient in our analysis is a \emph{worst-case reduction} (\Cref{prop:bk vcg:mhr worst case,prop:bk vcg:worst case,prop:bk supply limiting:worst case}): we show that to understand the competition complexity of the {\VCGAuction} (and its supply-limiting variant), it suffices to consider a very simple family of distributions---namely, truncated $\Reglevel$-generalized Pareto distributions (\Cref{def:truncated generalized Pareto distribution}). This family is parameterized by just a single parameter (the truncation point), which makes both numerical and analytical treatment tractable.

The main challenge in proving this reduction stems from a subtle technical issue: in multi-unit settings, the revenue gap between {\VCGAuction} and the {\BayesianOptimalMech} does \emph{not} behave monotonically with respect to standard stochastic dominance. In other words, a ``worse'' distribution (in the usual sense) does not always lead to a larger revenue gap, since the expected revenue of both mechanisms decreases. This breaks a common proof strategy used in the literature.

To overcome this, we shift our focus from the multiplicative revenue gap to an additive one, and uncover a clean \emph{structured three-interval property}: the behavior of the gap can be decomposed into three quantile intervals. While non-monotonicity appears in the middle interval, it does so in a consistent and controllable way. Leveraging the concavity of revenue curves under $\Reglevel$-regularity, we show that the worst case is always attained within the class of truncated $\Reglevel$-generalized Pareto distributions. The proof proceeds via a careful case analysis based on the location of the monopoly quantile relative to the three quantile intervals.

This reduction not only resolves the technical obstacle but also serves as a powerful tool: it allows us to compute exact competition complexity for small markets via numerical search. In addition, it also serves as a powerful analytical tool. For large markets, where the target approximation ratio $\ccapproxratio < 1$, it enables clean asymptotic results (\Cref{thm:bk vcg:large market,thm:bk supply limiting:large market}) via standard concentration inequalities and the slack introduced by $\ccapproxratio < 1$. However, when the goal is \emph{exact} optimality ($\ccapproxratio = 1$)---as in \Cref{thm:bk vcg:mhr finite market} for balanced markets under MHR distributions---the lack of slack makes the analysis substantially more delicate. 

We next explain the key challenge and technical approach behind establishing the uniform upper bound in \Cref{thm:bk vcg:mhr finite market}. Even after reducing the problem to the one-parameter family of truncated exponential distributions, a significant obstacle remains: the worst-case truncation point $\reserve \in [1,e]$ admits no closed-form expression and varies with the market size $n$ (\Cref{remark:worst truncation}). This rules out simple endpoint analyses.

To overcome this, we develop a refined four-part case analysis over the interval $[1,e]$, leveraging precise derivative estimates and concentration properties of order statistics (see \Cref{lem:VCG revenue small reserve,lem:VCG revenue middle reserve,lem:VCG revenue high reserve,lem:VCG revenue super high reserve,lem:VCG revenue e reserve}). This careful decomposition allows us to certify that the bound stated in the theorem suffices for all market size $n\geq 594$. The competition complexity for market size $n\leq 593$ is covered by the numerical computations.

\subsection{Related Work}

\xhdr{Bulow-Klemperer-style results.} 
Since the seminal work of \citet{BK-96}, a rich line of literature has sought to establish Bulow–Klemperer-style results in various settings. In single-parameter environments, \citet{HR-09,SS-13,FLR-19,FJ-24} study the competition complexity of the {\VCGAuction} when buyers’ valuations are drawn from non-identical regular or even weaker than regular distributions. \citet{BGG-20,CLMZ-24} consider two-sided markets and analyze the competition complexity of the trade-reduction mechanism for maximizing gains from trade. \citet{BCDV-22} investigate the competition complexity of (prior-dependent) dynamic pricing in the single-item setting under regular distributions.
Like one of the regimes studied in our work, these results in \citet{HR-09,SS-13,FLR-19,FJ-24,BCDV-22} focus on \emph{approximating} the Bayesian optimal revenue rather than fully outperforming it.

A closely related work is \citet{RTY-20}, which studies revenue guarantees under competition enhancement for regular distributions. They extend the multi-unit competition complexity result of \citet{BK-96} to general matching markets and also show that the {\SLVCGAuction} can achieve nontrivial revenue approximation \emph{without} adding extra buyers in such markets. In contrast, our work provides a comprehensive analysis of the competition complexity for both the {\VCGAuction} and the {\SLVCGAuction} under standard distributional assumptions that are \emph{stronger} than regularity---namely, $\Reglevel$-regularity and MHR.

There is also a growing body of work on competition complexity in multi-parameter environments, initiated by \citet{EFFTW-17}. For additive valuations, tight bounds on the competition complexity of the {\VCGAuction} have been established by \citet{BW-19,DRWX-24}. \citet{FFR-18,CS-21} study competition complexity for broader classes of auctions under regular distributions. Closer to our approach, \citet{LP-18,BCFLW-25} analyze the competition complexity of the {\VCGAuction} and the grand-bundle second-price auction for additive valuation functions under $\Reglevel$-regular and MHR distributions---paralleling our focus on stronger distributional assumptions.

\xhdr{Revenue approximation in the $m$-unit setting.}
There is also a body of literature on the multi-unit setting that bounds the revenue ratio between various simple mechanisms and the Bayesian optimal benchmark. The revenue guarantees of posted pricing mechanisms have been extensively studied \citep[e.g.,][]{HKS-07,CHMS-10,yan-11,DFK-16,FHL-20,JJLZ-21}. \citet{JJLZ-21} also analyze the revenue approximation achieved by the second-price auction with an anonymous reserve. Closely related to the multi-unit setting, \citet{CJK-20} study anonymous pricing mechanisms for buyers with concave valuations.

\section{Preliminaries}
\label{sec:prelim}

In this work, we study revenue maximization in a Bayesian multi-unit multi-buyer setting, aiming to match the optimal revenue by 
a prior-independent mechanism with some additional buyers.

\xhdr{Model.} 
A seller has $m$ identical units of a good to sell to $n\geq m$ unit-demand buyers. 
Each buyer $\indexi$ has a private value $\hisval{\indexi}\geq 0$ for receiving one unit, where the values are independently and identically (i.i.d.) drawn from a publicly known valuation distribution~$\buyerdist$, supported on non-negative numbers.  
We use $\buyerdist$ interchangeably to denote both the distribution and its cumulative distribution function (CDF), i.e., $\buyerdist(\val) = \Pr_{\variablex \sim \buyerdist}[\variablex < \val]$.  
We assume that $\buyerdist$ has a probability density function and denote its corresponding probability density function (PDF) by $\buyerpdf$. We also assume that there exists a price $\price \in \supp(\buyerdist)$ that maximizes the expected revenue $\price \cdot (1 - \buyercdf(\price))$ for the single-buyer problem.

\xhdr{Mechanisms.}
A selling \emph{mechanism} consists of an \emph{allocation rule} $\allocs(\bids)=\left(\alloci{1}{\bids},\ldots,\alloci{n}{\bids}\right)$ and a \emph{payment rule} $\prices(\bids)=\left(\pricei{1}{\bids},\ldots,\pricei{n}{\bids}\right)$, where $\alloci{\indexi}{\bids}\in[0,1]$ denotes the probability that buyer $\indexi$ receives one unit of the good when the reported bid profile is $\bids=(\bid_1,\bid_2,\ldots,\bid_n)$, and $\pricei{\indexi}{\bids}\in\reals_+$ denotes the payment charged to buyer $\indexi$. Feasibility relies on that at most $m$ units are allocated ex post, 
i.e., $\sum_{\indexi=1}^{n}\alloci{\indexi}{\bids}\le m$ for all $\bids$. Buyers aim to maximize their quasi-linear utility, which is $\alloci{\indexi}{\bids}\cdot \bid_i-\pricei{\indexi}{\bids}$ for buyer $i$ with allocation $\alloci{\indexi}{\bids}$ that is charged a payment of $\pricei{\indexi}{\bids}$.

A mechanism satisfies \emph{Bayesian incentive compatible (BIC)} if every buyer maximizes her interim utility (that is, her expected utility over the randomness of all other buyers' value realizations and randomness of the mechanism) by bidding her value truthfully when all other buyers bids truthfully; and satisfies \emph{dominant strategy incentive compatible (DSIC)} if every buyer maximizes her ex post utility regardless of all other buyers' reports. For mechanisms that are DSIC or BIC, we assume buyers report their private values truthfully, and thus interpret the allocation and payment rules $(\allocs, \prices)$ of a given mechanism as functions of the true valuation profile $\valvec = (\val_1, \dots, \val_n)$, mapping it to buyers' allocations and payments, respectively.

A DSIC/BIC mechanism satisfies \emph{interim individual rational (interim IR)} if every buyer's interim utility is non-negative for all her values; and satisfies \emph{ex post individual rational (ex post IR)} if every buyer's ex post utility is non-negative for all valuation profiles.

\xhdr{Classic distribution families.} 
In the mechanism design literature, the regularity and monotone hazard rate (MHR) conditions are two classic distributional assumptions that have been studied extensively:

\begin{itemize}
    \item \textbf{(Regularity)} A distribution $\buyerdist$ satisfies the \emph{regularity condition} if its \emph{virtual value} $
    \val-\frac{1-\buyercdf(\val)}{\buyerpdf(\val)}$ is non-decreasing in $\val$.
    \item \textbf{(MHR)} A distribution $\buyerdist$ satisfies the \emph{monotone hazard rate (MHR) condition} if its \emph{hazard rate}
$\frac{\buyerpdf(\val)}{1-\buyercdf(\val)}$
is non-decreasing in $\val$.
\end{itemize}
\noindent
We also recall the $\Reglevel$-regularity condition which interpolates between these two conditions in a parametric way.
The $\Reglevel$-regularity condition is also known as the \emph{$\alpha$-strongly regular condition} in the literature \citep{RR-14,RR-15}.
\begin{itemize}
    \item \textbf{($\Reglevel$-Regularity)} For any $\Reglevel\in[0,1]$, 
a distribution $\buyerdist$ satisfies the \emph{$\Reglevel$-regularity condition} if its \emph{$\Reglevel$-generalized virtual value} $
\Reglevel\cdot\val-\frac{1-\buyercdf(\val)}{\buyerpdf(\val)}$ is non-decreasing in $\val$.
\end{itemize}

\noindent By definition, it can be verified that the $\Reglevel$-regularity with $\Reglevel=1$ and $\Reglevel=0$ recover the classic regularity and MHR conditions, respectively. 
We also denote the classes of regular distributions, MHR distributions, and $\Reglevel$-regular distributions, by $\RegDists$, $\MHRDists$, and {$\lambdaDists$},  respectively.

\xhdr{Bayesian revenue maximization.}
The seller's objective is to maximize her expected \emph{revenue}, defined as the sum of payments collected from all buyers: $\expect[\valvec]{\sum_{\indexi\in[n]}\pricei{\indexi}{\valvec}}$. For ease of presentation, we let ${\sf RevX}_{m:n}(\buyerdist)$ denote the expected revenue of truthful (DSIC/BIC) mechanism ${\sf X}$ when selling $m$ units of an identical good to $n$ truthful buyers whose values are drawn i.i.d.\ from $\buyerdist$.
When the distribution is clear from the context, we omit $\buyerdist$ and simply write ${\sf RevX}_{m:n}$.

We recall Myerson's characterization of the Bayesian-optimal auction for selling $m$ identical units when buyer valuations are drawn i.i.d.\ from a regular distribution (Myerson's  {\BayesianOptimalMech} maximizes the revenue under the BIC and interim IR constraints).

\begin{lemma}[\citealp{mye-81}]
\label{lem:Bayesian optimal mechanism}
    Consider the setting of selling $m$ units of an identical item to $n\geq m$ buyers with values drawn i.i.d.\ from a regular distribution $\buyerdist$.
Fix a monopoly reserve $\optreserve \in \argmax\limits_{\price \geq 0} \, \price \cdot (1 - \buyerdist(\price))$. 
    The {\BayesianOptimalMech}  with monopoly reserve $\optreserve$ operates as follows: allocate one unit to each of the $m$ highest-valued buyers whose value is at least monopoly reserve $\optreserve$. Each allocated buyer pays the maximum of the monopoly reserve $\optreserve$ and the $(m+1)$-th highest value; all other buyers pay $0$.
\end{lemma}

Notably, the {\BayesianOptimalMech} is \emph{distribution-dependent}: implementing the monopoly reserve $\optreserve$ relies on the knowledge of the valuation distribution $\buyerdist$. 
Motivated by the Wilson doctrine \citep{wil-87}, it is therefore desirable to design mechanisms whose performance guarantees do not rely on such detailed prior information, which leads to the study of \emph{prior-independent mechanisms}. 
{One important prior-independent mechanism studied in the literature---and the main focus of our paper---is the {\VCGAuction}:}
\begin{definition}
\label{def:vcg auction}
In an $m$-unit $n$-buyer market with $m\leq n$, the {\VCGAuction} allocates one unit to each of the $m$ highest-value buyers, and charges each of them the $(m+1)$-highest value as the price. 
\end{definition}

The {\VCGAuction} is DSIC, ex post IR, and welfare-maximizing by construction. Remarkably, its implementation does not relies on any knowledge regarding the valuation distribution $\buyerdist$, and hence it is prior-independent. When the valuation distribution is regular, it is similar to the {\BayesianOptimalMech} and the only difference is that it does not have a prior-dependent monopoly reserve. 
While the {\VCGAuction} is welfare-maximizing, its revenue is in general strictly smaller than the revenue of the {\BayesianOptimalMech}.

\xhdr{The Bulow-Klemperer result and competition complexity.}
The seminal work of \citet{BK-96} shows that it is possible to obtain revenue that is as large as the Bayesian optimal revenue, by the {\VCGAuction} with some additional buyers.
Following their approach, we also use the {\VCGAuction} as our canonical prior-independent mechanism for this purpose.

\begin{theorem}[\citealp{BK-96}]
\label{thm:prelim:bk result}
    Consider the setting of selling $m$ units of an identical item to $n\geq m$ buyers with values drawn i.i.d.\ from a regular distribution $\buyerdist$. The {\VCGAuction} with $m$ additional buyers has revenue at least as large as 
the original {\BayesianOptimalMech}, i.e., 
    $\RevVCG_{m:n+m}(\buyerdist) \geq \RevOPT_{m:n}(\buyerdist)$. Moreover, there exists an instance with a regular distribution such that $m$ additional buyers are necessary for this to hold.
\end{theorem}

Following the result of \citet{BK-96}, a body of literature has extended these ideas across a wide range of model settings in revenue-maximization. The notion of competition complexity was formally introduced in \citet{EFFTW-17}.
\begin{definition}[Competition complexity]
\label{def:competition complexity}
For any given $m,n\in\naturals$ ($m\leq n$), any $\ccapproxratio\in(0, 1]$, and any distribution class $\DistClass$, the \emph{competition complexity} $\ccomplexity\left(m,n,\DistClass,\ccapproxratio\right)$ of the {\VCGAuction} in an $m$-unit $n$-buyer market over the distribution class $\DistClass$ with targeted approximation ratio $\ccapproxratio$ is defined to be
\begin{align*}    
    \ccomplexity\left(m,n,\DistClass,\ccapproxratio\right) &\triangleq 
    \min\left\{k\in\naturals:
    \forall\buyercdf\in\DistClass,~
    \RevVCG_{m:n+k}\left(\buyerdist\right) \geq 
    \ccapproxratio\cdot \RevOPT_{m:n}\left(\buyerdist\right)
    \right\},
\end{align*}
If no finite $k$ satisfies the condition, the competition complexity is defined to be $\infty$.  
For the special case of a balanced market ($m = n$) with full optimality ($\ccapproxratio = 1$), we simplify the notation to $\ccomplexity(n,\DistClass) \triangleq \ccomplexity(n,n,\DistClass,1)$.
\end{definition}
In above definition, given any \emph{target approximation ratio} $\ccapproxratio\in(0, 1]$,  the competition complexity $\ccomplexity\left(m,n,\DistClass,\ccapproxratio\right)$ is the minimum number of buyers that suffices to add to the {\VCGAuction} in order to achieve revenue that is at least $\ccapproxratio$-fraction of revenue of the {\BayesianOptimalMech} in the original $m$-unit $n$-buyer market,
where buyers' values are drawn i.i.d.\ from any (adversarially chosen and unknown) distribution from the class of distributions $\DistClass$.  
As a sanity check, \Cref{thm:prelim:bk result} from \citet{BK-96} can be stated as 
\begin{align*}
\ccomplexity\left(m,n,\RegDists,1\right) = m
\end{align*}
In subsequent sections, we go beyond this setting by analyzing the competition complexity under stronger distributional assumptions (e.g., MHR and $\Reglevel$-regular distributions)\footnote{Without the regularity assumption, the competition complexity is known to be unbounded \citep[see, e.g.,][]{FJ-24}.} and/or weaker approximation targets (i.e., $\ccapproxratio < 1$), which are not addressed in \citet{BK-96}.

\xhdr{Revenue curves.} We follow the revenue-curve viewpoint of \citet{BR-89} to describe the expected revenue obtained by posting a take-it-or-leave-it price when buyers' values are drawn from valuation distribution $\buyerdist$.

\begin{definition}[Revenue curve, \citealp{BR-89}]
    Fix a valuation distribution $\buyerdist$. The \emph{revenue curve} $\revcurve:[0, 1]\rightarrow\reals_+$ is defined as
    \begin{align*}
        \revcurve(\quant)\triangleq \quant\cdot \buyercdf^{-1}(1-\quant)
        \quad 
        \text{for all $\quant\in(0, 1]$}~,
        \;\;
        \mbox{and}
        \;\;
        \revcurve(0)\triangleq \lim\nolimits_{\quant\to 0^+} \revcurve(\quant)
    \end{align*}
    where $\buyercdf^{-1}(1 - \quant)  \triangleq \sup\{\val:\buyercdf(\val) \leq 1 - \quant\}$.
\end{definition}
We define the \emph{monopoly revenue} as the largest revenue $\max_{\quant\in[0, 1]}\revcurve(\quant)$ and denote every quantile (resp.\ price) attaining this largest revenue as a \emph{monopoly quantile} $\optquant$ (resp.\ \emph{monopoly reserve $\optreserve$}).\footnote{When the distribution $\buyerdist$ is regular, its revenue curve is bounded and continuous, so the monopoly revenue and monopoly quantile are well-defined. When $\buyerdist$ is $\Reglevel$-regular with $\Reglevel \in [0,1)$, the monopoly quantile is not only well-defined but also unique and strictly positive \citep{SS-19}, which implies that the monopoly reserve is uniquely defined as well.
For regular distributions, it is possible that the monopoly quantile is zero. In this case, we adopt the convention that the corresponding monopoly reserve is $\infty$, and the monopoly revenue is given by $\revcurve(0)$.}
The regularity of a valuation distribution has the following nice geometric interpretation on the induced revenue curve. 
\begin{lemma}[\citealp{BR-89}]
\label{lem:concave revenue curve}
    A valuation distribution $\buyerdist$ of the buyer is regular if and only if the induced revenue curve $\revcurve$ is weakly concave.
\end{lemma} 
\noindent

\section{Balanced Markets with an MHR Distribution}
\label{sec:bk mhr}

As we illustrate in \Cref{example:intro:multi-unit bk:regular}, for every $n$, there exists an $n$-unit $n$-buyer (balanced) market with values sampled i.i.d.\ from some regular distribution, for which $n$ additional buyers are necessary for the {\VCGAuction} to outperform the {\BayesianOptimalMech}, i.e.,
a number equal to the market size (and, by \Cref{thm:prelim:bk result}, $n$ additional buyers are also sufficient).
Notably, while the distribution in \Cref{example:intro:multi-unit bk:regular} satisfies regularity, it does \emph{not} satisfy the MHR condition. This raises a natural question: \emph{when the valuation distribution is MHR rather than merely regular, is it sufficient to add significantly fewer than $n$ buyers?}

We answer this question affirmatively: when $n$ is large enough, it suffices to add fewer than half as many additional buyers. Specifically, in \Cref{sec:bk vcg mhr:worst case}, we first show that, rather than optimizing over the entire (infinite-dimensional) space of MHR distributions, the tight competition complexity can be determined by restricting attention to a specific subclass of worst-case MHR distributions that admits a simple closed-form characterization (\Cref{prop:bk vcg:mhr worst case}). 

This worst-case reduction serves as a powerful tool for subsequent analyses. We leverage it to numerically compute the (nearly) exact competition complexity for any given finite market size $n$. Moreover, we use the same characterization to analytically derive a uniform upper bound on the growth rate of the competition complexity as a function of $n$, as stated in \Cref{thm:bk vcg:mhr finite market} (\Cref{sec:bk vcg mhr:finite market}). We show that the growth rate is asymptotically at most $(e^{1/e} - 1) \cdot n \approx 0.4447\cdot n$, which is less than half the rate that is sufficient and necessary for regular distributions. We also show that this rate is asymptotically tight in the same theorem.

\subsection{Characterization of a Worst-Case Sub-class of Distributions}
\label{sec:bk vcg mhr:worst case}

As in much of the literature, a common approach to studying revenue approximation is to identify either a worst-case distribution or a succinct subclass of distributions that contains the worst case. Our analysis follows this high-level agenda.

In a balanced market ($m = n$), the {\BayesianOptimalMech} (see \Cref{lem:Bayesian optimal mechanism}) simplifies to posting the monopoly reserve to all $n$ buyers and allocating one unit to each buyer whose value exceeds this reserve. Consequently, the optimal revenue is simply $n$ times the monopoly revenue.
Combining this with the revenue monotonicity of the {\VCGAuction} (i.e., $\RevVCG_{n:n}(\buyerdist_1) \leq \RevVCG_{n:n}(\buyerdist_2)$ whenever $\buyerdist_1$ is first-order stochastically dominated by $\buyerdist_2$), we observe that a distribution can be worst-case for the competition complexity $\ccomplexity(n,\MHRDists)$ only if it does not strictly first-order stochastically dominates any other MHR distribution with the same monopoly revenue. (Otherwise, such a dominated distribution would yield the same optimal revenue but lower VCG revenue, making it worse for competition complexity.)
After normalizing the monopoly revenue to $1$ (which is without loss of generality for competition complexity analysis), we show that, for every market size $n$, the worst-case distribution belongs to the following class of truncated exponential distributions, a natural subclass of the MHR distributions.

\begin{definition}[Truncated exponential distribution]
\label{def:truncated exponential distribution}
    For any $\reserve\in[1, e]$, the \emph{$\reserve$-truncated exponential distribution $\distExpReserve$} has support $\supp(\distExpReserve) = [0,\reserve]$, and its CDF $\distExpReserve:\reals_+\rightarrow [0, 1]$ is defined as 
\begin{align*}
        \distExpReserve(\val) = \left\{
        \begin{array}{ll}
         1- \reserve^{-\frac{\val}{\reserve}}    & \val\in[0, \reserve] \\
         1    & \val \in(\reserve, \infty)
        \end{array}
        \right.
    \end{align*}
    The \emph{class of Truncated Exponential Distributions $\expDists$} is defined as  
    \begin{align*}
        \expDists\triangleq \left\{\distExpReserve:\reserve\in[1,e]\right\}.
    \end{align*}
\end{definition}

It can be verified (by its definition) that for every $\reserve\in[1,e]$, the $\reserve$-truncated exponential distribution is MHR, and thus $\expDists\subsetneq \MHRDists$.
The fact that the competition complexity of the class of MHR distributions is the same at the one for the subclass of truncated exponential distributions $\expDists$ is formalized by the following result.
As it is a special case of the more general result presented in \Cref{prop:bk vcg:worst case} (covering imbalanced markets, arbitrary target approximation ratios $\ccapproxratio$, and $\Reglevel$-regular distributions), 
we omit its proof (we prove the more general result of \Cref{prop:bk vcg:worst case} in \Cref{sec:bk vcg}).

\begin{proposition}[Characterization of a subclass of worst-case MHR distributions]
\label{prop:bk vcg:mhr worst case}
    The competition complexity of the {\VCGAuction} in the $n$-unit, $n$-buyer balanced market over the class of MHR distributions $\MHRDists$ coincides with competition complexity
over the subclass of truncated exponential distributions $\expDists\subsetneq \MHRDists$; that is,
    \begin{align*}
        \ccomplexity\left(n,\MHRDists\right) = \ccomplexity\left(n,\expDists\right).
    \end{align*}
\end{proposition}

\begin{remark}
\label{remark:worst truncation}
    One might wonder whether an even stronger characterization of the worst-case is possible---for example, whether the worst-case distribution always occurs at a boundary of the allowable truncation range (i.e., $\reserve = 1$ or $\reserve = e$) or whether the same truncation value $\reserve$ is worst for all market sizes~$n$. 
    We show that neither is true: the worst-case truncation $\reserve$ generally depends on the market size $n$. See further discussion in \Cref{apx:worst case truncation discussion}.
\end{remark}

Although the definition of competition complexity quantifies over all MHR distributions, the proposition above guarantees that it suffices to consider a much smaller subclass. Specifically, the class of truncated exponential distributions consists of distributions that admit a simple closed-form expression parameterized by a single scalar $\reserve \in [1, e]$.

Leveraging our characterization we can numerically compute the competition complexity $\ccomplexity(n, \MHRDists)$ for any given market size $n$ as follows: we pick some large $M$ and perform a grid search over a discretization of the truncation $\reserve \in \{ 1,\; 1 + {(e-1)}/{M},\; 1 + 2{(e-1)}/{M}, \dots, e \}$,
with discretization step size $(e-1)/M$. For each grid point, we numerically compute the number of additional buyers that is sufficient and necessary, and then bound the discretization error (via a Lipschitzness argument) to account for all intermediate values of $\reserve$. Using this approach, we determine the exact competition complexity for all $n \leq 593$ and report the results in  \Cref{sec:numerical experiment}.

In the next subsection, we further exploit our characterization to derive a uniform upper and lower bound on the competition complexity that holds for all market sizes $n \in\naturals$.

\subsection{Uniform Upper and Lower Bound on the Competition Complexity}
\label{sec:bk vcg mhr:finite market}

In this section, we provide a uniform and lower upper bound on the competition complexity for MHR distributions in the balanced-market setting ($m=n$). Our analysis holds for every finite market size $n \in \naturals$ and reveals a contrast with the regular case.

\begin{restatable}[Balanced market with MHR distributions]{theorem}{thmMHRLUB}
\label{thm:bk vcg:mhr finite market}
    For any $n\in\naturals$, the competition complexity 
    of the {\VCGAuction} for an $n$-unit $n$-buyer (balanced) market over the class of MHR distributions $\MHRDists$ 
    satisfies 
      \begin{align*}
        \left(e^{\frac{1}{e}} - 1\right) \cdot n
        \leq \ccomplexity\left(n,\MHRDists\right)
        \leq 
        \left\lceil\left(e^{\frac{1}{e}} - 1\right) \cdot n + 1.05 \ln n\right\rceil
\end{align*}
\end{restatable}

As the leading term of $(e^{1/e} - 1) \cdot n$ is the same for both the lower and upper bounds, 
our result provides, up to low-order terms, a tight finite-market closed-form expression for the competition complexity for balanced market with MHR distributions.
In stark contrast, the competition complexity under regular distributions is exactly $n$ (i.e., more than twice as large), highlighting the substantial benefit of the MHR assumption.

To prove the lower bound in \Cref{thm:bk vcg:mhr finite market}, we directly consider the $e$-truncated exponential distribution and establish an upper bound the revenue in the {\VCGAuction} (which gives us the lower bound on the competition complexity).

Unlike many revenue approximation results for MHR distributions, which admit relatively straightforward analyses, our competition complexity upper bound analysis is considerably rather complex. We conclude this section by outlining the proof idea for this upper bound. The formal proof of \Cref{thm:bk vcg:large market} is deferred to \Cref{apx:mhr finite market proof}.

\xhdr{Proof idea behind the upper bound in \Cref{thm:bk vcg:mhr finite market}.}
Leveraging the worst-case characterization in \Cref{prop:bk vcg:mhr worst case}, it suffices to upper bound $\ccomplexity(n,\expDists)$---the competition complexity over the subclass $\expDists \subsetneq \MHRDists$ of truncated exponential distributions---to obtain an upper bound on $\ccomplexity(n,\MHRDists)$, the competition complexity over all MHR distributions. Recall that $\expDists$ consists of the $\reserve$-truncated exponential distributions for $\reserve \in [1, e]$ (\Cref{def:truncated exponential distribution}).

Let $\addBuyers$ denote the upper bound on competition complexity stated in \Cref{thm:bk vcg:mhr finite market}.
With a slight abuse of notation, let $\RevVCG_{n:n+\addBuyers}(\reserve)$ denote the expected revenue of the {\VCGAuction} in an $n$-unit, $(n+\addBuyers)$-buyer market with $\reserve$-truncated exponential distribution. 
Since the original market is balanced, the revenue of the {\BayesianOptimalMech} simplifies to $\RevOPT_{n:n} = n$ for all $\reserve \in [1, e]$. Consequently, it suffices to analyze $\RevVCG_{n:n+\addBuyers}(\reserve)$ as a function of $\reserve$ and prove that
\begin{align*}
\RevVCG_{n:n+\addBuyers}(\reserve) \geq n \quad \text{for all } \reserve \in [1, e].
\end{align*}
As noted in \Cref{remark:worst truncation}, the main challenge lies in identifying the worst-case truncation parameter $\reserve \in [1, e]$ that minimizes $\RevVCG_{n:n+\addBuyers}(\reserve)$. This minimizer admits no closed-form expression and varies with the market size $n$.

To overcome this, our proof proceeds via a careful case analysis. For market sizes $n \leq 593$, we compute the exact competition complexity numerically (reported in \Cref{sec:numerical experiment}). It remains to handle the case $n \geq 594$. For such sufficiently large $n$, we apply concentration inequalities and partition the interval $[1, e]$ into four subintervals, analyzing each separately. 
Then:
\begin{enumerate}[leftmargin=*]
    \item[(i)] For any $\reserve \in [1, 5/4]$, we show that the derivative of $\RevVCG_{n:n+\addBuyers}(\reserve)$ by $\reserve$ is positive, and thus $\RevVCG_{n:n+\addBuyers}(\reserve)$ is increasing in $\reserve$ over this interval. Consequently, its minimum occurs at $\reserve = 1$, and we verify that $\RevVCG_{n:n+\addBuyers}(1) = n$ (\Cref{lem:VCG revenue small reserve}).
    \item[(ii)] For any $\reserve \in [5/4, 2]$, the analysis is relatively straightforward: the additive gap $\RevVCG_{n:n+\addBuyers}(\reserve) - n$ is lower bounded by an absolute positive constant. To formalize this, we construct a simple closed-form lower bound on $\RevVCG_{n:n+\addBuyers}(\reserve)$ and show that it exceeds $n$ in \Cref{lem:VCG revenue middle reserve}.

    \item[(iii)] For any $\reserve \in [2, e/(1+1/n)]$, we prove that the derivative of $\RevVCG_{n:n+\addBuyers}(\reserve)$ by $\reserve$ is negative, implying that $\RevVCG_{n:n+\addBuyers}(\reserve)$ is decreasing in $\reserve$ over this interval (\Cref{lem:VCG revenue high reserve}). Then it suffices to verify that $\RevVCG_{n:n+\addBuyers}(e/(1+1/n)) \geq n$,
which is handled in Case (iv) below.

    \item[(iv)] For any $\reserve \in [e/(1+1/n), e]$, we first upper bound the absolute value of the derivative of $\RevVCG_{n:n+\addBuyers}(\reserve)$. Since the length of the interval is bounded, this yields the inequality \\ $\RevVCG_{n:n+\addBuyers}(\reserve) \geq \RevVCG_{n:n+\addBuyers}(e) - 8$ (\Cref{lem:VCG revenue super high reserve}). We then show that $\RevVCG_{n:n+\addBuyers}(e) \geq n + 8$ in \Cref{lem:VCG revenue e reserve}.\footnote{While our theoretical analysis does not formally establish this (as it is unnecessary for our result), numerical experiments suggest that the minimum of $\RevVCG_{n:n+\addBuyers}(\reserve)$---and the corresponding minimizer $\reserve$---indeed lies in the interval $[e/(1+1/n), e]$.}
\end{enumerate}

\section{General Markets with a \texorpdfstring{$\Reglevel$-Regular}{Lambda-Regular} Distribution}
\label{sec:bk vcg}

In this section, we generalize our analysis in \Cref{sec:bk mhr} along three natural dimensions. 

First, we move beyond the MHR assumption to the broader class of $\Reglevel$-regular distributions, which smoothly interpolates between regularity ($\Reglevel = 1$) and MHR ($\Reglevel = 0$). This allows us to quantify precisely how stronger distributional assumptions reduce the number of additional buyers.

Second, we consider general (possibly imbalanced) markets with $n$ buyers and $m \leq n$ units, rather than restricting to the balanced case ($m = n$). This captures a wider range of practical settings---from scarce goods ($m \ll n$) to abundant supply ($m \approx n$).

Third, instead of requiring the {\VCGAuction} to fully match the optimal Bayesian revenue (obtaining a $\ccapproxratio$ approximation with $\ccapproxratio = 1$), we relax the goal to achieving a 
$\ccapproxratio$-fraction of it, i.e., for $\eps>0$, getting a fraction of only $\ccapproxratio=1-\eps< 1$. 
This shift is not only theoretically natural but also practically relevant: if a seller can recruit only a limited number of extra buyers---say, $0.2n$ instead of the $0.4447n$ (that is necessary for full optimality under MHR)---what revenue guarantee can they still expect? Our framework provides a complete answer: by inverting our competition complexity bounds, one can determine the best approximation ratio $\ccapproxratio$ achievable for any given number of additional buyers.

Together, these extensions yield a unified and fine-grained understanding of how market size, distributional strength, and revenue targets jointly shape the cost of prior independence, in the sense of the number of additional buyers that is sufficient and necessary for the {\VCGAuction} to beat the optimum in revenue.

Similar to \Cref{sec:bk mhr}, we first show that, rather than searching over the entire (infinite dimensional) space of $\Reglevel$-regular distributions, the tight competition complexity can be determined by considering a specific subclass of $\Reglevel$-regular distributions that contains the worst-case distribution and admits a simple closed-form characterization (\Cref{prop:bk vcg:worst case}). This ``worst-case'' characterization then serves as a powerful tool for subsequent analyses: we apply it to analytically derive the tight bound on the competition complexity in (possibly imbalanced) large markets, for any $\Reglevel\in [0,1]$ and any $\ccapproxratio\in (0,1)$ (\Cref{thm:bk vcg:large market}) in \Cref{sec:bk vcg:large market}.

\subsection{Characterization of a Worst-Case Sub-class of Distributions}
\label{sec:bk vcg:worst case}

In this section, we identify a subclass of $\Reglevel$-regular distributions that is worst-case for the competition complexity.

\begin{definition}[Truncated generalized Pareto distribution]
\label{def:truncated generalized Pareto distribution}
    For any $\Reglevel\in(0, 1]$ and $\reserve\in[1,(1-\Reglevel)^{-1/\Reglevel}]$, the \emph{$\reserve$-truncated $\Reglevel$-generalized Pareto distribution~$\distLambdaReserve$} has support $\supp(\distLambdaReserve) = [0,\reserve]$ and CDF $\distLambdaReserve:\reals_+\rightarrow [0, 1]$ 
    defined as follows\footnote{The classic generalized Pareto distribution has CDF $\buyerdist(\val) = 1 - (1+\Reglevel\val/\sigma)^{-1/\Reglevel}$, where $\sigma$ and $\Reglevel$ denote the scale and shape parameters, respectively. Our $\reserve$-truncated $\Reglevel$-generalized Pareto distribution~$\distLambdaReserve$ is obtained by setting the scale to $\Reglevel\reserve/(\reserve^{\Reglevel}- 1)$ and the shape to $\Reglevel$, and then collapsing all probability mass (of $1/\reserve$ in total) above $\reserve$ into an atom at the truncation point~$\reserve$.}
    \begin{align*}
        \distLambdaReserve(\val) = \left\{
        \begin{array}{ll}
         1-\left(\frac{\reserve^\Reglevel - 1}{\reserve}\cdot \val + 1\right)^{-\frac{1}{\Reglevel}}    & \val\in[0, \reserve] \\
         1    & \val \in(\reserve, \infty)
        \end{array}
        \right.
    \end{align*}
For any $\Reglevel \in(0, 1]$, the \emph{class of truncated $\Reglevel$-generalized Pareto distributions $\truncGPDists$} is defined as
\begin{align*}
        \truncGPDists\triangleq \left\{\distLambdaReserve:\reserve\in[1,(1-\Reglevel)^{-1/\Reglevel}]\right\}.
    \end{align*}
    For $\Reglevel = 0$, we define $\mathcal{F}_{(0)}^*\triangleq \expDists$ where $\expDists$ is the class of truncated exponential distributions (\Cref{def:truncated exponential distribution}).\footnote{Recall that for $\Reglevel = 0$ (i.e., the MHR case), $\distExpReserve$ with $\reserve \in [1, e]$ denotes the $\reserve$-truncated exponential distribution defined in \Cref{def:truncated exponential distribution}. Throughout the paper, any expression involving $\Reglevel = 0$ is interpreted as the limit of its counterpart for $\Reglevel \in (0, 1]$. For example,
$\lim_{\Reglevel \to 0} (1 - \Reglevel)^{-1/\Reglevel} = e$ and $\lim_{\Reglevel \to 0}(({\reserve^\Reglevel - 1})/{\reserve}\cdot \val + 1)^{-{1}/{\Reglevel}} = \reserve^{-\val/\reserve}$.} 
\end{definition}

It can be verified (by its definition) that for every $\Reglevel\in[0, 1]$ and $\reserve\in[1,e]$,  the $\reserve$-truncated $\Reglevel$-generalized Pareto distribution~$\distLambdaReserve$ is $\Reglevel$-regular, and thus $\truncGPDists\subsetneq \lambdaDists$.

The following proposition establishes that for every $\Reglevel\in[0, 1]$, $m\leq n$, and $\ccapproxratio\in(0, 1]$, the distribution that defines $\ccomplexity\left(m,n,\lambdaDists,\ccapproxratio\right)$, the competition complexity of all $\Reglevel$-regular distributions,
belongs to the class of truncated $\Reglevel$-generalized Pareto distributions $\truncGPDists$.
Further discussion of the intuition why this distribution class contains the worst case and the implications of this worst-case characterization, follows immediately afterward.

\begin{restatable}[Characterization of a subclass of worst-case $\Reglevel$-regular distributions]{proposition}{propVCGWorstDists}
\label{prop:bk vcg:worst case}
    For any $\Reglevel\in[0, 1]$, $m, n \in \naturals$ with $m\leq n$, and $\ccapproxratio\in(0, 1]$, the competition complexity of the {\VCGAuction} in the $m$-unit, $n$-buyer market over the class of $\Reglevel$-regular distributions $\lambdaDists$ coincides with that over the subclass of truncated $\Reglevel$-generalized Pareto distributions $\truncGPDists\subsetneq \lambdaDists$; that is,
    \begin{align*}
        \ccomplexity\left(m,n,\lambdaDists,\ccapproxratio\right) = \ccomplexity\left(m,n,\truncGPDists,\ccapproxratio\right).
    \end{align*}
\end{restatable}
Similar to the worst-case characterization illustrated in \Cref{sec:bk mhr}, 
although the definition of competition complexity quantifies over all $\Reglevel$-regular distributions, the proposition above guarantees that it suffices to consider a much smaller subclass. Specifically, the class of truncated $\Reglevel$-generalized Pareto distributions $\truncGPDists$ consists of distributions parameterized by a single scalar $\reserve$, which ranges over the bounded interval $[1, (1-\Reglevel)^{-1/\Reglevel}]$ when $0< \Reglevel < 1$ (and interval $[1,e]$ when $\Reglevel = 0$).
This characterization dramatically simplifies both theoretical and numerical analysis of the competition complexity, as we leverage in the subsequent sections.

We now provide some intuition behind this worst-case distribution class and sketch the key ideas and challenges in 
the proof of \Cref{prop:bk vcg:worst case}. (The formal proof can be found in \Cref{apx:propVCGWorstDists}.)
By construction, for every $\reserve \in [1, e]$, the $\reserve$-truncated $\Reglevel$-generalized Pareto distribution~$\distLambdaReserve$ has monopoly reserve $\optreserve = \reserve$, monopoly quantile $\optquant = 1/\reserve$, and consequently monopoly revenue equal to 1. Moreover, for any $\Reglevel$-regular distribution $\buyerdist$ with monopoly reserve $\reserve$ and monopoly quantile $1/\reserve$,\footnote{It is known that, for any $\Reglevel$-regular distribution with monopoly revenue normalized to one, the monopoly reserve must lie in the interval $[1, (1-\Reglevel)^{-1/\Reglevel}]$ \citep{SS-19}.} the $\reserve$-truncated $\Reglevel$-generalized Pareto distribution~$\distLambdaReserve$ is first-order stochastically dominated (FOSD) by $\buyerdist$. In other words, after normalizing monopoly revenue to $1$ (which is without loss of generality in competition complexity analysis), the class $\truncGPDists$ can be viewed as the smallest distributions---in the FOSD sense---among all $\Reglevel$-regular distributions.

Since the $\reserve$-truncated $\Reglevel$-generalized Pareto distribution~$\distLambdaReserve$ is first-order stochastically dominated by every other $\Reglevel$-regular distribution with the same monopoly reserve and monopoly quantile, a natural approach to establishing its worst-case nature would be to show monotonicity of the \emph{multiplicative revenue gap} in the FOSD order.
Specifically, one might hope that the ratio
\begin{align*}
    \frac{\RevVCG_{m:n+k}(\buyerdist)}{\RevOPT_{m:n}(\buyerdist)}
\end{align*}
is non-decreasing as $\buyerdist$ becomes larger in the FOSD sense---while keeping the monopoly reserve and quantile fixed.
Unfortunately, this monotonicity does not hold in general when $m < n$. (It does hold in the special case of a balanced market, where $n = m$.)
Indeed, one can construct $\Reglevel$-regular distributions $\buyerdist_1$ and $\buyerdist_2$ such that $\buyerdist_1$ is FOSD-dominated by $\buyerdist_2$, yet the multiplicative revenue gap is larger under $\buyerdist_1$ (see such an example in \Cref{sec:FOSD monotonicity discussion})
The underlying reason is that, although both the {\VCGAuction} and the {\BayesianOptimalMech} share similar structural forms, the addition of $k$ extra buyers alters the distribution of the $(m+1)$-th highest value, which affects the revenues of both mechanisms in non-monotonic ways.

To overcome this obstacle, we shift our focus from the multiplicative revenue gap to the \emph{additive revenue gap}:
\begin{align*}
    \RevVCG_{m:n+k}(\buyerdist) - \ccapproxratio \cdot \RevOPT_{m:n}(\buyerdist).
\end{align*}
This shift is not merely a technical convenience---it is \emph{without loss of generality} for the purpose of identifying worst-case distributions. Recall that the competition complexity is defined as the smallest $k$ such that this additive gap is non-negative for \emph{all} distributions in the class. Thus, any distribution that minimizes the additive gap is also a candidate for maximizing the required $k$, making it sufficient to analyze this additive form.

This change of perspective reveals a crucial \emph{three-interval structure} in the behavior of the additive gap. Specifically, there exist quantile thresholds $\quant\primed \leq \quant\doubleprimed$ in $[0,1]$ such that the following holds. Consider any two $\Reglevel$-regular distributions $\buyerdist_1$ and $\buyerdist_2$ that share the same monopoly reserve and monopoly quantile, with $\buyerdist_1$ first-order stochastically dominated by $\buyerdist_2$. Suppose further that their CDFs differ only on values corresponding to a single quantile interval. Then:
\begin{itemize}
    \item If the CDFs differ only on $[0, \quant\primed]$, the additive revenue gap is \emph{smaller} under $\buyerdist_1$;
    \item If they differ only on $[\quant\primed, \quant\doubleprimed]$, the gap is \emph{larger} under $\buyerdist_1$;
    \item If they differ only on $[\quant\doubleprimed, 1]$, the gap is again \emph{smaller} under $\buyerdist_1$.
\end{itemize}

In other words, the failure of monotonicity is confined exclusively to the middle interval $[\quant\primed, \quant\doubleprimed]$---and even there, it behaves in a consistent and analyzable direction: the additive gap \emph{increases} under the FOSD-dominated distribution. (In the special case of a balanced market ($n = m$), the middle interval degenerates, i.e., $\quant\primed = \quant\doubleprimed$, and full monotonicity is restored.)

Leveraging the concavity of revenue curves under $\Reglevel$-regularity, we show that this structure allows us to ``compensate'' for the non-monotonicity in the middle interval using the monotonicity in the lower interval $[0, \quant\primed]$. This enables a precise comparison between an arbitrary $\Reglevel$-regular distribution and its truncated generalized Pareto counterpart, ultimately proving that the worst case always lies within our parametric subclass.

\subsection{Tight Competition Complexity for Asymptotically 
Large Markets}
\label{sec:bk vcg:large market}

In the previous section, we characterize the worst-case distribution for the competition complexity in \Cref{prop:bk vcg:worst case}. Although this characterization---combined with computer-aided numerical search (see \Cref{sec:numerical experiment} for an example)---is sufficient to determine the competition complexity for markets with a fixed, constant number of buyers, it does not yield a closed-form expression for the competition complexity as a function of the market parameters: the number of buyers $n$, the number of units $m$, the regularity level $\Reglevel$ in $\Reglevel$-regular distributions, and the target approximation ratio $\ccapproxratio$.

In this section, we shed light on how these parameters influence the competition complexity of the {\VCGAuction}. To ensure a tractable analysis and obtain clean asymptotic results, we consider the following large-market regime: let $\imbalanceratio \in [0, 1]$ denote the \emph{supply-to-demand ratio}, i.e., the asymptotic fraction of units per buyer. We characterize the competition complexity in the limit as the number of buyers $n \to \infty$, with the number of units set to $m = \lceil \imbalanceratio \cdot n \rceil$. Additionally, we assume that the target approximation ratio $\ccapproxratio$ is an absolute constant strictly less than one.\footnote{Focusing on $\ccapproxratio < 1$ and taking $n \to \infty$ significantly simplifies the analysis. Recall that the non-asymptotic result for the balanced market with MHR distributions and $\ccapproxratio = 1$ is given in \Cref{sec:bk mhr}; both the result and its analysis are substantially more involved than their asymptotic counterparts.}

Specifically, with a slight abuse of notation, we define the \emph{asymptotic competition complexity} with supply-to-demand ratio $\imbalanceratio$, over the class of $\Reglevel$-regular distributions $\lambdaDists$ and target approximation ratio $\ccapproxratio$, as
\begin{align*}
    \ccomplexityInfty(\imbalanceratio,\lambdaDists,\ccapproxratio)
    \triangleq 
    \lim_{n\rightarrow\infty} ~
        \frac{\ccomplexity\left(\lceil \imbalanceratio \cdot n \rceil,n,\lambdaDists,\ccapproxratio\right)}{n}.
\end{align*}
We remark that this definition normalizes the competition complexity by the number of buyers $n$. By definition, $\ccomplexityInfty(\imbalanceratio,\lambdaDists,\ccapproxratio) = 0$ implies that the sufficient and necessary number of additional buyers is sublinear in $n$, whereas a strictly positive value $\ccomplexityInfty(\imbalanceratio,\lambdaDists,\ccapproxratio) > 0$ implies that the leading term in the number of additional buyers is $\ccomplexityInfty(\imbalanceratio,\lambdaDists,\ccapproxratio) \cdot n$. 

As a sanity check, \Cref{thm:bk vcg:mhr finite market} implies that 
$
\ccomplexityInfty(1,\MHRDists,1) \leq e^{1/e} - 1$
for the balanced market (i.e., $\imbalanceratio = 1$) under MHR distributions. Below, we provide a comprehensive closed-form characterization of this asymptotic quantity. Its formal proof can be found in \Cref{apx:thmVCGLargeMarket}.

\begin{restatable}[Competition complexity in large market]{theorem}{thmVCGLargeMarket}
\label{thm:bk vcg:large market}
    For any $\Reglevel\in[0, 1]$, $\imbalanceratio\in(0, 1]$, and $\ccapproxratio\in(0, 1)$, 
the asymptotic competition complexity $\ccomplexityInfty(\imbalanceratio,\lambdaDists,\ccapproxratio)$ of the {\VCGAuction} with supply-to-demand ratio $\imbalanceratio$, over the class of $\Reglevel$-regular distributions $\lambdaDists$ and target approximation ratio $\ccapproxratio$ satisfies\footnote{We define operator $\plus{\cdot}\triangleq \max\{\cdot,0\}$.}
    \begin{align*}
        &\ccomplexityInfty(\imbalanceratio,\lambdaDists,\ccapproxratio)
        =
        \plus{
        \imbalanceratio {\left(\frac{\Reglevel\cdot\ccapproxratio}{\imbalanceratio}{\left(1-\Reglevel\right)}^{\frac{1-\Reglevel}{\Reglevel}}+1 \right)}^{\frac{1}{\Reglevel}}-1
        }
        \cdot \indicator{\imbalanceratio\geq (1-\Reglevel)^{\frac{1}{\Reglevel}}}
    \intertext{In the special cases of MHR distributions ($\Reglevel = 0$) and regular distributions ($\Reglevel = 1$), 
the asymptotic competition complexity simplifies}
        &
        \ccomplexityInfty(\imbalanceratio,\MHRDists,\ccapproxratio)
        =
        \plus{
        \imbalanceratio\cdot e^{\frac{\ccapproxratio}{e\imbalanceratio}} - 1}
        \cdot \indicator{\imbalanceratio\geq \frac{1}{e}}
        \\
        &
        \ccomplexityInfty(\imbalanceratio,\RegDists,\ccapproxratio)
        =
        \plus{
        \ccapproxratio + \imbalanceratio - 1 
        }
    \end{align*}
\end{restatable}

We present several interpretations of the asymptotic competition complexity in \Cref{thm:bk vcg:large market}.

First, consider the goal of \emph{almost} outperforming the {\BayesianOptimalMech} in the original market, i.e., $\ccapproxratio = 1 - \eps$ for some fixed $\varepsilon > 0$. In this case, the sufficient and necessary number of additional buyers by the {\VCGAuction} scales \emph{linearly} with the original number of buyers $n$ when the market is relatively balanced---specifically, when the supply-to-demand ratio satisfies $\imbalanceratio > (1 - \Reglevel)^{1/\Reglevel}$. Conversely, when the market is more imbalanced ($\imbalanceratio \leq (1 - \Reglevel)^{1/\Reglevel}$), the sufficient and necessary number of additional buyers grows only \emph{sublinearly} in $n$.

Second, still under the almost-optimality goal ($\ccapproxratio = 1 - \eps$), we quantify the asymptotic overhead across different distribution classes. For MHR distributions ($\Reglevel = 0$), the fraction of sufficient and necessary additional buyers is $e^{(1-\eps)/e} - 1 \approx 44.47\%$ when $\imbalanceratio = 1$ (a value that matches the leading term in the finite-market bound of \Cref{thm:bk vcg:mhr finite market}), and $\tfrac{1}{2} e^{2(1-\eps)/e} - 1 \approx 4.35\%$ when $\imbalanceratio = 0.5$. For $0.5$-regular distributions, the corresponding fractions are $((5-\eps)/4)^2 - 1 = 56.25\%$ and $((3-\eps)/2)^2 / 2 - 1 = 12.5\%$, respectively. In stark contrast, for regular distributions ($\Reglevel = 1$), the overhead is substantially higher: $1-\eps\approx 100\%$ when $\imbalanceratio = 1$ and $(1-\eps)/2 \approx 50\%$ when $\imbalanceratio = 0.5$. These results underscore how stronger regularity assumptions---moving from regular to MHR---dramatically reduce the competition complexity of the {\VCGAuction}. Also see \Cref{fig:bk vcg:large market} for an illustration.

Third, while \Cref{thm:bk vcg:large market} focuses on target approximation ratios $\ccapproxratio < 1$, its expression admits a continuous extension to $\ccapproxratio = 1$. Specifically, taking the limit as $\ccapproxratio \to 1$ recovers the asymptotic competition complexity for the balanced MHR setting ($\imbalanceratio = 1$, $\Reglevel = 0$) which can be computed from \Cref{thm:bk vcg:mhr finite market}. This suggests that the same functional form likely holds for $\ccapproxratio = 1$ in general settings, though establishing this rigorously would need a more intricate analysis (similar in spirit to the proof of \Cref{thm:bk vcg:mhr finite market}) which we leave for future work.

Fourth, holding any two of the three parameters fixed, the asymptotic competition complexity $\ccomplexity(\imbalanceratio,\lambdaDists,\ccapproxratio)$ is \emph{weakly increasing} and \emph{weakly convex} in each of the following: the supply-to-demand ratio $\imbalanceratio$, the regularity level $\Reglevel$, and the target approximation ratio $\ccapproxratio$.

Finally, we note that although \Cref{thm:bk vcg:large market} does not cover the case $\ccapproxratio = 1$, it nonetheless implies a lower bound on the asymptotic competition complexity for $\ccapproxratio = 1$. This follows because, by definition, the asymptotic competition complexity is non-decreasing in $\ccapproxratio$. We thus obtain the following immediate corollary.
\begin{corollary}
\label{cor:bk vcg:large market LB}
    For any $\Reglevel\in[0, 1]$, 
    and $\imbalanceratio\in(0, 1]$, 
    the asymptotic competition complexity $\ccomplexityInfty(\imbalanceratio,$
    $\lambdaDists,1)$ satisfies
    \begin{align*}
        \ccomplexityInfty(\imbalanceratio,\lambdaDists,1)
        \geq 
        \plus{
        \imbalanceratio {\left(\frac{\Reglevel}{\imbalanceratio}{\left(1-\Reglevel\right)}^{\frac{1-\Reglevel}{\Reglevel}}+1 \right)}^{\frac{1}{\Reglevel}}-1
        }
        \cdot \indicator{\imbalanceratio\geq (1-\Reglevel)^{\frac{1}{\Reglevel}}}
    \end{align*}
\end{corollary}

\subsection{Proof of \texorpdfstring{\Cref{prop:bk vcg:worst case}}{Proposition 4.1} and \texorpdfstring{\Cref{thm:bk vcg:large market}}{Theorem 4.2}}

We first introduce an effective way to express the expected revenue in the {\BayesianOptimalMech} and in the {\VCGAuction.}
For any given $m, n \in \naturals$ with $m\leq n$, let $\quant_{m:n}$ denote the $m$-th smallest order statistics {(equivalently, the $m$-th highest value)} of $n$ samples drawn from the uniform distribution $U[0, 1]$. We use $\osdensity_{m:n}$ to represent its probability density function (PDF). By definition, it satisfies
\begin{align*}
    \osdensity_{m:n}(\quant) = n\cdot \binom{n - 1}{m - 1}\cdot (1 - \quant)^{n - m}\cdot \quant^{m - 1}
\end{align*}
With this notation in place, we can express the expected revenues of the {\VCGAuction} and the {\BayesianOptimalMech} in terms of the revenue curve and uniform order statistics.

\begin{lemma}
\label{lem:prelim:revenue expression in uniform order statistic}
    In the setting of selling $m$ units of an identical item to $n$ buyers with values drawn i.i.d.\ from a regular distribution $\buyerdist$, the expected revenue of the {\VCGAuction} satisfies
    \begin{align*}
        \RevVCG_{m:n}(\buyerdist) = n\cdot \displaystyle\int_0^1 \revcurve(\quant)\cdot 
        \osdensity_{m:n-1}(\quant)\,\dd\quant
    \end{align*}
    The expected revenue of the {\BayesianOptimalMech} with any monopoly reserve $\optreserve$ and its corresponding monopoly quantile $\optquant$
satisfies 
\begin{align*}
        \RevOPT_{m:n}(\buyerdist) = n\cdot \displaystyle\int_0^1
        \revcurve(\min\{\quant,\optquant\})\cdot \osdensity_{m:n-1}(\quant)\,\dd\quant
    \end{align*}
    where $\revcurve$ is the revenue curve induced from valuation distribution $\buyerdist$.
\end{lemma}
\begin{proof}
Fix an arbitrary buyer. In the {\VCGAuction}, she essentially faces a randomized price-posting mechanism, where the price is equal to the $m$-th highest value among the other $n - 1$ buyers. By definition, the quantile of this price follows the $m$-th smallest order statistic of $n-1$ samples drawn from the uniform distribution $U[0, 1]$, with PDF $\osdensity_{m:n-1}$. When $m\ge n$, let $\osdensity_{m:n-1}(\quant)=\DiracDelta_1(\quant)$, where $\DiracDelta_1$ denotes the Dirac delta measure concentrated at $1$. 
Conditioning on this quantile $\quant_{m:n-1}$, the buyer's expected payment is $\revcurve(\quant_{m:n-1})$. Therefore, her expected payment, averaged over the randomness of the other buyers' values, is given by
    \begin{align*}
        \displaystyle\int_0^1 \revcurve(\quant)\cdot \osdensity_{m:n-1}(\quant)\dd\quant
    \end{align*}
    Since all buyers are ex ante symmetric, the expected revenue of the {\VCGAuction} is $n$ times one buyer's expected payment, which is the formula presented in the lemma statement.

    The analysis for the {\BayesianOptimalMech} is similar. Fix an arbitrary buyer; she essentially faces a price-posting mechanism, where the price is the maximum of either the monopoly reserve (i.e., the value of the monopoly quantile) or the value of the $m$-th smallest quantile of the other $n - 1$ buyers. Conditioning on the quantile $\quant_{m:n-1}$, the buyer's expected payment is $\revcurve(\min{\quant_{m:n-1}, \optquant})$. Therefore, her expected payment, averaged over the randomness of the other buyers' values, is given by 
    \begin{align*}
        \displaystyle\int_0^1 \revcurve(\quant)\cdot 
    \osdensity_{m:n-1}(\min\{\quant,\optquant\})\dd\quant
    \end{align*}
    Since all buyers are ex ante symmetric, the expected revenue of the {\BayesianOptimalMech} is $n$ times one buyer's expected payment, which is the formula presented in the lemma statement.

    This completes the proof of \Cref{lem:prelim:revenue expression in uniform order statistic}.
\end{proof}

We also recall the following equivalence condition between distributions and their induced revenue curves. 

\begin{restatable}[Revenue curve dominance]{lemma}{lem:revcurveDominance}
\label{lem:fosd}
Let $\buyerdist_1$ and $\buyerdist_2$ be two valuation distributions with corresponding revenue curves
$\revcurve_1$ and $\revcurve_2$. Then $\buyerdist_1$ is first-order stochastically dominated by $\buyerdist_2$,
i.e., $\buyerdist_1(\val)\ge \buyerdist_2(\val)$ for every value $\val\in\reals$, if and only if
$\revcurve_1(\quant)\le \revcurve_2(\quant)$ for every quantile $\quant\in[0,1]$.
\end{restatable}
\begin{proof}
Suppose $\buyerdist_1$ is first-order stochastically dominated (FOSD) by $\buyerdist_2$. We show $\revcurve_1(\quant)\le \revcurve_2(\quant)$ for every quantile $\quant\in[0,1]$. Fix any quantile $\quant\in(0, 1]$, define $\val_1 \triangleq \buyercdf_1^{-1}(1-\quant)$ and $\val_2 \triangleq \buyercdf_2^{-1}(1-\quant)$. Due to the FOSD relation, $\val_1 \leq \val_2$. Thus, $\revcurve_1(\quant) = \quant\cdot \val_1 \leq \quant\cdot\val_2 = \revcurve_2(\quant)$ for every $\quant\in(0, 1]$. Finally, $\revcurve_1(0) = \lim_{\quant\to 0}\revcurve_1(\quant) \leq \lim_{\quant\to 0}\revcurve_2(\quant) = \revcurve_2(0)$ as well.

Now we verify the opposite direction. Suppose $\revcurve_1(\quant)\le \revcurve_2(\quant)$ for every quantile $\quant\in[0,1]$. Fix any value $\val\primed\in\reals$. 
Let $\quant \triangleq 1-\buyercdf_1(\val\primed)$, and $\val\doubleprimed\triangleq \buyercdf_2^{-1}(1-\quant)$. Note that $\revcurve_1(\quant) \leq \revcurve_2(\quant)$, which implies $\quant\cdot \val\primed \leq \quant\cdot \val\doubleprimed$ and thus $\val\primed \leq \val\doubleprimed$. Hence, $\buyercdf_1(\val\primed) = 1 - \quant = \buyercdf_2(\val\doubleprimed) \geq \buyercdf_2(\val\primed)$, where the first and second equalities holds due to the construction of $\quant$ and $\val\doubleprimed$. This finishes the proof as desired.
\end{proof}

\label{apx:propVCGWorstDists}

We now prove \Cref{prop:bk vcg:worst case}.
\begin{proof}[Proof of \Cref{prop:bk vcg:worst case}]
    Define auxiliary notation $k \triangleq \ccomplexity\left(m,n,\lambdaDists, \ccapproxratio\right) - 1$.
    By the definition of the competition complexity, there exists an $\Reglevel$-regular distribution $\buyerdist\primed$ such that 
    \begin{align*}
        \RevVCG_{m:n+k}\left(\buyerdist\primed\right)
        <
        \ccapproxratio\cdot \RevOPT_{m:n}
        \left(\buyerdist\primed\right)
    \end{align*}
    Without loss of generality, we assume that the monopoly revenue in distribution $\buyerdist\primed$ is one. Invoking Proposition~4 in \citet{SS-19}, since the monopoly quantile of any $\Reglevel$-regular distribution is at least ${(1-\Reglevel)}^{{1}/{\Reglevel}}$, the monopoly reserve $\optreserve$ of distribution $\buyerdist\primed$ satisfies $\optreserve\in[1,{({1-\Reglevel})}^{-{1}/{\Reglevel}}]$. 

    Let $\revcurve\primed$ be the revenue curve of distribution $\buyerdist\primed$.
    Invoking \Cref{lem:prelim:revenue expression in uniform order statistic}, we express the revenue of the {\BayesianOptimalMech} and the {\VCGAuction} as follows (recall $\osdensity_{a:b}$ is the PDF of $a$-th smallest order statistics of $b$ i.i.d.\ samples from uniform distribution $U[0,1]$):
    \begin{align*}
        \RevOPT_{m:n}\left(\buyerdist\primed\right) 
        & = 
        n\cdot \displaystyle\int_0^1 \revcurve\primed\left(\min\{\quant,\optquant\}\right)\cdot \osdensity_{m:n - 1}\left(\quant\right)\,\dd\quant
        \\
        & = 
        n\cdot \left(\displaystyle\int_0^{\optquant} \revcurve\primed\left(\quant\right) \cdot \osdensity_{m:n - 1}\left(\quant\right)\,\dd\quant 
        +
        \int_{\optquant}^1 \revcurve\primed\left(\optquant\right) \cdot \osdensity_{m:n - 1}\left(\quant\right)\,\dd\quant
        \right)
    \end{align*}
    and
    \begin{align*}
        \RevVCG_{m:n+k}\left(\buyerdist\primed\right) 
        & = 
        \left(n + k\right)\cdot \displaystyle\int_0^1 \revcurve\primed\left(\quant\right) \cdot \osdensity_{m:n + k - 1}\left(\quant\right)\,\dd\quant 
        \\
        & = 
        \left(n + k\right)\cdot \left(\displaystyle\int_0^{\optquant} \revcurve\primed\left(\quant\right) \cdot \osdensity_{m:n + k - 1}\left(\quant\right)\,\dd\quant 
        +
        \displaystyle\int_{\optquant}^1 \revcurve\primed\left(\quant\right) \cdot \osdensity_{m:n + k - 1}\left(\quant\right)\,\dd\quant 
        \right)
    \end{align*}
    We utilize the following technical lemma about the PDF of order statistics from the uniform distribution.
    \begin{lemma}[Single-crossing property]
    \label{lem:bk vcg:osdensity single crossing}
        Given any $m,n,k\in\naturals$ ($m\leq n$), there exists threshold quantile $\quantSC\in[0, 1]$ such that
        \begin{align*}
            \forall \quant \in [0, \quantSC]:&
            \qquad
            \left(n + k\right)\cdot \osdensity_{m:n + k - 1}\left(\quant\right)
            \geq 
            \ccapproxratio\cdot n\cdot \osdensity_{m:n - 1}\left(\quant\right)
            \\
            \forall \quant \in [\quantSC, 1]:&
            \qquad
            \left(n + k\right)\cdot \osdensity_{m:n + k - 1}\left(\quant\right)
            \leq 
            \ccapproxratio\cdot n\cdot \osdensity_{m:n - 1}\left(\quant\right)
        \end{align*}
    \end{lemma}
    \begin{proof}
        Consider the (normalized) difference between the left-hand side and right-hand side of the inequality in the lemma statement:
        \begin{align*}
            & 
            \frac{\left(n + k\right)}{\ccapproxratio\cdot n}\cdot \osdensity_{m:n + k - 1}\left(\quant\right)
            - 
            \osdensity_{m:n - 1}\left(\quant\right)
            \\
            ={}&
            \frac{\left(n + k\right)}{\ccapproxratio\cdot n}
            \cdot 
            \left(n + k - 1\right)\cdot \binom{n + k - 2}{m - 1}\cdot
            \left(1 - \quant\right)^{n + k - 1 - m}\cdot \quant^{m - 1}
            -
            \left(n - 1\right)\cdot \binom{n - 2}{m - 1}
            \cdot 
            \left(1 - \quant\right)^{n - 1 - m}\cdot \quant^{m - 1}
            \\
            ={} &
            \left(1 - \quant\right)^{n - 1 - m}\cdot \quant^{m - 1} \cdot 
            \left(
            \frac{\left(n + k\right)}{\ccapproxratio\cdot n}
            \cdot 
            \left(n + k - 1\right)\cdot \binom{n + k - 2}{m - 1}\cdot
            \left(1 - \quant\right)^{k}
            -
            \left(n - 1\right)\cdot \binom{n - 2}{m - 1}
            \right)
        \end{align*}
        Note that the difference is zero for boundary quantiles $\quant \in \{0, 1\}$. 
        In addition,
        \begin{align*}
            \frac{\left(n + k\right)}{\ccapproxratio\cdot n}
            \cdot 
            \left(n + k - 1\right)\cdot \binom{n + k - 2}{m - 1}\cdot
            \left(1 - \quant\right)^{k}
            -
            \left(n - 1\right)\cdot \binom{n - 2}{m - 1}
        \end{align*}
        is decreasing in $\quant\in[0, 1]$. Therefore, there exists at most one quantile $\quantSC \in \left(0, 1\right)$ such that the difference equals zero. Moreover, the difference is weakly positive (resp.\ weakly negative) for all quantiles smaller than (resp.\ larger than) $\quantSC$. This completes the proof of \Cref{lem:bk vcg:osdensity single crossing}.
    \end{proof}
    Invoking \Cref{lem:bk vcg:osdensity single crossing}, there exists threshold quantile $\quantSC\in[0, 1]$ such that
    \begin{align*}
        \forall \quant \in [0, \quantSC]:&
        \qquad
        \left(n + k\right)\cdot \osdensity_{m:n + k - 1}\left(\quant\right)
        \geq 
        \ccapproxratio\cdot n\cdot \osdensity_{m:n - 1}\left(\quant\right)
        \\
        \forall \quant \in [\quantSC, 1]:&
        \qquad
        \left(n + k\right)\cdot \osdensity_{m:n + k - 1}\left(\quant\right)
        \leq 
        \ccapproxratio\cdot n\cdot \osdensity_{m:n - 1}\left(\quant\right)
    \end{align*}
    Define auxiliary function $\auxfunc:[0, 1]\rightarrow\reals$ as follows:
    \begin{align*}
        \auxfunc\left(\quant\right) \triangleq \ccapproxratio\cdot n\cdot \osdensity_{m:n - 1}\left(\quant\right)
         - \left(n + k\right)\cdot \osdensity_{m:n + k - 1}\left(\quant\right)
    \end{align*}
    By construction, $\auxfunc\left(\quant\right) \leq 0$ for all  $\quant\in[0, \quantSC]$ and $\auxfunc\left(\quant\right) \geq 0$ for all $\quant\in[\quantSC,1]$.
    Consider the following two cases:

    \begin{figure}
        \centering
        \subfloat[Case (i)]{

\begin{tikzpicture}
\begin{axis}[
width=8cm,height=6cm,
xmin=0.0,xmax=1.05,
ymin=-0.15,ymax=1.05,
xtick={0, 0.5, 0.7, 1},   
xticklabels = {0, $\optquant$, $\quant\primed$, 1},
tick label style={font=\footnotesize},
ytick distance=1,
minor tick num=0,
axis line style=gray,
axis x line=middle,
axis y line=left,
tick style={black},
legend style={draw=none,fill=white,font=\small,at={(0.98,0.98)},anchor=north east},
]

\def\rstar{2.0}
\def\qstar{0.5}
\def\qdag{0.7}
\def\lnrstar{0.6931471805599453}
\def\scale{0.7142857142857143}

\definecolor{orangeLine}{RGB}{240,160,0}
\definecolor{blueLine}{RGB}{58,166,255}
\definecolor{greenLine}{RGB}{0,163,122}

\addplot[black, line width=1.8pt, domain=0:1, samples=500] {4*x*(1-x)};

\addplot[black!40, solid, line width=1.6pt, domain=1e-6:1, samples=700]
{ifthenelse(x<=\qstar, \rstar*x, (\rstar/\lnrstar)*x*ln(1/x))};

\addplot[gray, dotted, line width=1pt] coordinates {(\qstar,0) (\qstar,1.05)};

\addplot[gray, dotted, line width=1pt] coordinates {(\qdag,0) (\qdag,1.05)};

\end{axis}
\end{tikzpicture}         }
        ~~~~
        \subfloat[Case (ii)]{\begin{tikzpicture}
\begin{axis}[
width=8cm,height=6cm,
xmin=0.0,xmax=1.05,
ymin=-0.15,ymax=1.05,
xtick={0, 0.5, 0.3, 1},   
xticklabels = {0, $\optquant$, $\quant\primed$, 1},
tick label style={font=\footnotesize},
ytick distance=1,
minor tick num=0,
axis line style=gray,
axis x line=middle,
axis y line=left,
tick style={black},
legend style={draw=none,fill=white,font=\small,at={(0.98,0.98)},anchor=north east},
]
\def\rstar{2.0}
\def\qstar{0.5}
\def\qdag{0.3}
\def\lnrstar{0.6931471805599453}
\def\scale{0.7142857142857143}
\definecolor{orangeLine}{HTML}{F0A000}
\definecolor{blueLine}{HTML}{3AA6FF}
\definecolor{greenLine}{HTML}{00A37A}
\addplot[black,line width=1.8pt,domain=0:1,samples=600]{4*x*(1-x)};
\addplot[black!40,line width=1.6pt,solid,domain=1e-6:1,samples=900]{ifthenelse(x<=\qstar,\rstar*x,(\rstar/\lnrstar)*x*ln(1/x))};
\addplot[black,line width=1.8pt,dashed,domain=0:1,samples=600]{\scale*(4*x*(1-x))};
\addplot[gray,line width=1pt,dotted] coordinates {(\qstar,0) (\qstar,1.05)};
\addplot[gray, line width=1pt,dotted] coordinates {(\qdag,0) (\qdag,1.05)};
\end{axis}
\end{tikzpicture} }
        \caption{Graphical illustration of the analysis for \Cref{prop:bk vcg:worst case}. The black (resp.\ gray) solid curve is the revenue curve $\revcurve\primed$ (resp.\ $\WorstRevcurve$) induced by distribution $\buyerdist\primed$ (resp.\ $\WorstRegDist$). The black dashed line in Case (ii) is $\optreserve/\val\primed\cdot \revcurve\primed(\quant)$. Quantile $\optquant$ is the monopoly quantile and $\quant\primed$ is the quantile threshold defined in \Cref{lem:bk vcg:osdensity single crossing}.}
    \label{fig:bk vcg:worst case reduction}
    \end{figure}

    \xhdr{Case (i) [Quantiles $\optquant\leq \quantSC$]:} In this case, 
    consider $\optreserve$-truncated $\Reglevel$-generalized Pareto distribution~$\WorstRegDist$. By construction, $\WorstRegDist$ has the same monopoly reserve, monopoly quantile, and monopoly revenue as distribution $\buyerdist\primed$. In addition, $\distLambdaReserve$ is first order stochastically dominated by $\buyerdist\primed$, and thus the induced revenue curves $\WorstRevcurve\left(\quant\right) \leq \revcurve\primed\left(\quant\right)$ for all quantiles $\quant\in[0, 1]$. See \Cref{fig:bk vcg:worst case reduction} for a graphical illustration. Note that
    \begin{align*}
        &\ccapproxratio\cdot \RevOPT_{m:n}\left(\WorstRegDist\right)
        -
        \RevVCG_{m:n + k}\left(\WorstRegDist\right) 
        \\
        ={} & 
        \ccapproxratio\cdot n\cdot \left(\displaystyle\int_0^{\optquant} \WorstRevcurve\left(\quant\right) \cdot \osdensity_{m:n - 1}\left(\quant\right)\,\dd\quant 
        +
        \int_{\optquant}^1 \WorstRevcurve\left(\optquant\right) \cdot \osdensity_{m:n - 1}\left(\quant\right)\,\dd\quant
        \right)
        \\
        & \qquad 
        - 
        \left(n + k\right)\cdot \left(\displaystyle\int_0^{\optquant} \WorstRevcurve\left(\quant\right) \cdot \osdensity_{m:n + k - 1}\left(\quant\right)\,\dd\quant 
        +
        \displaystyle\int_{\optquant}^1 \WorstRevcurve\left(\quant\right) \cdot \osdensity_{m:n + k - 1}\left(\quant\right)\,\dd\quant 
        \right)
        \\
        ={} & 
        \displaystyle\int_0^{\optquant}
        \WorstRevcurve\left(\quant\right)
        \cdot 
        \auxfunc\left(\quant\right)
        \,
        \dd \quant
+
        \ccapproxratio\cdot n\cdot \int_{\optquant}^1 \WorstRevcurve\left(\optquant\right) \cdot \osdensity_{m:n - 1}\left(\quant\right)\,\dd\quant
        \\
        &\qquad
        -
        \left(n + k\right)\cdot \displaystyle\int_{\optquant}^1 \WorstRevcurve\left(\quant\right) \cdot \osdensity_{m:n + k - 1}\left(\quant\right)\,\dd\quant 
        \\
        \geq{} & 
        \displaystyle\int_0^{\optquant}
        \revcurve\primed\left(\quant\right)
        \cdot 
        \auxfunc\left(\quant\right)
        \,
        \dd \quant
+
        \ccapproxratio\cdot n\cdot \int_{\optquant}^1 \revcurve\primed\left(\optquant\right) \cdot \osdensity_{m:n - 1}\left(\quant\right)\,\dd\quant
        -
        \left(n + k\right)\cdot \displaystyle\int_{\optquant}^1 \revcurve\primed\left(\quant\right) \cdot \osdensity_{m:n + k - 1}\left(\quant\right)\,\dd\quant 
        \\
        ={} & 
        \ccapproxratio\cdot \RevOPT_{m:n}\left(\WorstRegDist\primed\right)
        -
        \RevVCG_{m:n + k}\left(\buyerdist\primed\right) 
        \\
        \geq {}& 0
    \end{align*}
    where the first inequality holds since (i) revenue curves satisfy $\WorstRevcurve\left(\quant\right)\leq \revcurve\primed\left(\quant\right)$ for all quantiles $\quant\in[0, 1]$, $\WorstRevcurve\left(\optquant\right) = \revcurve\primed\left(\optquant\right)$, and (ii) auxiliary function $\auxfunc\left(\quant\right)
    \leq 0$ for all quantiles $\quant\in[0,\optquant]\subseteq[0,\quantSC]$ (\Cref{lem:bk vcg:osdensity single crossing}). This completes the proof of the proposition statement for Case (i).

    \xhdr{Case (ii) [Quantiles $\optquant\geq \quantSC$]:}
    In this case, 
    consider $\optreserve$-truncated $\Reglevel$-generalized Pareto distribution~$\WorstRegDist$. By construction, $\WorstRegDist$ has the same monopoly reserve, monopoly quantile, and monopoly revenue as distribution $\buyerdist\primed$. In addition, $\WorstRegDist$ is first order stochastically dominated by $\buyerdist\primed$.
    Let $\valSC$ be the value such that $1 - \buyerdist\primed\left(\valSC\right) = \quantSC$. The revenue curve $\WorstRevcurve$ induced by distribution $\WorstRegDist$ satisfies 
    \begin{align*}
        \forall \quant\in[0,\quantSC]:&\qquad 
        \WorstRevcurve\left(\quant\right) \leq \frac{\optreserve}{\valSC}\cdot  \revcurve\primed\left(\quant\right) \\
        \forall \quant\in[\quantSC,\optquant]:& \qquad 
        \WorstRevcurve\left(\quant\right) \geq 
        \frac{\optreserve}{\valSC}\cdot \revcurve\primed\left(\quant\right)\\
        \forall \quant\in[\optquant, 1]:& \qquad 
        \WorstRevcurve\left(\quant\right) \leq \revcurve\primed\left(\quant\right)
    \end{align*}
    and $\WorstRevcurve\left(\optquant\right) = \revcurve\primed\left(\optquant\right)$,
    where the first and second inequalities holds since distribution $\buyerdist\primed$ is $\Reglevel$-regular and thus the revenue curve $\revcurve\primed$ is concave. See \Cref{fig:bk vcg:worst case reduction} for a graphical illustration. Note that 
    \begin{align*}
        &\ccapproxratio\cdot \RevOPT_{m:n}\left(\WorstRegDist\right)
        -
        \RevVCG_{m:n + k}\left(\WorstRegDist\right) 
        \\    
        ={} & 
        \ccapproxratio\cdot n\cdot \left(\displaystyle\int_0^{\optquant} \WorstRevcurve\left(\quant\right) \cdot \osdensity_{m:n - 1}\left(\quant\right)\,\dd\quant 
        +
        \int_{\optquant}^1 \WorstRevcurve\left(\optquant\right) \cdot \osdensity_{m:n - 1}\left(\quant\right)\,\dd\quant
        \right)
        \\
        & \qquad 
        - 
        \left(n + k\right)\cdot \left(\displaystyle\int_0^{\optquant} \WorstRevcurve\left(\quant\right) \cdot \osdensity_{m:n + k - 1}\left(\quant\right)\,\dd\quant 
        +
        \displaystyle\int_{\optquant}^1 \WorstRevcurve\left(\quant\right) \cdot \osdensity_{m:n + k - 1}\left(\quant\right)\,\dd\quant 
        \right)
        \\
        ={} & 
        \ccapproxratio\cdot n\cdot \left(\displaystyle\int_0^{\quantSC} \WorstRevcurve\left(\quant\right) \cdot \osdensity_{m:n - 1}\left(\quant\right)\,\dd\quant 
        +
        \displaystyle\int_{\quantSC}^{\optquant} \WorstRevcurve\left(\quant\right) \cdot \osdensity_{m:n - 1}\left(\quant\right)\,\dd\quant 
        +
        \int_{\optquant}^1 \WorstRevcurve\left(\optquant\right) \cdot \osdensity_{m:n - 1}\left(\quant\right)\,\dd\quant
        \right)
        \\
        & \qquad 
        - 
        \left(n + k\right)\cdot \left(\displaystyle\int_0^{\quantSC} \WorstRevcurve\left(\quant\right) \cdot \osdensity_{m:n + k - 1}\left(\quant\right)\,\dd\quant 
        +
        \displaystyle\int_{\quantSC}^{\optquant} \WorstRevcurve\left(\quant\right) \cdot \osdensity_{m:n + k - 1}\left(\quant\right)\,\dd\quant \right.
        \\
        &
        \hspace{10cm}
        \left.
        +
        \displaystyle\int_{\optquant}^1 \WorstRevcurve\left(\quant\right) \cdot \osdensity_{m:n + k - 1}\left(\quant\right)\,\dd\quant 
        \right)
        \\
        ={} & 
        \displaystyle\int_0^{\quantSC}
        \WorstRevcurve\left(\quant\right)
        \cdot 
        \auxfunc\left(\quant\right)
        \,
        \dd \quant
+
        \displaystyle\int_{\quantSC}^{\optquant}
        \WorstRevcurve\left(\quant\right)
        \cdot 
        \auxfunc\left(\quant\right)
        \,
        \dd \quant
        \\
        & \qquad
        +
        \ccapproxratio\cdot n\cdot \int_{\optquant}^1 \WorstRevcurve\left(\optquant\right) \cdot \osdensity_{m:n - 1}\left(\quant\right)\,\dd\quant
        -
        \left(n + k\right)\cdot \displaystyle\int_{\optquant}^1 \WorstRevcurve\left(\quant\right) \cdot \osdensity_{m:n + k - 1}\left(\quant\right)\,\dd\quant 
        \\
        = {} & 
        \frac{\optreserve}{\valSC}
        \cdot 
        \left(
        \frac{\valSC}{\optreserve}
        \cdot 
        \displaystyle\int_0^{\quantSC}
        \WorstRevcurve\left(\quant\right)
        \cdot 
        \auxfunc\left(\quant\right)
        \,
        \dd \quant
        \right.
+
        \frac{\valSC}{\optreserve}\cdot 
        \displaystyle\int_{\quantSC}^{\optquant}
        \WorstRevcurve\left(\quant\right)
        \cdot 
        \auxfunc\left(\quant\right)
        \,
        \dd \quant
        \\
        & \left.\qquad
        +
        \frac{\valSC}{\optreserve}\cdot 
        \ccapproxratio\cdot n\cdot \int_{\optquant}^1 \WorstRevcurve\left(\optquant\right) \cdot \osdensity_{m:n - 1}\left(\quant\right)\,\dd\quant
        -
        \frac{\valSC}{\optreserve}\cdot 
        \left(n + k\right)\cdot \displaystyle\int_{\optquant}^1 \WorstRevcurve\left(\quant\right) \cdot \osdensity_{m:n + k - 1}\left(\quant\right)\,\dd\quant 
        \right)
        \\  
        \geq {} & 
        \frac{\optreserve}{\valSC}
        \cdot 
        \left(
        \displaystyle\int_0^{\quantSC}
        \revcurve\primed\left(\quant\right)
        \cdot 
        \auxfunc\left(\quant\right)
        \,
        \dd \quant
        \right.
+
        \displaystyle\int_{\quantSC}^{\optquant}
        \revcurve\primed\left(\quant\right)
        \cdot 
        \auxfunc\left(\quant\right)
        \,
        \dd \quant
        \\
        & \left.\qquad
        + {}~
        \ccapproxratio\cdot n\cdot \int_{\optquant}^1 \revcurve\primed\left(\optquant\right) \cdot \osdensity_{m:n - 1}\left(\quant\right)\,\dd\quant
        -
        \left(n + k\right)\cdot \displaystyle\int_{\optquant}^1 \revcurve\primed\left(\quant\right) \cdot \osdensity_{m:n + k - 1}\left(\quant\right)\,\dd\quant 
        \right)
        \\
        ={} & \frac{\optreserve}{\valSC}
        \cdot\left(\ccapproxratio\cdot \RevOPT_{m:n}\left(\buyerdist\primed\right)
        -
        \RevVCG_{m:n + k}\left(\buyerdist\primed\right) \right)
        \\
        \geq {} & 0
    \end{align*}
    Here, the third equality holds due to the construction of auxiliary function $\auxfunc$, and the first inequality holds since 
    \begin{align*}
        &\frac{\valSC}{\optreserve}
        \cdot 
        \displaystyle\int_0^{\quantSC}
        \WorstRevcurve\left(\quant\right)
        \cdot 
        \auxfunc\left(\quant\right)
        \,
        \dd \quant
        \overset{\left(a\right)}{\geq} 
        \displaystyle\int_0^{\quantSC}
        \revcurve\primed\left(\quant\right)
        \cdot 
        \auxfunc\left(\quant\right)
        \,
        \dd \quant
        \\
        &\frac{\valSC}{\optreserve}
        \cdot 
        \displaystyle\int_{\quantSC}^{\optquant}
        \WorstRevcurve\left(\quant\right)
        \cdot 
        \auxfunc\left(\quant\right)
        \,
        \dd \quant
        \overset{\left(b\right)}{\geq} 
        \displaystyle\int_{\quantSC}^{\optquant}
        \revcurve\primed\left(\quant\right)
        \cdot 
        \auxfunc\left(\quant\right)
        \,
        \dd \quant
        \\
        &\frac{\valSC}{\optreserve}\cdot 
        \ccapproxratio\cdot n\cdot \int_{\optquant}^1 \WorstRevcurve\left(\optquant\right) \cdot \osdensity_{m:n - 1}\left(\quant\right)\,\dd\quant
        -
        \frac{\valSC}{\optreserve}\cdot 
        \left(n + k\right)\cdot \displaystyle\int_{\optquant}^1 \WorstRevcurve\left(\quant\right) \cdot \osdensity_{m:n + k - 1}\left(\quant\right)\,\dd\quant 
        \\
        & \qquad \qquad\overset{\left(c\right)}{\geq}{} 
        \ccapproxratio\cdot n\cdot \int_{\optquant}^1 \WorstRevcurve\left(\optquant\right) \cdot \osdensity_{m:n - 1}\left(\quant\right)\,\dd\quant
        -
        \left(n + k\right)\cdot \displaystyle\int_{\optquant}^1 \WorstRevcurve\left(\quant\right) \cdot \osdensity_{m:n + k - 1}\left(\quant\right)\,\dd\quant 
        \\
        & \qquad \qquad\overset{\left(d\right)}{\geq}{} 
        \ccapproxratio\cdot n\cdot \int_{\optquant}^1 \revcurve\primed\left(\optquant\right) \cdot \osdensity_{m:n - 1}\left(\quant\right)\,\dd\quant
        -
        \left(n + k\right)\cdot \displaystyle\int_{\optquant}^1 \revcurve\primed\left(\quant\right) \cdot \osdensity_{m:n + k - 1}\left(\quant\right)\,\dd\quant 
    \end{align*}
    where inequality~(a) holds since $\WorstRevcurve\left(\quant\right) \leq \frac{\optreserve}{\valSC}\cdot \revcurve\primed\left(\quant\right)$ and $\auxfunc\left(\quant\right) \leq 0$ for all quantiles $\quant \in [0, \quantSC]$, inequality~(b) holds since $\WorstRevcurve\left(\quant\right) \geq \frac{\optreserve}{\valSC}\cdot \revcurve\primed\left(\quant\right)$ and $\auxfunc\left(\quant\right) \geq 0$ for all quantiles $\quant \in [\quantSC,\optquant]$, and inequality~(c) holds since $\frac{\valSC}{\optreserve} \geq 1$ (implied by the case assumption that $\quantSC \leq \optquant$ and the construction of $\valSC$) and $\ccapproxratio\cdot n\cdot \int_{\optquant}^1 \WorstRevcurve\left(\optquant\right) \cdot \osdensity_{m:n - 1}\left(\quant\right)\,\dd\quant
    -
    \left(n + k\right)\cdot \int_{\optquant}^1 \WorstRevcurve\left(\quant\right) \cdot \osdensity_{m:n + k - 1}\left(\quant\right)\,\dd\quant  \geq 0$ (implied by the fact that $\WorstRevcurve\left(\optquant\right) \geq \WorstRevcurve\left(\quant\right)$ for all $\quant\in[\optquant,1]$ and $\ccapproxratio\cdot n\cdot\osdensity_{m:n-1}\left(\quant\right)\,\dd\quant - \left(n+k\right)\cdot \osdensity_{m:n+k-1}\left(\quant\right)\geq 0$ for all $\quant\in[\optquant, 1]$), and inequality~(d) holds since $\WorstRevcurve\left(\optquant\right) = \revcurve\primed\left(\optquant\right)$, $\WorstRevcurve\left(\quant\right) \leq \revcurve\primed\left(\quant\right)$ for all $\quant\in[\optquant,1]$. This completes the proof of the proposition statement for Case~(ii). 

    Putting the analysis for the two cases together, we finish the proof of \Cref{prop:bk vcg:worst case}.
\end{proof}
 
\label{apx:thmVCGLargeMarket}

We now prove \Cref{thm:bk vcg:large market}.
\begin{proof}[Proof of \Cref{thm:bk vcg:large market}] 
    Fix $\Reglevel\in[0, 1]$, $\imbalanceratio\in(0, 1]$, and $\ccapproxratio\in(0, 1)$. For ease of presentation, define auxiliary notations $m_n$ and $k_n$ as
    \begin{align*}
        m_n \triangleq \lceil\imbalanceratio\cdot n\rceil
        \;\;\mbox{and}\;\;
        k_n \triangleq \ccomplexity\left(\lceil\imbalanceratio\cdot n\rceil,n,\lambdaDists,\ccapproxratio\right)
    \end{align*}
    Invoking \Cref{prop:bk vcg:worst case}, $k_n$ is the smallest non-negative integer such that for all $\reserve\in[1,(1-\Reglevel)^{-1/\Reglevel}]$,
    \begin{align}
    \label{eq:bk vcg:large market:condition}
        \RevVCG_{m_n:n+k_n}\left(\distLambdaReserve\right) \geq 
        \ccapproxratio\cdot 
        \RevOPT_{m_n:n}\left(\distLambdaReserve\right) 
    \end{align}
    where $\distLambdaReserve$ is the $\reserve$-truncated $\Reglevel$-generalized Pareto distribution (\Cref{def:truncated generalized Pareto distribution}). By construction, distribution $\Reglevel$ has monopoly reserve $\reserve$, monopoly quantile $1/\reserve$, and thus monopoly revenue of one. 

    Invoking \Cref{lem:prelim:revenue expression in uniform order statistic}, we express the revenue of the {\VCGAuction} and the {\BayesianOptimalMech} as follows (recall $\osdensity_{a:b}$ is the PDF of $a$-th smallest order statistics of $b$ i.i.d.\ samples from uniform distribution $U[0,1]$):
    \begin{align*}
        \lim_{n\rightarrow \infty}
        \RevVCG_{m_n:n+k_n}\left(\distLambdaReserve\right) 
        &= 
        \lim_{n\rightarrow \infty}
        \left(n+k_n\right)\cdot \displaystyle\int_0^1 \revcurve_{(\Reglevel,\reserve)}\left(\quant\right)
        \cdot \osdensity_{m_n:n+k_n - 1}\left(\quant\right)\,\dd \quant \notag
        \\
        &\overset{(a)}{=}
        \lim_{n\rightarrow \infty} \left(n+k_n\right)\cdot \revcurve_{(\Reglevel,\reserve)}\left(\frac{m_n}{n+k_n}\right)~,
    \end{align*}
    and
    \begin{align*}
        \lim_{n\rightarrow \infty}
        \RevOPT_{m_n:n}\left(\distLambdaReserve\right) 
        &= 
        \lim_{n\rightarrow \infty}
        n\cdot \displaystyle\int_0^1
        \revcurve_{(\Reglevel,\reserve)}\left(\min\left\{\quant,\frac{1}{\reserve}\right\}\right)\cdot \osdensity_{m_n:n-1}\left(\quant\right)\,\dd\quant \notag
        \\
        &\overset{(b)}{=}
        \lim_{n\rightarrow \infty} n\cdot \revcurve_{(\Reglevel,\reserve)}\left(\min\left\{\frac{m_n}{n},\frac{1}{\reserve}\right\}\right)~,
    \end{align*}
    where $\revcurve_{(\Reglevel,\reserve)}$ is the induced revenue curve of distribution $\distLambdaReserve$, and both equalities~(a) and (b) hold due to the concentration of the order statistics ($\expect[\quant\sim\osdensity_{m_n:n+k_n - 1}]{\quant}=m_n/(n+k_n)$ and $\expect[\quant\sim\osdensity_{m_n:n - 1}]{\quant}=m_n/n$) and the Lipschitz continuity (implied by the $\Reglevel$-regularity of $\distLambdaReserve$) of the revenue curve $\revcurve_{(\Reglevel,\reserve)}(\cdot)$ and $\revcurve_{(\Reglevel,\reserve)}(\min\{\cdot,1/\reserve\})$ at $m_n/(n+k_n)$ and $m_n/n$, respectively.

    Below, we consider two cases depending on the value of the supply-to-demand ratio $\imbalanceratio$ separately.

    \xhdr{Case (i) [Supply-to-demand ratio $\imbalanceratio\in[(1-\Reglevel)^{1/\Reglevel}, 1]$]:}
    In this case, we further consider three subcases depending on the value of truncation $\reserve$ (which is also the monopoly reserve by construction) in distribution $\distLambdaReserve$.
    
    \xhdr{Case (i.a) [Truncation $\reserve\in[1,n/m_n]$]:} In this subcase, the revenue curve $\revcurve_{(\Reglevel,\reserve)}$ satisfies 
    \begin{align*}
        \revcurve_{(\Reglevel,\reserve)}\left(\frac{m_n}{n+k_n}\right) = \reserve\cdot \frac{m_n}{n+k_n}
        \;\;
        \mbox{and}
        \;\;
        \revcurve_{(\Reglevel,\reserve)}\left(\min\left\{\frac{m_n}{n},\frac{1}{\reserve}\right\}\right)
        =
        \reserve\cdot \frac{m_n}{n}
    \end{align*}
    Therefore, 
    \begin{align*}
        \lim_{n\rightarrow\infty}~
        \frac{\RevVCG_{m_n:n+k_n}\left(\distLambdaReserve\right)}{\RevOPT_{m_n:n}\left(\distLambdaReserve\right)}
        =
        \lim_{n\rightarrow\infty}~\frac{\left(n+k_n\right)\cdot \reserve \cdot \frac{m_n}{n + k_n}}{n\cdot \reserve\cdot \frac{m_n}{n}}
        =
        1 > \ccapproxratio
    \end{align*}

    \xhdr{Case (i.b) [Truncation $\reserve \in(n/m_n, (n+k_n)/m_n]$]:} In this subcase, the revenue curve $\revcurve_{(\Reglevel,\reserve)}$ satisfies
    \begin{align*}
        \revcurve_{(\Reglevel,\reserve)}\left(\frac{m_n}{n+k_n}\right) = \reserve\cdot \frac{m_n}{n+k_n}
        \;\;
        \mbox{and}
        \;\;
        \revcurve_{(\Reglevel,\reserve)}\left(\min\left\{\frac{m_n}{n},\frac{1}{\reserve}\right\}\right)
        =
        1
    \end{align*}
    Therefore, 
    \begin{align*}
        \lim_{n\rightarrow\infty}~
        \frac{\RevVCG_{m_n:n+k_n}\left(\distLambdaReserve\right)}{\RevOPT_{m_n:n}\left(\distLambdaReserve\right)}
        =
        \lim_{n\rightarrow\infty}~\frac{\left(n+k_n\right)\cdot \reserve \cdot \frac{m_n}{n + k_n}}{n\cdot 1}
        =
        \frac{m_n}{n}\cdot \reserve > 1 > \ccapproxratio
    \end{align*}

    \xhdr{Case (i.c) [Truncation $\reserve \in((n+k_n)/m_n,(1-\Reglevel)^{-1/\Reglevel}]$]:} In this subcase, the revenue curve $\revcurve_{(\Reglevel,\reserve)}$ satisfies
    \begin{align*}
        &\revcurve_{(\Reglevel,\reserve)}\left(\frac{m_n}{n+k_n}\right) = \frac{\reserve}{\reserve^\Reglevel - 1}\cdot \left(\left(\frac{n+k_n}{m_n}\right)^\Reglevel-1\right)\cdot \frac{m_n}{n+k_n}
        \geq 
        \frac{1}{\Reglevel}
        \cdot 
        \left(\frac{1}{1-\Reglevel}\right)^{\frac{1-\Reglevel}{\Reglevel}}\cdot \left(\left(\frac{n+k_n}{m_n}\right)^\Reglevel-1\right)\cdot \frac{m_n}{n+k_n}
        \\
        \mbox{and}
        \;\;
        &\revcurve_{(\Reglevel,\reserve)}\left(\min\left\{\frac{m_n}{n},\frac{1}{\reserve}\right\}\right)
        =
        1 
    \end{align*}
    where the inequality holds since $\reserve/(\reserve^\Reglevel - 1)$ is decreasing in $\reserve\in((n+k_n)/m_n),(1-\Reglevel)^{-1/\Reglevel}]$.
    Therefore, 
    \begin{align*}
        \lim_{n\rightarrow\infty}~
        \frac{\RevVCG_{m_n:n+k_n}\left(\distLambdaReserve\right)}{\RevOPT_{m_n:n}\left(\distLambdaReserve\right)}
        \geq
        \lim_{n\rightarrow\infty}~\frac{(n+k_n)\cdot\frac{1}{\Reglevel}
        \cdot 
        \left(\frac{1}{1-\Reglevel}\right)^{\frac{1-\Reglevel}{\Reglevel}}\cdot \left(\left(\frac{n+k_n}{m_n}\right)^\Reglevel-1\right)\cdot \frac{m_n}{n+k_n}}{n\cdot 1}
        \ge \ccapproxratio
    \end{align*}
    where the last inequality holds due to the value of $k_n$ provided in the theorem statement.

    \xhdr{Case (ii) [Supply-to-demand ratio $\imbalanceratio\in(0,(1-\Reglevel)^{1/\Reglevel})$]:}
    In this case, since truncation $\reserve \leq (1-\Reglevel)^{-1/\Reglevel}$ by construction (\Cref{def:truncated generalized Pareto distribution}), it guarantees that $\reserve \in [1, n/m_n]$ and thus the same argument as Case (i.a) applies.

    \smallskip
    \noindent
    Combining all the pieces together, we complete the proof of \Cref{thm:bk vcg:large market} as desired.
\end{proof}
  
\section{Supply-Limiting VCG Auction}
\label{sec:bk supply limiting}

Unlike the single-unit setting, where the {\VCGAuction} achieves the optimal competition complexity of just $1$ additional buyer under regular (and even MHR) distributions \citep{BK-96}, it remains unclear whether its competition complexity is optimal among all prior-independent mechanisms in the multi-unit setting. Motivated by this question, we study the following variant of the {\VCGAuction} and analyze its competition complexity.

\begin{definition}[{\SLVCGAuction}]
    In an $m$-unit $n$-buyer market with $m\leq n$, the {\SLVCGAuction} with supply $\supply \in [m]$ runs the {\VCGAuction} to sell $\supply$ units to the $n$ buyers.
\end{definition}

We remark that the supply $\supply$ is a design parameter of the {\SLVCGAuction} and must be chosen without knowledge of the valuation distribution $\buyerdist$, i.e., in a prior-independent manner. 
It can be verified that the {\SLVCGAuction} is therefore prior-independent, DSIC, and ex post IR. Moreover, the classic {\VCGAuction} corresponds to a special case of the {\SLVCGAuction} in which the supply is degenerate at $m$.
We define its competition complexity as follows:

\begin{definition}[Competition complexity of {\SLVCGAuction}]
    For any $m, n \in \naturals$ with $m \leq n$, any $\ccapproxratio \in (0, 1]$, and any distribution class $\DistClass$, the \emph{competition complexity} of the {\SLVCGAuction} is defined as
    \begin{align*}    
\ccomplexitySL\left(m, n, \DistClass, \ccapproxratio\right) 
        &\triangleq 
        \min\left\{k \in \naturals :
        \exists\supply\in[m],~\forall\,\buyerdist \in \DistClass,\,
        \RevSLVCG_{\supply:n+k}\left(\buyerdist\right) \geq 
        \ccapproxratio \cdot \RevOPT_{m:n}\left(\buyerdist\right)
        \right\}.
    \end{align*}
    If no finite $k$ satisfies the condition, the competition complexity is defined to be $\infty$.
\end{definition}

Similar to the analysis in \Cref{sec:bk vcg}, we first show that the worst-case distribution for the {\SLVCGAuction} is again given by the truncated $\Reglevel$-generalized Pareto distributions (\Cref{def:truncated generalized Pareto distribution}), and then present its asymptotic competition complexity.
Both arguments closely mirror their counterparts for the classic {\VCGAuction} and are deferred to \Cref{apx:propVCGSLWorstDists,apx:thmVCGSLLargeMarket}.

\begin{restatable}[Worst-case distribution characterization]{proposition}{propVCGSLWorstDists}
\label{prop:bk supply limiting:worst case}
    For any $\Reglevel\in[0, 1]$, $m, n \in \naturals$ with $m\leq n$, and $\ccapproxratio\in(0, 1]$, the competition complexity of the {\SLVCGAuction} in the $m$-unit, $n$-buyer market over the class of $\Reglevel$-regular distributions $\lambdaDists$ coincides with that over the class of truncated $\Reglevel$-generalized Pareto distributions $\truncGPDists$; that is,
    \begin{align*}
        \ccomplexitySL\left(m,n,\lambdaDists,\ccapproxratio\right) = \ccomplexitySL\left(m,n,\truncGPDists,\ccapproxratio\right).
    \end{align*}
\end{restatable}

Similar to \Cref{sec:bk vcg:large market}, to ensure a tractable analysis and obtain clean asymptotic results, we consider the
the large-market regime parameterized by the supply-to-demand ratio $\imbalanceratio\in(0, 1)$ and define the asymptotic optimal competition complexity of the {\SLVCGAuction} with supply-to-demand ratio $\imbalanceratio$ as 
\begin{align*}
    \ccomplexitySLInfty(\imbalanceratio,\lambdaDists,\ccapproxratio) \triangleq 
    \lim_{n\to \infty}~
    \frac{\ccomplexitySL(\lceil\imbalanceratio\cdot n\rceil, n,\lambdaDists,\ccapproxratio)}{n}~.
\end{align*}
Our characterization of the asymptotic optimal competition complexity of the {\SLVCGAuction} is as follows.
\begin{restatable}[Competition complexity in large market]{theorem}{thmVCGSLLargeMarket}
\label{thm:bk supply limiting:large market}
    For any $\Reglevel\in[0, 1]$, $\imbalanceratio\in(0, 1]$, and $\ccapproxratio\in(0, 1)$, the asymptotic optimal competition complexity $\ccomplexitySLInfty(\imbalanceratio,\lambdaDists,\ccapproxratio)$ of the {\VCGAuction} with supply-to-demand ratio $\imbalanceratio$, over the class of $\Reglevel$-regular distributions $\lambdaDists$ and target approximation ratio $\ccapproxratio$ satisfies
    \begin{align*}
        &\ccomplexitySLInfty(\imbalanceratio,\lambdaDists,\ccapproxratio)
        =
        \plus{
        \imbalanceratio\cdot\ccapproxratio\cdot {\left(\frac{\Reglevel}{\imbalanceratio}\cdot{\left(1-\Reglevel\right)}^{\frac{1-\Reglevel}{\Reglevel}}+1 \right)}^{\frac{1}{\Reglevel}}-1
        }
        \cdot \indicator{\imbalanceratio\geq (1-\Reglevel)^{\frac{1}{\Reglevel}}} 
    \intertext{
    which is achieved by setting supply $\supply = \lceil\ccapproxratio\cdot \lceil\imbalanceratio\cdot n\rceil\rceil$. In the special cases of MHR distributions ($\Reglevel = 0$) and regular distributions ($\Reglevel = 1$), 
    the asymptotic competition complexity simplifies to}
        &
        \ccomplexitySLInfty(\imbalanceratio,\MHRDists,\ccapproxratio)
        =
        \plus{
        \imbalanceratio\cdot\ccapproxratio\cdot e^{\frac{1}{e\imbalanceratio}} - 1}
        \cdot \indicator{\imbalanceratio\geq \frac{1}{e}}
        \\
        &
        \ccomplexitySLInfty(\imbalanceratio,\RegDists,\ccapproxratio)
        =
        \plus{
        \imbalanceratio\cdot \ccapproxratio + \ccapproxratio - 1 
        }
    \end{align*}
\end{restatable}

We present several interpretations of the asymptotic competition complexity in \Cref{thm:bk supply limiting:large market}.

First, we observe that as the target approximation ratio $\ccapproxratio$ approaches 1, the optimal supply for the {\SLVCGAuction}, namely $\supply = \lceil \ccapproxratio\cdot \lceil\imbalanceratio \cdot n \rceil\rceil$, converges to the original supply $m = \lceil \imbalanceratio \cdot n \rceil$. Consequently, the competition complexities of the classic {\VCGAuction} and the {\SLVCGAuction} coincide in this regime. The intuition is straightforward: within the class of $\Reglevel$-regular distributions lies the degenerate distribution that places a point mass at value 1 (i.e., all buyers have deterministic value 1). Under this distribution, to fully outperform the {\BayesianOptimalMech}, it is necessary to sell all $\lceil \imbalanceratio \cdot n \rceil$ units; hence, any strictly smaller supply $\supply$ would fail to achieve the target. Therefore, when the goal is exact optimality ($\ccapproxratio = 1$), prior-independent supply limiting offers no advantage over the standard {\VCGAuction}.

However, under approximate optimality ($\ccapproxratio < 1$), the {\SLVCGAuction} strictly improves upon the {\VCGAuction} by requiring fewer additional buyers. As stated in \Cref{thm:bk supply limiting:large market}, in the asymptotic large-market regime, the optimal supply choice is $\supply = \lceil \ccapproxratio \cdot \lceil\imbalanceratio \cdot n \rceil\rceil$. This ensures that even in the worst-case deterministic instance (where all values equal 1), the revenue achieves exactly the target fraction $\ccapproxratio$ of the optimal---making it the minimal supply that satisfies the approximation guarantee across all $\Reglevel$-regular distributions.

\Cref{fig:VCG-vs-SLVCG} compares the asymptotic competition complexity of the {\VCGAuction} and {\SLVCGAuction} for MHR ($\MHRDists$), $0.5$-regular ($\mathcal{F}_{\text{0.5-Reg}}$), and regular ($\RegDists$) distributions in the balanced market ($\imbalanceratio = 1$) as $\ccapproxratio$ varies in $(0,1)$. 
For example, when $\ccapproxratio = 0.8$, it is sufficient for the {\SLVCGAuction} to add only $15.6\%$, $25\%$, and $60\%$ additional buyers for $\MHRDists$, $\mathcal{F}_{\text{0.5-Reg}}$, and $\RegDists$, respectively---compared to $34.2\%$, $44\%$, and $80\%$ for the {\VCGAuction}. 
When $\ccapproxratio = 0.6$, the it is sufficient for {\SLVCGAuction} to add a \emph{sublinear} number of additional buyers for both $\MHRDists$ and $\mathcal{F}_{\text{0.5-Reg}}$, and only $20\%$ for $\RegDists$, whereas it is necessary for the {\VCGAuction} to add $24.7\%$, $32.3\%$, and $60\%$ additional buyers, respectively. 
\bibliographystyle{apalike}
	\bibliography{refs.bib}

\appendix

\section{Analysis of \texorpdfstring{\Cref{example:intro:multi-unit bk:regular}}{Example 1.1}}
\label{apx:intro example analysis}

In this section, we analyze \Cref{example:intro:multi-unit bk:regular}.
First, we consider the {\BayesianOptimalMech}. Since the market is balanced, {\BayesianOptimalMech} reduces to offer the monopoly price to all buyers, and thus, the expected revenue is 
\begin{align*}
    \RevOPT_{n:n}(\buyerdist) = n\cdot H\cdot \frac{1}{H+1} = n - \frac{1}{3}
\end{align*}
where the last equality holds since $H = 3n-1$.

Next we analyze the {\VCGAuction} with $2N - 1$ buyers. Since all buyers are ex ante symmetric. It suffices to analyze a fixed buyer. Notably, a fixed buyer $i$ receives the item if and only if her value $\val_i$ is higher than the $n$-th highest values among others $2n-2$ buyers. Thus, by considering the interim allocation and revenue curve in the quantile space, the expected payment from a fixed agent $i$ can be expressed as
\begin{align*}
    \displaystyle\int_0^1 \revcurve(\quant)\cdot \osdensity_{(n:2n-2)}(\quant)\,\dd \quant 
    \leq 
    \displaystyle\int_0^1 (1-\quant)\cdot \osdensity_{(n:2n-2)}(\quant)\,\dd \quant
    \leq \frac{1}{2}
\end{align*}
where $\revcurve$ is the revenue curve induced by distribution $\buyerdist$, and $\osdensity_{a:b}:[0,1]\rightarrow \reals_+$ is the probability density function (PDF) of $a$-th order statistics from $b$ i.i.d.\ samples from the uniform distribution $U[0, 1]$. Here the first inequality holds due to the fact that $\revcurve(\quant) \leq 1 - \quant$ for all $\quant\in[0,1]$ (implied by the construction of $\buyerdist$), and 
the second inequality holds due to the facts that $\int_0^1  \osdensity_{(n:2n-2)}(\quant)\,\dd \quant=1$ and $\int_0^1 \quant\cdot \osdensity_{(n:2n-2)}(\quant)\,\dd \quant\geq \frac{1}{2}$ (guaranteed since $\osdensity_{(n:2n-2)}$ is a PDF, and its expectation is larger than $1/2$ due to the symmetry of uniform distribution).
Thus, the expected revenue in the {\VCGAuction} with $2n - 1$ buyers is 
\begin{align*}
    \RevVCG_{n:2n-1}(\buyerdist) \leq \frac{1}{2}\cdot (2n-1) = n - \frac{1}{2}
\end{align*}
which is strictly less than the expected revenue in the {\BayesianOptimalMech}.

\section{Counterexample for Multiplicative Gap Monotonicity}
\label{sec:FOSD monotonicity discussion}
In this section, we provide an example demonstrating that the multiplicative revenue gap,
\begin{align*}
    \frac{\RevVCG_{m:n+k}(\buyerdist)}{\RevOPT_{m:n}(\buyerdist)},
\end{align*}
can be \emph{strictly decreasing} as the distribution $\buyerdist$ becomes larger in the first-order stochastic dominance (FOSD) order—even when the monopoly reserve and monopoly quantile are held fixed.

\begin{example}
\label{example:ratio monotonicity}
Consider two regular revenue curves $\revcurve_1$ and $\revcurve_2$ defined by
\begin{align*}
        \revcurve_1(\quant) \triangleq \left\{
        \begin{array}{ll}
         \frac{8}{7}\cdot \quant    & \quant\in[0,\frac{7}{8}] \\
         1    & \quant\in[\frac{7}{8},1]
        \end{array}
        \right.
        \;\;\mbox{and}\;\;
\revcurve_2(\quant) \triangleq \left\{
        \begin{array}{ll}
         \frac{8}{7}\cdot \quant    & \quant\in[0,\frac{3}{4}] \\
          \frac{6}{7} + \frac{4}{7}\cdot (\quant-\frac{3}{4}) & \quant\in[\frac{3}{4},1]   
        \end{array}
        \right.
    \end{align*}
Let $\buyerdist_1$ and $\buyerdist_2$ denote the valuation distributions corresponding to $\revcurve_1$ and $\revcurve_2$, respectively. Both distributions are regular, and $\buyerdist_1$ first-order stochastically dominates $\buyerdist_2$. (See \Cref{fig:ratio monotonicity} for an illustration.)

Fix $m = 2$, $n = 3$, and $k = 1$. Using \Cref{lem:prelim:revenue expression in uniform order statistic}, we compute the expected revenues of the {\BayesianOptimalMech} and the {\VCGAuction} under both distributions:
\begin{align*}
    \RevVCG_{2:4}(\buyerdist_1) &\approx 2.27734375, &
    \RevVCG_{2:4}(\buyerdist_2) &\approx 2.25446428571, \\
    \RevOPT_{2:3}(\buyerdist_1) &\approx 2.234375, &
    \RevOPT_{2:3}(\buyerdist_2) &= 2.1875.
\end{align*}
As expected from revenue monotonicity, both mechanisms yield strictly higher revenue under $\buyerdist_1$ than under $\buyerdist_2$. However, the multiplicative revenue ratio satisfies
\begin{align*}
    \frac{\RevVCG_{2:4}(\buyerdist_1)}{\RevOPT_{2:3}(\buyerdist_1)}
    \approx 1.0192
    \;<\;
    1.0306
    \approx
    \frac{\RevVCG_{2:4}(\buyerdist_2)}{\RevOPT_{2:3}(\buyerdist_2)},
\end{align*}
showing that the ratio is strictly lower under the FOSD-dominant distribution $\buyerdist_1$.
\end{example}

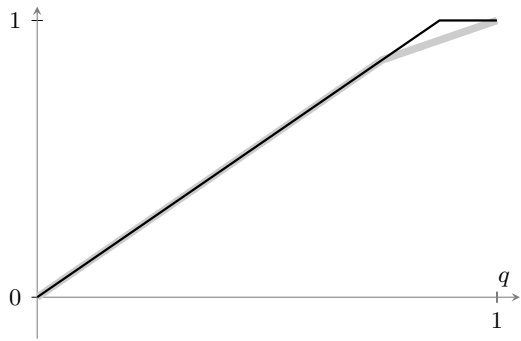
\begin{figure}
    \centering

\begin{tikzpicture}
\begin{axis}[
width=8cm,height=6cm,
xmin=0.0,xmax=1.05,
ymin=-0.15,ymax=1.05,
xtick={0,1},   
xticklabels = {0, 1},
tick label style={font=\footnotesize},
ytick distance=1,
minor tick num=0,
axis line style=gray,
axis x line=middle,
axis y line=left,
tick style={black},
legend style={draw=none,fill=white,font=\small,at={(0.98,0.98)},anchor=north east},
xlabel={\footnotesize $\quant$},
]

\def\rstar{2.0}
\def\qstar{0.5}
\def\qdag{0.7}
\def\lnrstar{0.6931471805599453}
\def\scale{0.7142857142857143}

\definecolor{orangeLine}{RGB}{240,160,0}
\definecolor{blueLine}{RGB}{58,166,255}
\definecolor{greenLine}{RGB}{0,163,122}

\addplot[black!20, line width=1mm] coordinates {
    (0, 0)
    (0.75, 0.85714)
    (1, 1)
    };

\addplot[black, line width=0.3mm] coordinates {
    (0, 0)
    (0.875, 1)
    (1, 1)
    };

\end{axis}
\end{tikzpicture}     \caption{Illustration of \Cref{example:ratio monotonicity}. The black (gray) curve is revenue curve $\revcurve_1$ ($\revcurve_2$). }
    \label{fig:ratio monotonicity}
\end{figure} 
\section{Further Discussion of \texorpdfstring{\Cref{remark:worst truncation}}{Remark 3.2}}
\label{apx:worst case truncation discussion}
We begin with an illustrative example in \Cref{fig:bk vcg:mhr:VCG revenue}. 
Consider a balanced market under MHR distributions with market size $n=10$ and $\addBuyers=5$ additional buyers (the competition complexity for $n=10$ in \Cref{tab:additional-buyers-mhr}). 
With a slight abuse of notation, let $\RevVCG_{n:n+\addBuyers}(\optquant)$ denote the expected revenue of the {\VCGAuction} in the $n$-unit, $(n+\addBuyers)$-buyer market when buyer values are drawn from the ${1}/{\optquant}$-truncated exponential distribution.  
Our numerical calculations in \Cref{fig:bk vcg:mhr:VCG revenue} suggest that $\RevVCG_{n:n+\addBuyers}(\optquant)$ is not single-peaked: as $\optquant$ increases from $1/e$, the function first decreases, then increases, and finally decreases again.  
Consequently, the left-side local minimum of $\RevVCG_{n:n+\addBuyers}(\optquant)$ is not attained at the boundary point $\optquant=1/e$.

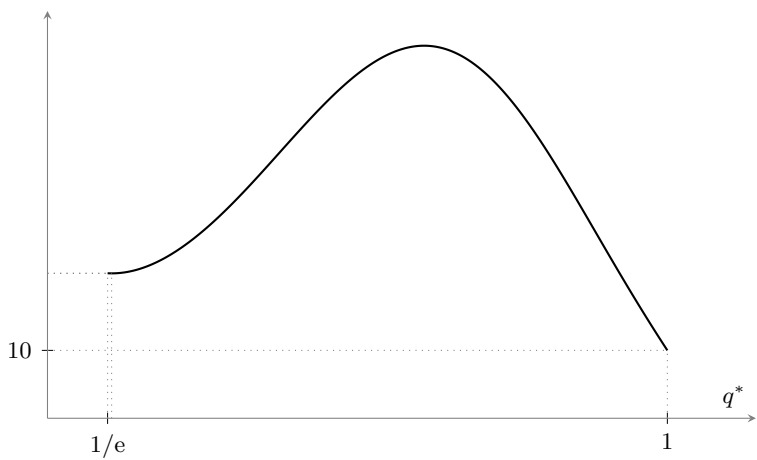
\begin{figure}
    \centering
\begin{tikzpicture}
 \begin{axis}[
width=11cm,
height=7cm,
xmin=0.3,xmax=1.1,
ymin=9.5,ymax=12.5,
xtick={0.36788, 1},   
xticklabels = {1/e, 1},
tick label style={font=\footnotesize}, 
scaled y ticks=false,     
yticklabel style={/pgf/number format/fixed}, 
ytick={10}, 
yticklabels = {10},
minor tick num=0,
axis line style=gray,
axis x line=middle,
axis y line=left,
tick style={black},
legend style={draw=none,fill=white,font=\small,at={(0.98,0.98)},anchor=north east},
xlabel={\footnotesize $\optquant$},
]
\addplot[smooth, thick] coordinates {
        (3.678804411714423073e-01,1.056814444687058696e+01)
        (3.724280567025829991e-01,1.056737789575261388e+01)
        (3.769756722337236909e-01,1.056804627619234260e+01)
        (3.815232877648644383e-01,1.057011830796789376e+01)
        (3.860709032960051301e-01,1.057356356623717630e+01)
        (3.906185188271458220e-01,1.057835237072577073e+01)
        (3.951661343582865138e-01,1.058445567973534196e+01)
        (3.997137498894272056e-01,1.059184498864036605e+01)
        (4.042613654205678975e-01,1.060049223258183382e+01)
        (4.088089809517086448e-01,1.061036969310391065e+01)
        (4.133565964828493366e-01,1.062144990851377813e+01)
        (4.179042120139900285e-01,1.063370558777619479e+01)
        (4.224518275451307203e-01,1.064710952778314024e+01)
        (4.269994430762714677e-01,1.066163453386533000e+01)
        (4.315470586074121595e-01,1.067725334343680288e+01)
        (4.360946741385528513e-01,1.069393855268613613e+01)
        (4.406422896696935432e-01,1.071166254624855085e+01)
        (4.451899052008342350e-01,1.073039742981216982e+01)
        (4.497375207319749268e-01,1.075011496562920144e+01)
        (4.542851362631156187e-01,1.077078651091892958e+01)
        (4.588327517942563660e-01,1.079238295916422175e+01)
        (4.633803673253970579e-01,1.081487468431679133e+01)
        (4.679279828565377497e-01,1.083823148793886837e+01)
        (4.724755983876784415e-01,1.086242254932024842e+01)
        (4.770232139188191889e-01,1.088741637861989453e+01)
        (4.815708294499598807e-01,1.091318077309049706e+01)
        (4.861184449811005726e-01,1.093968277645268472e+01)
        (4.906660605122412644e-01,1.096688864149276732e+01)
        (4.952136760433819562e-01,1.099476379596437958e+01)
        (4.997612915745226481e-01,1.102327281187976382e+01)
        (5.043089071056633399e-01,1.105237937828100492e+01)
        (5.088565226368040317e-01,1.108204627758526861e+01)
        (5.134041381679448346e-01,1.111223536560081726e+01)
        (5.179517536990854154e-01,1.114290755531259869e+01)
        (5.224993692302262183e-01,1.117402280453726959e+01)
        (5.270469847613669101e-01,1.120554010754768015e+01)
        (5.315946002925076019e-01,1.123741749076632068e+01)
        (5.361422158236482938e-01,1.126961201262563961e+01)
        (5.406898313547889856e-01,1.130207976769085754e+01)
        (5.452374468859296774e-01,1.133477589513763739e+01)
        (5.497850624170703693e-01,1.136765459167296299e+01)
        (5.543326779482110611e-01,1.140066912898256035e+01)
        (5.588802934793517529e-01,1.143377187578244047e+01)
        (5.634279090104925558e-01,1.146691432454545101e+01)
        (5.679755245416331366e-01,1.150004712296617804e+01)
        (5.725231400727739395e-01,1.153312011021912298e+01)
        (5.770707556039146313e-01,1.156608235805577678e+01)
        (5.816183711350553232e-01,1.159888221677611142e+01)
        (5.861659866661960150e-01,1.163146736609894205e+01)
        (5.907136021973367068e-01,1.166378487094377547e+01)
        (5.952612177284773987e-01,1.169578124212401349e+01)
        (5.998088332596180905e-01,1.172740250193785627e+01)
        (6.043564487907588934e-01,1.175859425462881092e+01)
        (6.089040643218994742e-01,1.178930176167254196e+01)
        (6.134516798530402770e-01,1.181947002183077799e+01)
        (6.179992953841808578e-01,1.184904385589617881e+01)
        (6.225469109153216607e-01,1.187796799603458631e+01)
        (6.270945264464623525e-01,1.190618717961270789e+01)
        (6.316421419776030444e-01,1.193364624738039304e+01)
        (6.361897575087437362e-01,1.196029024585694778e+01)
        (6.407373730398844280e-01,1.198606453375064085e+01)
        (6.452849885710251199e-01,1.201091489221973596e+01)
        (6.498326041021658117e-01,1.203478763876181823e+01)
        (6.543802196333066146e-01,1.205762974449639024e+01)
        (6.589278351644471954e-01,1.207938895458317852e+01)
        (6.634754506955879982e-01,1.210001391149592109e+01)
        (6.680230662267285791e-01,1.211945428084817067e+01)
        (6.725706817578693819e-01,1.213766087944441452e+01)
        (6.771182972890100737e-01,1.215458580520606446e+01)
        (6.816659128201507656e-01,1.217018256859837244e+01)
        (6.862135283512914574e-01,1.218440622516047078e+01)
        (6.907611438824321493e-01,1.219721350871720666e+01)
        (6.953087594135729521e-01,1.220856296482788572e+01)
        (6.998563749447135329e-01,1.221841508400389387e+01)
        (7.044039904758543358e-01,1.222673243420428513e+01)
        (7.089516060069949166e-01,1.223347979209611402e+01)
        (7.134992215381357195e-01,1.223862427254458751e+01)
        (7.180468370692763003e-01,1.224213545577709894e+01)
        (7.225944526004171031e-01,1.224398551164521010e+01)
        (7.271420681315577950e-01,1.224414932038953197e+01)
        (7.316896836626984868e-01,1.224260458929473572e+01)
        (7.362372991938391786e-01,1.223933196460535022e+01)
        (7.407849147249798705e-01,1.223431513805814674e+01)
        (7.453325302561206733e-01,1.222754094737370067e+01)
        (7.498801457872612541e-01,1.221899947003840126e+01)
        (7.544277613184020570e-01,1.220868410969904438e+01)
        (7.589753768495426378e-01,1.219659167448526915e+01)
        (7.635229923806834407e-01,1.218272244657085146e+01)
        (7.680706079118240215e-01,1.216708024228332974e+01)
        (7.726182234429648243e-01,1.214967246207299212e+01)
        (7.771658389741055162e-01,1.213051012965706121e+01)
        (7.817134545052462080e-01,1.210960791966330419e+01)
        (7.862610700363868999e-01,1.208698417310946738e+01)
        (7.908086855675275917e-01,1.206266090007132163e+01)
        (7.953563010986683945e-01,1.203666376891280443e+01)
        (7.999039166298089754e-01,1.200902208147730121e+01)
        (8.044515321609497782e-01,1.197976873366968320e+01)
        (8.089991476920903590e-01,1.194894016089469346e+01)
        (8.135467632232311619e-01,1.191657626785909940e+01)
        (8.180943787543718537e-01,1.188272034229293439e+01)
        (8.226419942855125456e-01,1.184741895219960917e+01)
        (8.271896098166532374e-01,1.181072182630611245e+01)
        (8.317372253477939292e-01,1.177268171745322967e+01)
        (8.362848408789346211e-01,1.173335424874231059e+01)
        (8.408324564100753129e-01,1.169279774233980618e+01)
        (8.453800719412161158e-01,1.165107303093438595e+01)
        (8.499276874723566966e-01,1.160824325194398554e+01)
        (8.544753030034974994e-01,1.156437362468255969e+01)
        (8.590229185346380802e-01,1.151953121081882259e+01)
        (8.635705340657788831e-01,1.147378465859247854e+01)
        (8.681181495969194639e-01,1.142720393139810398e+01)
        (8.726657651280602668e-01,1.137986002150310938e+01)
        (8.772133806592008476e-01,1.133182464983522841e+01)
        (8.817609961903416504e-01,1.128316995295672953e+01)
        (8.863086117214824533e-01,1.123396815853832820e+01)
        (8.908562272526230341e-01,1.118429125085559761e+01)
        (8.954038427837638370e-01,1.113421062805569228e+01)
        (8.999514583149044178e-01,1.108379675318297508e+01)
        (9.044990738460452206e-01,1.103311880120906174e+01)
        (9.090466893771858015e-01,1.098224430458729195e+01)
        (9.135943049083266043e-01,1.093123880014368510e+01)
        (9.181419204394671851e-01,1.088016548042748255e+01)
        (9.226895359706079880e-01,1.082908485297474321e+01)
        (9.272371515017487908e-01,1.077805441128938391e+01)
        (9.317847670328893717e-01,1.072712832171796116e+01)
        (9.363323825640301745e-01,1.067635713078872683e+01)
        (9.408799980951707553e-01,1.062578749800262656e+01)
        (9.454276136263115582e-01,1.057546195950492240e+01)
        (9.499752291574521390e-01,1.052541872853225158e+01)
        (9.545228446885929419e-01,1.047569153902171557e+01)
        (9.590704602197335227e-01,1.042630953928754778e+01)
        (9.636180757508743255e-01,1.037729724321756386e+01)
        (9.681656912820149063e-01,1.032867454701754717e+01)
        (9.727133068131557092e-01,1.028045682013763518e+01)
        (9.772609223442965121e-01,1.023265507965206389e+01)
        (9.818085378754370929e-01,1.018527625803331915e+01)
        (9.863561534065778957e-01,1.013832357496510461e+01)
        (9.909037689377184766e-01,1.009179702457677763e+01)
        (9.954513844688592794e-01,1.004569399025622189e+01)
        (9.999989999999999712e-01,1.000001000001000051e+01)
};

\addplot[gray, dotted] coordinates {
  (3.724280567025829991e-01,9.5)
  (3.724280567025829991e-01,1.056737789575261388e+01)
};

\addplot[gray, dotted] coordinates {
  (0.3,1.056737789575261388e+01)
  (3.724280567025829991e-01,1.056737789575261388e+01)
};

\addplot[gray, dotted] coordinates {
  (3.678804411714423073e-01,9.5)
  (3.678804411714423073e-01,1.056814444687058696e+01)
};

\addplot[gray, dotted] coordinates {
  (0.3,1.056814444687058696e+01)
  (3.678804411714423073e-01,1.056814444687058696e+01)
};
\addplot[gray, dotted] coordinates {
  (0.3,10)
  (1,10)
};

\addplot[gray, dotted] coordinates {
  (1,9.5)
  (1,10)
};
\end{axis}
\end{tikzpicture}
 \caption{The expected revenue of the {\VCGAuction} (of selling 10 units to 15 buyers whose values are drawn i.i.d.\ from $1/\optquant$-truncated exponential distribution) as a function of monopoly quantile $\optquant\in\left[{1}/{e},1\right]$. 
The local-minimum expected revenue of the {\VCGAuction} is not attained at $\optquant={1}/{e}$; instead, it is attained at some $\optquant$ in a right neighborhood of ${1}/{e}$ and at $\optquant=1$.
}

\label{fig:bk vcg:mhr:VCG revenue}
\end{figure}

We then show that, in balanced markets under MHR distributions, for every finite market size $n$ and the corresponding competition complexity $\addBuyers$,
the left-side local minimum of the expected revenue $\RevVCG_{n:n+\addBuyers}(\optquant)$ is not attained at the boundary point $\optquant=1/e$.
Indeed, for finite every market size $n$, at the monopoly quantile $\optquant=1/e$ we have
${\dd\RevVCG_{n:n+\addBuyers}(\optquant)}/{\dd \optquant}\le 0$.

\begin{claim}\label{clm:qStar-not-local-min}
In balanced markets under MHR distributions, for every market size $n$, the expected revenue $\RevVCG_{n:n+\addBuyers}(\optquant)$ is decreasing in a right neighborhood of $\optquant=1/e$.
\end{claim}
\begin{proof}[Proof of \Cref{clm:qStar-not-local-min}]
Define a random variable $X\sim\BetaDist{n+1}{\addBuyers}$ supported on $[0,1]$, with density
$p\left(\variablex\right)\triangleq ({1}/{\BetaFun{n+1}{\addBuyers}})\cdot \variablex^{n}\cdot\left(1-\variablex\right)^{\addBuyers-1}$.
Let $
\rprob \triangleq \prob{\left(X\ge \optquant\right)}
= \int_{\optquant}^{1} p\left(\variablex\right)\,\dd \variablex$
and
$\lnexpect \triangleq \expect{\left(\left(-\ln X\right)\cdot \IF_{\left\{X\ge \optquant\right\}}\right)}
= \int_{\optquant}^{1} \left(-\ln \variablex\right)\cdot p\left(\variablex\right)\,\dd \variablex$.
According to the result of \Cref{lem:VCG revenue derivative}, we know that
\[
\frac{\dd \ln\left(\RevVCG_{n:n+\addBuyers}\left(\optquant\right)\right)}{\dd \optquant}
=
-\frac{1}{\optquant}
+
\frac{1}{\optquant}\cdot
\frac{\lnexpect}{
\left(\ln \optquant\right)^2\cdot
\left(1-\rprob-\frac{\lnexpect}{\ln \optquant}\right)
}, 
\]
substituting the monopoly quantile $\optquant=1/e$ yields
\[
\frac{\dd \ln\left(\RevVCG_{n:n+\addBuyers}\left(\optquant\right)\right)}{\dd \optquant}
=
e\cdot
\frac{h(\frac{1}{e})-1}{
\left(1-p_{+}(\frac{1}{e})+h(\frac{1}{e})\right)
}\le 0
\]
the inequality holds since $1-p_{+}({1}/{e})\ge 0$ and $0\le h({1}/{e})\le -\ln{(1/e)}=1$. 

By \Cref{lem:VCG revenue derivative}, we know that ${\dd \ln\left(\RevVCG_{n:n+\addBuyers}\left(\optquant\right)\right)}/{\dd \optquant}$ and ${\dd \RevVCG_{n:n+\addBuyers}\left(\optquant\right)}/{\dd \optquant}$ have the same sign.
Therefore, the revenue $\RevVCG_{n:n+\addBuyers}\left(\optquant\right)$ is decreasing in a right neighborhood of $\optquant=1/e$. This complete the proof of \Cref{clm:qStar-not-local-min}.
\end{proof}

Finally, we show that the local minimum of the expected revenue
$\RevVCG_{n:n+\addBuyers}(\optquant)$, which is attained at some $\optquant$ in a right neighborhood of $1/e$,
does not admit a closed-form expression.
This local minimum is attained at a monopoly quantile $\optquant$ that satisfies
\[
\frac{\dd}{\dd \optquant}\RevVCG_{n:n+\addBuyers}(\optquant)=-\frac{1}{\optquant}
+
\frac{1}{\optquant}\cdot
\frac{\lnexpect}{
\left(\ln \optquant\right)^2\cdot
\left(1-\rprob-\frac{\lnexpect}{\ln \optquant}\right)
}=0,
\]
where $\optquant$ depends on $n$ and does not admit a closed-form expression in general. 

\section{Numerical Experiment}
\label{sec:numerical experiment}
In this section, we investigate the numerical behavior of the competition complexity for finite market sizes $n$, and study how fast the ratio between the competition complexity and $n$ converges as $n$ grows. We begin by computing the exact competition complexity for all $n \le 593$ in balanced markets under an MHR distribution when approximation ratio $\Gamma = 1$. The complete set of numerical results is reported in \Cref{tab:additional-buyers-mhr}.

To determine the minimum number of additional buyers needed for each market size $n \le 593$, for each fixed $n$ we iterate over $\addBuyers \in \{1,2,\ldots,n\}$ and check whether $\addBuyers$ additional buyers suffice under the worst-case distribution. Concretely, for a given $n$ and $\addBuyers$, we discretize the admissible range of the monopoly reserve $\optreserve \in [1,e]$ into a uniform grid with step size $0.01$ and evaluate the corresponding expected revenue of the {\VCGAuction} at each grid point. We then take the minimum value over this grid as an approximation of the minimum expected revenue of the {\VCGAuction} over all $\optreserve$ for the given $n$ and $\addBuyers$. The algorithm returns the smallest $\addBuyers$ for which this minimum expected revenue of the {\VCGAuction} is at least the Bayesian-optimal revenue $n$. Repeating this procedure for all $n\le593$ produces the table of minimal additional buyers $\addBuyers$ reported in the paper.

\newcommand{\BlockSep}{0pt}  
\newcommand{\ColWn}{8.5mm}       
\newcommand{\ColWk}{8.5mm}      
\newcommand{\RowH}{4mm}     

\newcommand{\ThickV}{\vrule width 0.4pt}
\newcommand{\ThinV}{\vrule width 0.4pt}
\newlength{\OldArrayRuleWidth}

\newcommand{\ThickH}{\noalign{\global\OldArrayRuleWidth=\arrayrulewidth}\noalign{\global\setlength{\arrayrulewidth}{0.6pt}}\cline{1-12}\noalign{\global\setlength{\arrayrulewidth}{\OldArrayRuleWidth}}}

\newcommand{\DoubleV}{\ThickV\hskip 2.9pt\ThickV}

\newcolumntype{C}[1]{>{\centering\arraybackslash\rule{0pt}{\RowH}}m{#1}}

\setlength{\LTleft}{0pt}
\setlength{\LTright}{0pt}

\begin{longtable}
{@{}@{\extracolsep{0pt}}
!{\ThickV}C{\ColWn}!{\ThinV}C{\ColWk}!{\DoubleV}C{\ColWn}!{\ThinV}C{\ColWk}!{\DoubleV}C{\ColWn}!{\ThinV}C{\ColWk}!{\DoubleV}C{\ColWn}!{\ThinV}C{\ColWk}!{\DoubleV}C{\ColWn}!{\ThinV}C{\ColWk}!{\DoubleV}C{\ColWn}!{\ThinV}C{\ColWk}!{\ThickV}@{}}

\caption{Competition Complexity in Balanced Markets under MHR Distributions ($n\le593$, $\ccapproxratio=1$)}
\label{tab:additional-buyers-mhr}\\

\ThickH
\bfseries $n$ & \bfseries $t_n$ &
\bfseries $n$ & \bfseries $t_n$ &
\bfseries $n$ & \bfseries $t_n$ &
\bfseries $n$ & \bfseries $t_n$ &
\bfseries $n$ & \bfseries $t_n$ &
\bfseries $n$ & \bfseries $t_n$ \\
\ThickH
\endfirsthead

\ThickH
\bfseries $n$ & \bfseries $t_n$ &
\bfseries $n$ & \bfseries $t_n$ &
\bfseries $n$ & \bfseries $t_n$ &
\bfseries $n$ & \bfseries $t_n$ &
\bfseries $n$ & \bfseries $t_n$ &
\bfseries $n$ & \bfseries $t_n$ \\
\ThickH
\endhead

\ThickH
\endfoot

\ThickH
\endlastfoot

1 & 1 & 101 & 46 & 201 & 90 & 301 & 135 & 401 & 179 & 501 & 224 \\
  2 & 2 & 102 & 46 & 202 & 91 & 302 & 135 & 402 & 179 & 502 & 224 \\
  3 & 2 & 103 & 47 & 203 & 91 & 303 & 135 & 403 & 180 & 503 & 224 \\
  4 & 3 & 104 & 47 & 204 & 91 & 304 & 136 & 404 & 180 & 504 & 225 \\
  5 & 3 & 105 & 47 & 205 & 92 & 305 & 136 & 405 & 181 & 505 & 225 \\
  6 & 3 & 106 & 48 & 206 & 92 & 306 & 137 & 406 & 181 & 506 & 226 \\
  7 & 4 & 107 & 48 & 207 & 93 & 307 & 137 & 407 & 182 & 507 & 226 \\
  8 & 4 & 108 & 49 & 208 & 93 & 308 & 138 & 408 & 182 & 508 & 227 \\
  9 & 5 & 109 & 49 & 209 & 94 & 309 & 138 & 409 & 183 & 509 & 227 \\
  10 & 5 & 110 & 50 & 210 & 94 & 310 & 139 & 410 & 183 & 510 & 228 \\
  11 & 6 & 111 & 50 & 211 & 95 & 311 & 139 & 411 & 183 & 511 & 228 \\
  12 & 6 & 112 & 51 & 212 & 95 & 312 & 139 & 412 & 184 & 512 & 228 \\
  13 & 7 & 113 & 51 & 213 & 95 & 313 & 140 & 413 & 184 & 513 & 229 \\
  14 & 7 & 114 & 51 & 214 & 96 & 314 & 140 & 414 & 185 & 514 & 229 \\
  15 & 7 & 115 & 52 & 215 & 96 & 315 & 141 & 415 & 185 & 515 & 230 \\
  16 & 8 & 116 & 52 & 216 & 97 & 316 & 141 & 416 & 186 & 516 & 230 \\
  17 & 8 & 117 & 53 & 217 & 97 & 317 & 142 & 417 & 186 & 517 & 231 \\
  18 & 9 & 118 & 53 & 218 & 98 & 318 & 142 & 418 & 187 & 518 & 231 \\
  19 & 9 & 119 & 54 & 219 & 98 & 319 & 143 & 419 & 187 & 519 & 232 \\
  20 & 10 & 120 & 54 & 220 & 99 & 320 & 143 & 420 & 187 & 520 & 232 \\
  21 & 10 & 121 & 55 & 221 & 99 & 321 & 143 & 421 & 188 & 521 & 232 \\
  22 & 11 & 122 & 55 & 222 & 99 & 322 & 144 & 422 & 188 & 522 & 233 \\
  23 & 11 & 123 & 55 & 223 & 100 & 323 & 144 & 423 & 189 & 523 & 233 \\
  24 & 11 & 124 & 56 & 224 & 100 & 324 & 145 & 424 & 189 & 524 & 234 \\
  25 & 12 & 125 & 56 & 225 & 101 & 325 & 145 & 425 & 190 & 525 & 234 \\
  26 & 12 & 126 & 57 & 226 & 101 & 326 & 146 & 426 & 190 & 526 & 235 \\
  27 & 13 & 127 & 57 & 227 & 102 & 327 & 146 & 427 & 191 & 527 & 235 \\
  28 & 13 & 128 & 58 & 228 & 102 & 328 & 147 & 428 & 191 & 528 & 236 \\
  29 & 14 & 129 & 58 & 229 & 103 & 329 & 147 & 429 & 191 & 529 & 236 \\
  30 & 14 & 130 & 59 & 230 & 103 & 330 & 147 & 430 & 192 & 530 & 236 \\
  31 & 15 & 131 & 59 & 231 & 103 & 331 & 148 & 431 & 192 & 531 & 237 \\
  32 & 15 & 132 & 59 & 232 & 104 & 332 & 148 & 432 & 193 & 532 & 237 \\
  33 & 15 & 133 & 60 & 233 & 104 & 333 & 149 & 433 & 193 & 533 & 238 \\
  34 & 16 & 134 & 60 & 234 & 105 & 334 & 149 & 434 & 194 & 534 & 238 \\
  35 & 16 & 135 & 61 & 235 & 105 & 335 & 150 & 435 & 194 & 535 & 239 \\
  36 & 17 & 136 & 61 & 236 & 106 & 336 & 150 & 436 & 195 & 536 & 239 \\
  37 & 17 & 137 & 62 & 237 & 106 & 337 & 151 & 437 & 195 & 537 & 240 \\
  38 & 18 & 138 & 62 & 238 & 107 & 338 & 151 & 438 & 195 & 538 & 240 \\
  39 & 18 & 139 & 63 & 239 & 107 & 339 & 151 & 439 & 196 & 539 & 240 \\
  40 & 19 & 140 & 63 & 240 & 107 & 340 & 152 & 440 & 196 & 540 & 241 \\
  41 & 19 & 141 & 63 & 241 & 108 & 341 & 152 & 441 & 197 & 541 & 241 \\
  42 & 19 & 142 & 64 & 242 & 108 & 342 & 153 & 442 & 197 & 542 & 242 \\
  43 & 20 & 143 & 64 & 243 & 109 & 343 & 153 & 443 & 198 & 543 & 242 \\
  44 & 20 & 144 & 65 & 244 & 109 & 344 & 154 & 444 & 198 & 544 & 243 \\
  45 & 21 & 145 & 65 & 245 & 110 & 345 & 154 & 445 & 199 & 545 & 243 \\
  46 & 21 & 146 & 66 & 246 & 110 & 346 & 155 & 446 & 199 & 546 & 244 \\
  47 & 22 & 147 & 66 & 247 & 111 & 347 & 155 & 447 & 199 & 547 & 244 \\
  48 & 22 & 148 & 67 & 248 & 111 & 348 & 155 & 448 & 200 & 548 & 244 \\
  49 & 23 & 149 & 67 & 249 & 111 & 349 & 156 & 449 & 200 & 549 & 245 \\
  50 & 23 & 150 & 67 & 250 & 112 & 350 & 156 & 450 & 201 & 550 & 245 \\
  51 & 23 & 151 & 68 & 251 & 112 & 351 & 157 & 451 & 201 & 551 & 246 \\
  52 & 24 & 152 & 68 & 252 & 113 & 352 & 157 & 452 & 202 & 552 & 246 \\
  53 & 24 & 153 & 69 & 253 & 113 & 353 & 158 & 453 & 202 & 553 & 247 \\
  54 & 25 & 154 & 69 & 254 & 114 & 354 & 158 & 454 & 203 & 554 & 247 \\
  55 & 25 & 155 & 70 & 255 & 114 & 355 & 159 & 455 & 203 & 555 & 248 \\
  56 & 26 & 156 & 70 & 256 & 115 & 356 & 159 & 456 & 203 & 556 & 248 \\
  57 & 26 & 157 & 71 & 257 & 115 & 357 & 159 & 457 & 204 & 557 & 248 \\
  58 & 27 & 158 & 71 & 258 & 115 & 358 & 160 & 458 & 204 & 558 & 249 \\
  59 & 27 & 159 & 71 & 259 & 116 & 359 & 160 & 459 & 205 & 559 & 249 \\
  60 & 27 & 160 & 72 & 260 & 116 & 360 & 161 & 460 & 205 & 560 & 250 \\
  61 & 28 & 161 & 72 & 261 & 117 & 361 & 161 & 461 & 206 & 561 & 250 \\
  62 & 28 & 162 & 73 & 262 & 117 & 362 & 162 & 462 & 206 & 562 & 251 \\
  63 & 29 & 163 & 73 & 263 & 118 & 363 & 162 & 463 & 207 & 563 & 251 \\
  64 & 29 & 164 & 74 & 264 & 118 & 364 & 163 & 464 & 207 & 564 & 252 \\
  65 & 30 & 165 & 74 & 265 & 119 & 365 & 163 & 465 & 207 & 565 & 252 \\
  66 & 30 & 166 & 75 & 266 & 119 & 366 & 163 & 466 & 208 & 566 & 252 \\
  67 & 31 & 167 & 75 & 267 & 119 & 367 & 164 & 467 & 208 & 567 & 253 \\
  68 & 31 & 168 & 75 & 268 & 120 & 368 & 164 & 468 & 209 & 568 & 253 \\
  69 & 31 & 169 & 76 & 269 & 120 & 369 & 165 & 469 & 209 & 569 & 254 \\
  70 & 32 & 170 & 76 & 270 & 121 & 370 & 165 & 470 & 210 & 570 & 254 \\
  71 & 32 & 171 & 77 & 271 & 121 & 371 & 166 & 471 & 210 & 571 & 255 \\
  72 & 33 & 172 & 77 & 272 & 122 & 372 & 166 & 472 & 211 & 572 & 255 \\
  73 & 33 & 173 & 78 & 273 & 122 & 373 & 167 & 473 & 211 & 573 & 256 \\
  74 & 34 & 174 & 78 & 274 & 123 & 374 & 167 & 474 & 211 & 574 & 256 \\
  75 & 34 & 175 & 79 & 275 & 123 & 375 & 167 & 475 & 212 & 575 & 256 \\
  76 & 35 & 176 & 79 & 276 & 123 & 376 & 168 & 476 & 212 & 576 & 257 \\
  77 & 35 & 177 & 79 & 277 & 124 & 377 & 168 & 477 & 213 & 577 & 257 \\
  78 & 35 & 178 & 80 & 278 & 124 & 378 & 169 & 478 & 213 & 578 & 258 \\
  79 & 36 & 179 & 80 & 279 & 125 & 379 & 169 & 479 & 214 & 579 & 258 \\
  80 & 36 & 180 & 81 & 280 & 125 & 380 & 170 & 480 & 214 & 580 & 259 \\
  81 & 37 & 181 & 81 & 281 & 126 & 381 & 170 & 481 & 215 & 581 & 259 \\
  82 & 37 & 182 & 82 & 282 & 126 & 382 & 171 & 482 & 215 & 582 & 260 \\
  83 & 38 & 183 & 82 & 283 & 127 & 383 & 171 & 483 & 215 & 583 & 260 \\
  84 & 38 & 184 & 83 & 284 & 127 & 384 & 171 & 484 & 216 & 584 & 260 \\
  85 & 39 & 185 & 83 & 285 & 127 & 385 & 172 & 485 & 216 & 585 & 261 \\
  86 & 39 & 186 & 83 & 286 & 128 & 386 & 172 & 486 & 217 & 586 & 261 \\
  87 & 39 & 187 & 84 & 287 & 128 & 387 & 173 & 487 & 217 & 587 & 262 \\
  88 & 40 & 188 & 84 & 288 & 129 & 388 & 173 & 488 & 218 & 588 & 262 \\
  89 & 40 & 189 & 85 & 289 & 129 & 389 & 174 & 489 & 218 & 589 & 263 \\
  90 & 41 & 190 & 85 & 290 & 130 & 390 & 174 & 490 & 219 & 590 & 263 \\
  91 & 41 & 191 & 86 & 291 & 130 & 391 & 175 & 491 & 219 & 591 & 264 \\
  92 & 42 & 192 & 86 & 292 & 131 & 392 & 175 & 492 & 219 & 592 & 264 \\
  93 & 42 & 193 & 87 & 293 & 131 & 393 & 175 & 493 & 220 & 593 & 264 \\
  94 & 43 & 194 & 87 & 294 & 131 & 394 & 176 & 494 & 220 &  &  \\
  95 & 43 & 195 & 87 & 295 & 132 & 395 & 176 & 495 & 221 &  &  \\
  96 & 43 & 196 & 88 & 296 & 132 & 396 & 177 & 496 & 221 &  &  \\
  97 & 44 & 197 & 88 & 297 & 133 & 397 & 177 & 497 & 222 &  &  \\
  98 & 44 & 198 & 89 & 298 & 133 & 398 & 178 & 498 & 222 &  &  \\
  99 & 45 & 199 & 89 & 299 & 134 & 399 & 178 & 499 & 223 &  &  \\
  100 & 45 & 200 & 90 & 300 & 134 & 400 & 179 & 500 & 223 &  &  \\

\end{longtable}

Next, we visualize the growth of the competition complexity reported in \Cref{tab:additional-buyers-mhr} as $n$ increases, as well as the convergence of the ratio between the competition complexity and $n$. The results in \Cref{fig:MHR distribution:balance market} show that as $n$ increases, the minimal number of additional buyers $\addBuyers$ is monotonically non-decreasing, and the ratio ${\addBuyers}/{n}$ quickly approaches asymptotic competition complexity $(e^{1/e}-1)\approx 0.445$.

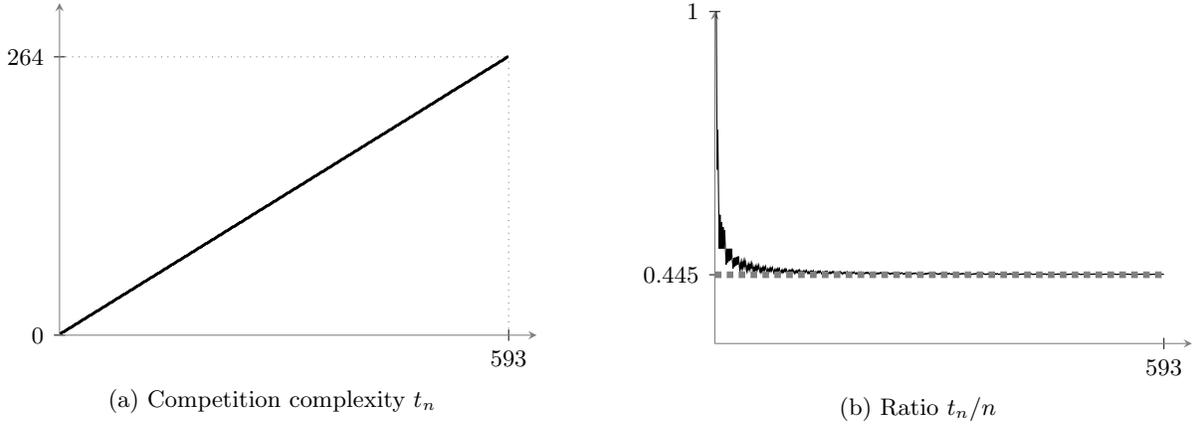
\begin{figure}
\centering
\begin{minipage}{0.48\textwidth}
  \centering
  \subfloat[Competition complexity $t_n$
  ]{
    \begin{tikzpicture}
  \begin{axis}[
width=\textwidth,
height=6cm,
xmin=0,xmax=630,
ymin=0,ymax=315,
xtick={0, 593},   
xticklabels = {0, 593},
tick label style={font=\footnotesize}, scaled y ticks=false,                 yticklabel style={/pgf/number format/fixed}, ytick={0, 264}, 
yticklabels = {0, 264},
minor tick num=0,
axis line style=gray,
axis x line=middle,
axis y line=left,
tick style={black},
legend style={draw=none,fill=white,font=\small,at={(0.98,0.98)},anchor=north east},
]

    \addplot[gray, dotted] coordinates {(593,0) (593,264)};

    \addplot[gray, dotted] coordinates {(0,264) (593,264)};

    \addplot[black, solid, line width=1pt] coordinates {
(1,1)
(2,2)
(3,2)
(4,3)
(5,3)
(6,3)
(7,4)
(8,4)
(9,5)
(10,5)
(11,6)
(12,6)
(13,7)
(14,7)
(15,7)
(16,8)
(17,8)
(18,9)
(19,9)
(20,10)
(21,10)
(22,11)
(23,11)
(24,11)
(25,12)
(26,12)
(27,13)
(28,13)
(29,14)
(30,14)
(31,15)
(32,15)
(33,15)
(34,16)
(35,16)
(36,17)
(37,17)
(38,18)
(39,18)
(40,19)
(41,19)
(42,19)
(43,20)
(44,20)
(45,21)
(46,21)
(47,22)
(48,22)
(49,23)
(50,23)
(51,23)
(52,24)
(53,24)
(54,25)
(55,25)
(56,26)
(57,26)
(58,27)
(59,27)
(60,27)
(61,28)
(62,28)
(63,29)
(64,29)
(65,30)
(66,30)
(67,31)
(68,31)
(69,31)
(70,32)
(71,32)
(72,33)
(73,33)
(74,34)
(75,34)
(76,35)
(77,35)
(78,35)
(79,36)
(80,36)
(81,37)
(82,37)
(83,38)
(84,38)
(85,39)
(86,39)
(87,39)
(88,40)
(89,40)
(90,41)
(91,41)
(92,42)
(93,42)
(94,43)
(95,43)
(96,43)
(97,44)
(98,44)
(99,45)
(100,45)
(101,46)
(102,46)
(103,47)
(104,47)
(105,47)
(106,48)
(107,48)
(108,49)
(109,49)
(110,50)
(111,50)
(112,51)
(113,51)
(114,51)
(115,52)
(116,52)
(117,53)
(118,53)
(119,54)
(120,54)
(121,55)
(122,55)
(123,55)
(124,56)
(125,56)
(126,57)
(127,57)
(128,58)
(129,58)
(130,59)
(131,59)
(132,59)
(133,60)
(134,60)
(135,61)
(136,61)
(137,62)
(138,62)
(139,63)
(140,63)
(141,63)
(142,64)
(143,64)
(144,65)
(145,65)
(146,66)
(147,66)
(148,67)
(149,67)
(150,67)
(151,68)
(152,68)
(153,69)
(154,69)
(155,70)
(156,70)
(157,71)
(158,71)
(159,71)
(160,72)
(161,72)
(162,73)
(163,73)
(164,74)
(165,74)
(166,75)
(167,75)
(168,75)
(169,76)
(170,76)
(171,77)
(172,77)
(173,78)
(174,78)
(175,79)
(176,79)
(177,79)
(178,80)
(179,80)
(180,81)
(181,81)
(182,82)
(183,82)
(184,83)
(185,83)
(186,83)
(187,84)
(188,84)
(189,85)
(190,85)
(191,86)
(192,86)
(193,87)
(194,87)
(195,87)
(196,88)
(197,88)
(198,89)
(199,89)
(200,90)
(201,90)
(202,91)
(203,91)
(204,91)
(205,92)
(206,92)
(207,93)
(208,93)
(209,94)
(210,94)
(211,95)
(212,95)
(213,95)
(214,96)
(215,96)
(216,97)
(217,97)
(218,98)
(219,98)
(220,99)
(221,99)
(222,99)
(223,100)
(224,100)
(225,101)
(226,101)
(227,102)
(228,102)
(229,103)
(230,103)
(231,103)
(232,104)
(233,104)
(234,105)
(235,105)
(236,106)
(237,106)
(238,107)
(239,107)
(240,107)
(241,108)
(242,108)
(243,109)
(244,109)
(245,110)
(246,110)
(247,111)
(248,111)
(249,111)
(250,112)
(251,112)
(252,113)
(253,113)
(254,114)
(255,114)
(256,115)
(257,115)
(258,115)
(259,116)
(260,116)
(261,117)
(262,117)
(263,118)
(264,118)
(265,119)
(266,119)
(267,119)
(268,120)
(269,120)
(270,121)
(271,121)
(272,122)
(273,122)
(274,123)
(275,123)
(276,123)
(277,124)
(278,124)
(279,125)
(280,125)
(281,126)
(282,126)
(283,127)
(284,127)
(285,127)
(286,128)
(287,128)
(288,129)
(289,129)
(290,130)
(291,130)
(292,131)
(293,131)
(294,131)
(295,132)
(296,132)
(297,133)
(298,133)
(299,134)
(300,134)
(301,135)
(302,135)
(303,135)
(304,136)
(305,136)
(306,137)
(307,137)
(308,138)
(309,138)
(310,139)
(311,139)
(312,139)
(313,140)
(314,140)
(315,141)
(316,141)
(317,142)
(318,142)
(319,143)
(320,143)
(321,143)
(322,144)
(323,144)
(324,145)
(325,145)
(326,146)
(327,146)
(328,147)
(329,147)
(330,147)
(331,148)
(332,148)
(333,149)
(334,149)
(335,150)
(336,150)
(337,151)
(338,151)
(339,151)
(340,152)
(341,152)
(342,153)
(343,153)
(344,154)
(345,154)
(346,155)
(347,155)
(348,155)
(349,156)
(350,156)
(351,157)
(352,157)
(353,158)
(354,158)
(355,159)
(356,159)
(357,159)
(358,160)
(359,160)
(360,161)
(361,161)
(362,162)
(363,162)
(364,163)
(365,163)
(366,163)
(367,164)
(368,164)
(369,165)
(370,165)
(371,166)
(372,166)
(373,167)
(374,167)
(375,167)
(376,168)
(377,168)
(378,169)
(379,169)
(380,170)
(381,170)
(382,171)
(383,171)
(384,171)
(385,172)
(386,172)
(387,173)
(388,173)
(389,174)
(390,174)
(391,175)
(392,175)
(393,175)
(394,176)
(395,176)
(396,177)
(397,177)
(398,178)
(399,178)
(400,179)
(401,179)
(402,179)
(403,180)
(404,180)
(405,181)
(406,181)
(407,182)
(408,182)
(409,183)
(410,183)
(411,183)
(412,184)
(413,184)
(414,185)
(415,185)
(416,186)
(417,186)
(418,187)
(419,187)
(420,187)
(421,188)
(422,188)
(423,189)
(424,189)
(425,190)
(426,190)
(427,191)
(428,191)
(429,191)
(430,192)
(431,192)
(432,193)
(433,193)
(434,194)
(435,194)
(436,195)
(437,195)
(438,195)
(439,196)
(440,196)
(441,197)
(442,197)
(443,198)
(444,198)
(445,199)
(446,199)
(447,199)
(448,200)
(449,200)
(450,201)
(451,201)
(452,202)
(453,202)
(454,203)
(455,203)
(456,203)
(457,204)
(458,204)
(459,205)
(460,205)
(461,206)
(462,206)
(463,207)
(464,207)
(465,207)
(466,208)
(467,208)
(468,209)
(469,209)
(470,210)
(471,210)
(472,211)
(473,211)
(474,211)
(475,212)
(476,212)
(477,213)
(478,213)
(479,214)
(480,214)
(481,215)
(482,215)
(483,215)
(484,216)
(485,216)
(486,217)
(487,217)
(488,218)
(489,218)
(490,219)
(491,219)
(492,219)
(493,220)
(494,220)
(495,221)
(496,221)
(497,222)
(498,222)
(499,223)
(500,223)
(501,224)
(502,224)
(503,224)
(504,225)
(505,225)
(506,226)
(507,226)
(508,227)
(509,227)
(510,228)
(511,228)
(512,228)
(513,229)
(514,229)
(515,230)
(516,230)
(517,231)
(518,231)
(519,232)
(520,232)
(521,232)
(522,233)
(523,233)
(524,234)
(525,234)
(526,235)
(527,235)
(528,236)
(529,236)
(530,236)
(531,237)
(532,237)
(533,238)
(534,238)
(535,239)
(536,239)
(537,240)
(538,240)
(539,240)
(540,241)
(541,241)
(542,242)
(543,242)
(544,243)
(545,243)
(546,244)
(547,244)
(548,244)
(549,245)
(550,245)
(551,246)
(552,246)
(553,247)
(554,247)
(555,248)
(556,248)
(557,248)
(558,249)
(559,249)
(560,250)
(561,250)
(562,251)
(563,251)
(564,252)
(565,252)
(566,252)
(567,253)
(568,253)
(569,254)
(570,254)
(571,255)
(572,255)
(573,256)
(574,256)
(575,256)
(576,257)
(577,257)
(578,258)
(579,258)
(580,259)
(581,259)
(582,260)
(583,260)
(584,260)
(585,261)
(586,261)
(587,262)
(588,262)
(589,263)
(590,263)
(591,264)
(592,264)
(593,264)
    };
\end{axis}
\end{tikzpicture}     \label{fig:MHR:balanced:minnum}
  }
\end{minipage}
\hfill
\begin{minipage}{0.48\textwidth}
  \centering
  \subfloat[Ratio $\addBuyers/n$
  ]{
    \begin{tikzpicture}
  \begin{axis}[
width=\textwidth,
height=6cm,
xmin=0,xmax=630,
ymin=0.3,ymax=1,
xtick={0, 593},   
xticklabels = {0, 593},
tick label style={font=\footnotesize}, scaled y ticks=false,                 yticklabel style={/pgf/number format/fixed}, ytick={0, 0.445,1}, 
yticklabels = {0, 0.445,1},
minor tick num=0,
axis line style=gray,
axis x line=middle,
axis y line=left,
tick style={black},
legend style={draw=none,fill=white,font=\small,at={(0.98,0.98)},anchor=north east},
]
\addplot[black, solid] coordinates {
(1,1.0)
(2,1.0)
(3,0.6666666666666666)
(4,0.75)
(5,0.6)
(6,0.5)
(7,0.5714285714285714)
(8,0.5)
(9,0.5555555555555556)
(10,0.5)
(11,0.5454545454545454)
(12,0.5)
(13,0.5384615384615384)
(14,0.5)
(15,0.4666666666666667)
(16,0.5)
(17,0.47058823529411764)
(18,0.5)
(19,0.47368421052631576)
(20,0.5)
(21,0.47619047619047616)
(22,0.5)
(23,0.4782608695652174)
(24,0.4583333333333333)
(25,0.48)
(26,0.46153846153846156)
(27,0.48148148148148145)
(28,0.4642857142857143)
(29,0.4827586206896552)
(30,0.4666666666666667)
(31,0.4838709677419355)
(32,0.46875)
(33,0.45454545454545453)
(34,0.47058823529411764)
(35,0.45714285714285713)
(36,0.4722222222222222)
(37,0.4594594594594595)
(38,0.47368421052631576)
(39,0.46153846153846156)
(40,0.475)
(41,0.4634146341463415)
(42,0.4523809523809524)
(43,0.46511627906976744)
(44,0.45454545454545453)
(45,0.4666666666666667)
(46,0.45652173913043476)
(47,0.46808510638297873)
(48,0.4583333333333333)
(49,0.46938775510204084)
(50,0.46)
(51,0.45098039215686275)
(52,0.46153846153846156)
(53,0.4528301886792453)
(54,0.46296296296296297)
(55,0.45454545454545453)
(56,0.4642857142857143)
(57,0.45614035087719296)
(58,0.46551724137931033)
(59,0.4576271186440678)
(60,0.45)
(61,0.45901639344262296)
(62,0.45161290322580644)
(63,0.4603174603174603)
(64,0.453125)
(65,0.46153846153846156)
(66,0.45454545454545453)
(67,0.4626865671641791)
(68,0.45588235294117646)
(69,0.4492753623188406)
(70,0.45714285714285713)
(71,0.4507042253521127)
(72,0.4583333333333333)
(73,0.4520547945205479)
(74,0.4594594594594595)
(75,0.4533333333333333)
(76,0.4605263157894737)
(77,0.45454545454545453)
(78,0.44871794871794873)
(79,0.45569620253164556)
(80,0.45)
(81,0.4567901234567901)
(82,0.45121951219512196)
(83,0.4578313253012048)
(84,0.4523809523809524)
(85,0.4588235294117647)
(86,0.45348837209302323)
(87,0.4482758620689655)
(88,0.45454545454545453)
(89,0.449438202247191)
(90,0.45555555555555555)
(91,0.45054945054945056)
(92,0.45652173913043476)
(93,0.45161290322580644)
(94,0.4574468085106383)
(95,0.45263157894736844)
(96,0.4479166666666667)
(97,0.4536082474226804)
(98,0.4489795918367347)
(99,0.45454545454545453)
(100,0.45)
(101,0.45544554455445546)
(102,0.45098039215686275)
(103,0.4563106796116505)
(104,0.4519230769230769)
(105,0.44761904761904764)
(106,0.4528301886792453)
(107,0.4485981308411215)
(108,0.4537037037037037)
(109,0.44954128440366975)
(110,0.45454545454545453)
(111,0.45045045045045046)
(112,0.45535714285714285)
(113,0.45132743362831856)
(114,0.4473684210526316)
(115,0.45217391304347826)
(116,0.4482758620689655)
(117,0.452991452991453)
(118,0.4491525423728814)
(119,0.453781512605042)
(120,0.45)
(121,0.45454545454545453)
(122,0.45081967213114754)
(123,0.44715447154471544)
(124,0.45161290322580644)
(125,0.448)
(126,0.4523809523809524)
(127,0.44881889763779526)
(128,0.453125)
(129,0.4496124031007752)
(130,0.45384615384615384)
(131,0.45038167938931295)
(132,0.44696969696969696)
(133,0.45112781954887216)
(134,0.44776119402985076)
(135,0.45185185185185184)
(136,0.4485294117647059)
(137,0.45255474452554745)
(138,0.4492753623188406)
(139,0.45323741007194246)
(140,0.45)
(141,0.44680851063829785)
(142,0.4507042253521127)
(143,0.44755244755244755)
(144,0.4513888888888889)
(145,0.4482758620689655)
(146,0.4520547945205479)
(147,0.4489795918367347)
(148,0.4527027027027027)
(149,0.44966442953020136)
(150,0.44666666666666666)
(151,0.4503311258278146)
(152,0.4473684210526316)
(153,0.45098039215686275)
(154,0.44805194805194803)
(155,0.45161290322580644)
(156,0.44871794871794873)
(157,0.45222929936305734)
(158,0.44936708860759494)
(159,0.44654088050314467)
(160,0.45)
(161,0.4472049689440994)
(162,0.4506172839506173)
(163,0.44785276073619634)
(164,0.45121951219512196)
(165,0.4484848484848485)
(166,0.45180722891566266)
(167,0.4491017964071856)
(168,0.44642857142857145)
(169,0.44970414201183434)
(170,0.4470588235294118)
(171,0.4502923976608187)
(172,0.4476744186046512)
(173,0.4508670520231214)
(174,0.4482758620689655)
(175,0.4514285714285714)
(176,0.44886363636363635)
(177,0.4463276836158192)
(178,0.449438202247191)
(179,0.44692737430167595)
(180,0.45)
(181,0.44751381215469616)
(182,0.45054945054945056)
(183,0.44808743169398907)
(184,0.45108695652173914)
(185,0.4486486486486487)
(186,0.44623655913978494)
(187,0.44919786096256686)
(188,0.44680851063829785)
(189,0.4497354497354497)
(190,0.4473684210526316)
(191,0.450261780104712)
(192,0.4479166666666667)
(193,0.45077720207253885)
(194,0.4484536082474227)
(195,0.4461538461538462)
(196,0.4489795918367347)
(197,0.4467005076142132)
(198,0.4494949494949495)
(199,0.4472361809045226)
(200,0.45)
(201,0.44776119402985076)
(202,0.4504950495049505)
(203,0.4482758620689655)
(204,0.44607843137254904)
(205,0.44878048780487806)
(206,0.44660194174757284)
(207,0.4492753623188406)
(208,0.44711538461538464)
(209,0.44976076555023925)
(210,0.44761904761904764)
(211,0.45023696682464454)
(212,0.4481132075471698)
(213,0.4460093896713615)
(214,0.4485981308411215)
(215,0.44651162790697674)
(216,0.44907407407407407)
(217,0.4470046082949309)
(218,0.44954128440366975)
(219,0.4474885844748858)
(220,0.45)
(221,0.4479638009049774)
(222,0.44594594594594594)
(223,0.4484304932735426)
(224,0.44642857142857145)
(225,0.4488888888888889)
(226,0.4469026548672566)
(227,0.44933920704845814)
(228,0.4473684210526316)
(229,0.4497816593886463)
(230,0.44782608695652176)
(231,0.4458874458874459)
(232,0.4482758620689655)
(233,0.44635193133047213)
(234,0.44871794871794873)
(235,0.44680851063829785)
(236,0.4491525423728814)
(237,0.4472573839662447)
(238,0.4495798319327731)
(239,0.4476987447698745)
(240,0.44583333333333336)
(241,0.44813278008298757)
(242,0.4462809917355372)
(243,0.448559670781893)
(244,0.44672131147540983)
(245,0.4489795918367347)
(246,0.44715447154471544)
(247,0.4493927125506073)
(248,0.4475806451612903)
(249,0.4457831325301205)
(250,0.448)
(251,0.44621513944223107)
(252,0.44841269841269843)
(253,0.44664031620553357)
(254,0.44881889763779526)
(255,0.4470588235294118)
(256,0.44921875)
(257,0.4474708171206226)
(258,0.44573643410852715)
(259,0.44787644787644787)
(260,0.4461538461538462)
(261,0.4482758620689655)
(262,0.44656488549618323)
(263,0.44866920152091255)
(264,0.44696969696969696)
(265,0.4490566037735849)
(266,0.4473684210526316)
(267,0.44569288389513106)
(268,0.44776119402985076)
(269,0.44609665427509293)
(270,0.44814814814814813)
(271,0.44649446494464945)
(272,0.4485294117647059)
(273,0.4468864468864469)
(274,0.4489051094890511)
(275,0.44727272727272727)
(276,0.44565217391304346)
(277,0.44765342960288806)
(278,0.4460431654676259)
(279,0.44802867383512546)
(280,0.44642857142857145)
(281,0.4483985765124555)
(282,0.44680851063829785)
(283,0.44876325088339225)
(284,0.4471830985915493)
(285,0.4456140350877193)
(286,0.44755244755244755)
(287,0.445993031358885)
(288,0.4479166666666667)
(289,0.4463667820069204)
(290,0.4482758620689655)
(291,0.44673539518900346)
(292,0.4486301369863014)
(293,0.447098976109215)
(294,0.445578231292517)
(295,0.44745762711864406)
(296,0.44594594594594594)
(297,0.4478114478114478)
(298,0.4463087248322148)
(299,0.44816053511705684)
(300,0.44666666666666666)
(301,0.4485049833887043)
(302,0.4470198675496689)
(303,0.44554455445544555)
(304,0.4473684210526316)
(305,0.4459016393442623)
(306,0.4477124183006536)
(307,0.44625407166123776)
(308,0.44805194805194803)
(309,0.44660194174757284)
(310,0.4483870967741935)
(311,0.44694533762057875)
(312,0.44551282051282054)
(313,0.4472843450479233)
(314,0.445859872611465)
(315,0.44761904761904764)
(316,0.4462025316455696)
(317,0.4479495268138801)
(318,0.44654088050314467)
(319,0.4482758620689655)
(320,0.446875)
(321,0.4454828660436137)
(322,0.4472049689440994)
(323,0.4458204334365325)
(324,0.44753086419753085)
(325,0.4461538461538462)
(326,0.44785276073619634)
(327,0.44648318042813456)
(328,0.4481707317073171)
(329,0.44680851063829785)
(330,0.44545454545454544)
(331,0.4471299093655589)
(332,0.4457831325301205)
(333,0.44744744744744747)
(334,0.44610778443113774)
(335,0.44776119402985076)
(336,0.44642857142857145)
(337,0.44807121661721067)
(338,0.4467455621301775)
(339,0.44542772861356933)
(340,0.4470588235294118)
(341,0.44574780058651026)
(342,0.4473684210526316)
(343,0.446064139941691)
(344,0.4476744186046512)
(345,0.4463768115942029)
(346,0.4479768786127168)
(347,0.44668587896253603)
(348,0.4454022988505747)
(349,0.4469914040114613)
(350,0.44571428571428573)
(351,0.4472934472934473)
(352,0.4460227272727273)
(353,0.4475920679886686)
(354,0.4463276836158192)
(355,0.447887323943662)
(356,0.44662921348314605)
(357,0.44537815126050423)
(358,0.44692737430167595)
(359,0.4456824512534819)
(360,0.44722222222222224)
(361,0.44598337950138506)
(362,0.44751381215469616)
(363,0.4462809917355372)
(364,0.4478021978021978)
(365,0.4465753424657534)
(366,0.4453551912568306)
(367,0.44686648501362397)
(368,0.44565217391304346)
(369,0.44715447154471544)
(370,0.44594594594594594)
(371,0.4474393530997305)
(372,0.44623655913978494)
(373,0.4477211796246649)
(374,0.446524064171123)
(375,0.44533333333333336)
(376,0.44680851063829785)
(377,0.44562334217506633)
(378,0.4470899470899471)
(379,0.44591029023746703)
(380,0.4473684210526316)
(381,0.4461942257217848)
(382,0.4476439790575916)
(383,0.4464751958224543)
(384,0.4453125)
(385,0.44675324675324674)
(386,0.44559585492227977)
(387,0.4470284237726098)
(388,0.44587628865979384)
(389,0.4473007712082262)
(390,0.4461538461538462)
(391,0.4475703324808184)
(392,0.44642857142857145)
(393,0.44529262086513993)
(394,0.4467005076142132)
(395,0.44556962025316454)
(396,0.44696969696969696)
(397,0.44584382871536526)
(398,0.4472361809045226)
(399,0.44611528822055135)
(400,0.4475)
(401,0.4463840399002494)
(402,0.44527363184079605)
(403,0.4466501240694789)
(404,0.44554455445544555)
(405,0.4469135802469136)
(406,0.4458128078817734)
(407,0.44717444717444715)
(408,0.44607843137254904)
(409,0.4474327628361858)
(410,0.44634146341463415)
(411,0.44525547445255476)
(412,0.44660194174757284)
(413,0.44552058111380144)
(414,0.4468599033816425)
(415,0.4457831325301205)
(416,0.44711538461538464)
(417,0.4460431654676259)
(418,0.4473684210526316)
(419,0.44630071599045346)
(420,0.4452380952380952)
(421,0.44655581947743467)
(422,0.44549763033175355)
(423,0.44680851063829785)
(424,0.44575471698113206)
(425,0.4470588235294118)
(426,0.4460093896713615)
(427,0.44730679156908665)
(428,0.4462616822429907)
(429,0.44522144522144524)
(430,0.44651162790697674)
(431,0.44547563805104406)
(432,0.44675925925925924)
(433,0.4457274826789838)
(434,0.4470046082949309)
(435,0.4459770114942529)
(436,0.44724770642201833)
(437,0.4462242562929062)
(438,0.4452054794520548)
(439,0.44646924829157175)
(440,0.44545454545454544)
(441,0.4467120181405896)
(442,0.4457013574660634)
(443,0.4469525959367946)
(444,0.44594594594594594)
(445,0.44719101123595506)
(446,0.4461883408071749)
(447,0.4451901565995526)
(448,0.44642857142857145)
(449,0.44543429844098)
(450,0.44666666666666666)
(451,0.44567627494456763)
(452,0.4469026548672566)
(453,0.445916114790287)
(454,0.44713656387665196)
(455,0.4461538461538462)
(456,0.4451754385964912)
(457,0.44638949671772427)
(458,0.44541484716157204)
(459,0.4466230936819172)
(460,0.44565217391304346)
(461,0.44685466377440347)
(462,0.4458874458874459)
(463,0.4470842332613391)
(464,0.44612068965517243)
(465,0.44516129032258067)
(466,0.44635193133047213)
(467,0.44539614561027835)
(468,0.4465811965811966)
(469,0.44562899786780386)
(470,0.44680851063829785)
(471,0.445859872611465)
(472,0.4470338983050847)
(473,0.44608879492600423)
(474,0.4451476793248945)
(475,0.4463157894736842)
(476,0.44537815126050423)
(477,0.44654088050314467)
(478,0.4456066945606695)
(479,0.44676409185803756)
(480,0.44583333333333336)
(481,0.446985446985447)
(482,0.4460580912863071)
(483,0.4451345755693582)
(484,0.4462809917355372)
(485,0.44536082474226807)
(486,0.44650205761316875)
(487,0.4455852156057495)
(488,0.44672131147540983)
(489,0.4458077709611452)
(490,0.44693877551020406)
(491,0.4460285132382892)
(492,0.4451219512195122)
(493,0.4462474645030426)
(494,0.44534412955465585)
(495,0.44646464646464645)
(496,0.44556451612903225)
(497,0.44668008048289737)
(498,0.4457831325301205)
(499,0.4468937875751503)
(500,0.446)
(501,0.4471057884231537)
(502,0.44621513944223107)
(503,0.44532803180914515)
(504,0.44642857142857145)
(505,0.44554455445544555)
(506,0.44664031620553357)
(507,0.4457593688362919)
(508,0.4468503937007874)
(509,0.44597249508840864)
(510,0.4470588235294118)
(511,0.4461839530332681)
(512,0.4453125)
(513,0.44639376218323584)
(514,0.4455252918287938)
(515,0.44660194174757284)
(516,0.44573643410852715)
(517,0.44680851063829785)
(518,0.44594594594594594)
(519,0.44701348747591524)
(520,0.4461538461538462)
(521,0.44529750479846447)
(522,0.446360153256705)
(523,0.44550669216061184)
(524,0.44656488549618323)
(525,0.44571428571428573)
(526,0.4467680608365019)
(527,0.4459203036053131)
(528,0.44696969696969696)
(529,0.44612476370510395)
(530,0.44528301886792454)
(531,0.4463276836158192)
(532,0.44548872180451127)
(533,0.44652908067542213)
(534,0.44569288389513106)
(535,0.44672897196261685)
(536,0.4458955223880597)
(537,0.44692737430167595)
(538,0.44609665427509293)
(539,0.4452690166975881)
(540,0.4462962962962963)
(541,0.4454713493530499)
(542,0.44649446494464945)
(543,0.44567219152854515)
(544,0.44669117647058826)
(545,0.44587155963302755)
(546,0.4468864468864469)
(547,0.4460694698354662)
(548,0.44525547445255476)
(549,0.44626593806921677)
(550,0.44545454545454544)
(551,0.44646098003629764)
(552,0.44565217391304346)
(553,0.44665461121157324)
(554,0.44584837545126355)
(555,0.44684684684684683)
(556,0.4460431654676259)
(557,0.4452423698384201)
(558,0.44623655913978494)
(559,0.44543828264758495)
(560,0.44642857142857145)
(561,0.44563279857397503)
(562,0.44661921708185054)
(563,0.44582593250444047)
(564,0.44680851063829785)
(565,0.44601769911504424)
(566,0.4452296819787986)
(567,0.4462081128747795)
(568,0.4454225352112676)
(569,0.44639718804920914)
(570,0.4456140350877193)
(571,0.44658493870402804)
(572,0.4458041958041958)
(573,0.4467713787085515)
(574,0.445993031358885)
(575,0.44521739130434784)
(576,0.4461805555555556)
(577,0.44540727902946275)
(578,0.4463667820069204)
(579,0.44559585492227977)
(580,0.44655172413793104)
(581,0.4457831325301205)
(582,0.44673539518900346)
(583,0.44596912521440824)
(584,0.4452054794520548)
(585,0.4461538461538462)
(586,0.4453924914675768)
(587,0.4463373083475298)
(588,0.445578231292517)
(589,0.4465195246179966)
(590,0.4457627118644068)
(591,0.4467005076142132)
(592,0.44594594594594594)
(593,0.4451939291736931)
};

    \addplot[gray, dotted, line width=2.5pt] coordinates {(0,0.445) (593,0.445)};
\node at (axis cs:0,0.445) [left,  xshift=2pt, black!20, text=black] {$0.445$};
  \end{axis}
\end{tikzpicture}     \label{fig:MHR:balanced:minratio}
  }
\end{minipage}

\caption{Competition complexity for balanced markets under MHR distributions (to achieve $\ccapproxratio=1$).
The black solid curve in Figure~(a) plots the competition complexity $\addBuyers$ as a function of the market size $n$.
The black solid curve in Figure (b) plots the ratio of the competition complexity $\addBuyers$ to $n$, as a function of $n$, and the light-gray dotted line marks the asymptotic competition complexity $0.445$.
}
\label{fig:MHR distribution:balance market}
\end{figure}

Next, we present numerical results for general markets to provide further insights into the competition complexity beyond the balanced-market setting. 
We consider general markets under MHR distributions in \Cref{fig:MHR distribution:general market} and general markets under regular distributions in \Cref{fig:regular distribution:general market}, with market size $n$ and supply-to-demand ratio $\imbalanceratio$, where the number of goods is $m_n=\lceil \imbalanceratio\cdot n\rceil$.
Let the competition complexity be denoted by $\addBuyers$.

\Cref{fig:MHR distribution:general market} illustrates both the growth of the competition complexity $\addBuyers$ as $n$ increases and the convergence of the ratio $\addBuyers/n$ under MHR distributions when the approximation ratio is $\ccapproxratio=1$.
We let $n$ range from $4$ to $400$ with step size $4$, and set $\imbalanceratio \in \{0.5, 0.75\}$. 
The results show that when $\imbalanceratio = 0.5$ or $\imbalanceratio = 0.75$, the competition complexity $\addBuyers$ is monotonically non-decreasing in $n$. 
Moreover, as $n$ grows, the ratio $\addBuyers/n$ gradually converges to the asymptotic competition complexity
$(\imbalanceratio\cdot e^{\frac{1}{e\imbalanceratio}} - 1)$ when $\imbalanceratio\geq {1}/{e}$.

\begin{figure}
\centering
\begin{minipage}{0.48\textwidth}
\centering
\subfloat[Competition complexity $\addBuyers$
  ]{
    \begin{tikzpicture}
  \begin{axis}[
width=\textwidth,
height=6cm,
xmin=0,xmax=420,
ymin=5,ymax=100,
xtick={0, 400},   
xticklabels = {0, 400},
tick label style={font=\footnotesize}, scaled y ticks=false,                 yticklabel style={/pgf/number format/fixed}, ytick={5, 23, 96}, 
yticklabels = {0, 18, 91},
minor tick num=0,
axis line style=gray,
axis x line=middle,
axis y line=left,
tick style={black},
legend style={draw=none,fill=white,font=\small,at={(0.98,0.98)},anchor=north east},
]

\pgfplotsset{
  myLine/.style={line width=1pt}
}

\addplot[black!45, myLine] coordinates {(4,6) (8,6) (12,7) (16,7) (20,7) (24,7) (28,7) (32,8) (36,8) (40,8) (44,8) (48,8) (52,8) (56,9) (60,9) (64,9) (68,9) (72,9) (76,9) (80,10) (84,10) (88,10) (92,10) (96,10) (100,10) (104,11) (108,11) (112,11) (116,11) (120,11) (124,11) (128,12) (132,12) (136,12) (140,12) (144,12) (148,12) (152,13) (156,13) (160,13) (164,13) (168,13) (172,14) (176,14) (180,14) (184,14) (188,14) (192,14) (196,15) (200,15) (204,15) (208,15) (212,15) (216,15) (220,16) (224,16) (228,16) (232,16) (236,16) (240,16) (244,17) (248,17) (252,17) (256,17) (260,17) (264,18) (268,18) (272,18) (276,18) (280,18) (284,18) (288,19) (292,19) (296,19) (300,19) (304,19) (308,19) (312,20) (316,20) (320,20) (324,20) (328,20) (332,20) (336,21) (340,21) (344,21) (348,21) (352,21) (356,22) (360,22) (364,22) (368,22) (372,22) (376,22) (380,23) (384,23) (388,23) (392,23) (396,23) (400,23)};

    \addplot[gray, dotted]coordinates {(400,0) (400,96)};
    \addplot[gray, dotted]coordinates {(0,96) (400,96)};
\addplot[black, myLine] coordinates {(4,7) (8,8) (12,9) (16,9) (20,10) (24,11) (28,12) (32,13) (36,14) (40,15) (44,16) (48,17) (52,18) (56,18) (60,19) (64,20) (68,21) (72,22) (76,23) (80,24) (84,25) (88,26) (92,27) (96,27) (100,28) (104,29) (108,30) (112,31) (116,32) (120,33) (124,34) (128,35) (132,35) (136,36) (140,37) (144,38) (148,39) (152,40) (156,41) (160,42) (164,43) (168,44) (172,44) (176,45) (180,46) (184,47) (188,48) (192,49) (196,50) (200,51) (204,52) (208,53) (212,53) (216,54) (220,55) (224,56) (228,57) (232,58) (236,59) (240,60) (244,61) (248,62) (252,62) (256,63) (260,64) (264,65) (268,66) (272,67) (276,68) (280,69) (284,70) (288,71) (292,71) (296,72) (300,73) (304,74) (308,75) (312,76) (316,77) (320,78) (324,79) (328,80) (332,80) (336,81) (340,82) (344,83) (348,84) (352,85) (356,86) (360,87) (364,88) (368,89) (372,89) (376,90) (380,91) (384,92) (388,93) (392,94) (396,95) (400,96)};
    
\end{axis}
\end{tikzpicture}     \label{fig:MHR:general:minnum}
  }
\end{minipage}
\hfill
\begin{minipage}{0.48\textwidth}
\centering
\subfloat[Ratio $\addBuyers/n$
  ]{
    \begin{tikzpicture}
\begin{axis}[
width=\textwidth,
height=6cm,
xmin=0,xmax=420,
ymin=0.05,ymax=0.6,
xtick={0, 400},   
xticklabels = {0, 400},
tick label style={font=\footnotesize}, scaled y ticks=false,                 yticklabel style={/pgf/number format/fixed}, ytick={0.05, 0.094, 0.275}, 
yticklabels = {0, 0.044, 0.225},
minor tick num=0,
axis line style=gray,
axis x line=middle,
axis y line=left,
tick style={black},
legend style={draw=none,fill=white,font=\small,at={(0.98,0.98)},anchor=north east},
]

\pgfplotsset{
  myLine/.style={line width=1pt}}

\addplot[black!20, myLine, domain=4:400, samples=100] {0.05};
\addplot[black!45, myLine] coordinates {(4,0.3) (8,0.175) (12,0.21666666666666667) (16,0.175) (20,0.15000000000000002) (24,0.13333333333333333) (28,0.12142857142857143) (32,0.14375) (36,0.13333333333333333) (40,0.125) (44,0.11818181818181818) (48,0.1125) (52,0.1076923076923077) (56,0.12142857142857143) (60,0.11666666666666667) (64,0.1125) (68,0.10882352941176471) (72,0.10555555555555556) (76,0.10263157894736842) (80,0.1125) (84,0.10952380952380952) (88,0.10681818181818181) (92,0.10434782608695653) (96,0.10208333333333333) (100,0.1) (104,0.1076923076923077) (108,0.10555555555555556) (112,0.10357142857142856) (116,0.10172413793103449) (120,0.1) (124,0.09838709677419355) (128,0.1046875) (132,0.10303030303030303) (136,0.10147058823529412) (140,0.1) (144,0.09861111111111112) (148,0.0972972972972973) (152,0.10263157894736842) (156,0.10128205128205128) (160,0.1) (164,0.09878048780487805) (168,0.09761904761904762) (172,0.10232558139534884) (176,0.10113636363636364) (180,0.1) (184,0.09891304347826088) (188,0.09787234042553192) (192,0.096875) (196,0.10102040816326531) (200,0.1) (204,0.09901960784313726) (208,0.09807692307692309) (212,0.09716981132075472) (216,0.0962962962962963) (220,0.1) (224,0.09910714285714287) (228,0.09824561403508772) (232,0.09741379310344828) (236,0.09661016949152543) (240,0.09583333333333333) (244,0.09918032786885246) (248,0.09838709677419355) (252,0.09761904761904762) (256,0.096875) (260,0.09615384615384616) (264,0.09924242424242424) (268,0.09850746268656717) (272,0.09779411764705882) (276,0.09710144927536232) (280,0.09642857142857143) (284,0.09577464788732395) (288,0.09861111111111112) (292,0.09794520547945205) (296,0.0972972972972973) (300,0.09666666666666668) (304,0.09605263157894736) (308,0.09545454545454546) (312,0.09807692307692309) (316,0.09746835443037975) (320,0.096875) (324,0.0962962962962963) (328,0.09573170731707317) (332,0.09518072289156626) (336,0.09761904761904762) (340,0.09705882352941177) (344,0.09651162790697675) (348,0.09597701149425288) (352,0.09545454545454546) (356,0.09775280898876404) (360,0.09722222222222222) (364,0.0967032967032967) (368,0.09619565217391304) (372,0.0956989247311828) (376,0.09521276595744682) (380,0.09736842105263158) (384,0.096875) (388,0.09639175257731959) (392,0.09591836734693879) (396,0.09545454545454546) (400,0.095)};
\addplot[black, myLine, ] coordinates {(4,0.55) (8,0.425) (12,0.3833333333333333) (16,0.3) (20,0.3) (24,0.3) (28,0.3) (32,0.3) (36,0.3) (40,0.3) (44,0.3) (48,0.3) (52,0.3) (56,0.28214285714285714) (60,0.2833333333333333) (64,0.284375) (68,0.2852941176470588) (72,0.2861111111111111) (76,0.2868421052631579) (80,0.2875) (84,0.28809523809523807) (88,0.28863636363636364) (92,0.2891304347826087) (96,0.2791666666666667) (100,0.28) (104,0.28076923076923077) (108,0.2814814814814815) (112,0.28214285714285714) (116,0.2827586206896552) (120,0.2833333333333333) (124,0.2838709677419355) (128,0.284375) (132,0.2772727272727273) (136,0.27794117647058825) (140,0.2785714285714286) (144,0.2791666666666667) (148,0.27972972972972976) (152,0.2802631578947368) (156,0.28076923076923077) (160,0.28125) (164,0.28170731707317076) (168,0.28214285714285714) (172,0.27674418604651163) (176,0.2772727272727273) (180,0.2777777777777778) (184,0.2782608695652174) (188,0.27872340425531916) (192,0.2791666666666667) (196,0.2795918367346939) (200,0.28) (204,0.2803921568627451) (208,0.28076923076923077) (212,0.27641509433962264) (216,0.27685185185185185) (220,0.2772727272727273) (224,0.27767857142857144) (228,0.27807017543859647) (232,0.278448275862069) (236,0.2788135593220339) (240,0.2791666666666667) (244,0.2795081967213115) (248,0.27983870967741936) (252,0.2761904761904762) (256,0.2765625) (260,0.27692307692307694) (264,0.2772727272727273) (268,0.27761194029850744) (272,0.27794117647058825) (276,0.2782608695652174) (280,0.2785714285714286) (284,0.27887323943661974) (288,0.2791666666666667) (292,0.276027397260274) (296,0.27635135135135136) (300,0.27666666666666667) (304,0.2769736842105263) (308,0.2772727272727273) (312,0.2775641025641026) (316,0.27784810126582277) (320,0.278125) (324,0.27839506172839507) (328,0.2786585365853659) (332,0.27590361445783135) (336,0.2761904761904762) (340,0.27647058823529413) (344,0.27674418604651163) (348,0.2770114942528736) (352,0.2772727272727273) (356,0.27752808988764044) (360,0.2777777777777778) (364,0.27802197802197803) (368,0.2782608695652174) (372,0.27580645161290324) (376,0.27606382978723404) (380,0.27631578947368424) (384,0.2765625) (388,0.2768041237113402) (392,0.2770408163265306) (396,0.2772727272727273) (400,0.2775)};
\addplot[gray, dotted]coordinates{(0,0.094)(400,0.094)};
    \addplot[gray, dotted]coordinates{(0,0.275)(400,0.275)};
  \end{axis}
\end{tikzpicture}     \label{fig:MHR:general:minratio}
  }
\end{minipage}\caption{Competition complexity for general markets under MHR distributions (to achieve $\ccapproxratio=1$). 
The dark-gray, and black curves correspond to the supply-to-demand ratios $\imbalanceratio=0.5$, and $\imbalanceratio=0.75$, respectively. 
In Figure~(a), the curves plot the competition complexity $\addBuyers$ as a function of the market size $n$. 
In Figure~(b), the curves plot the ratio of the competition complexity $\addBuyers$
to $n$, as a function of $n$, and the dotted lines mark the asymptotic competition complexities $0.044$ and $0.225$, corresponding to $\imbalanceratio=0.5$ and $\imbalanceratio=0.75$, respectively.
}
\label{fig:MHR distribution:general market}
\end{figure}
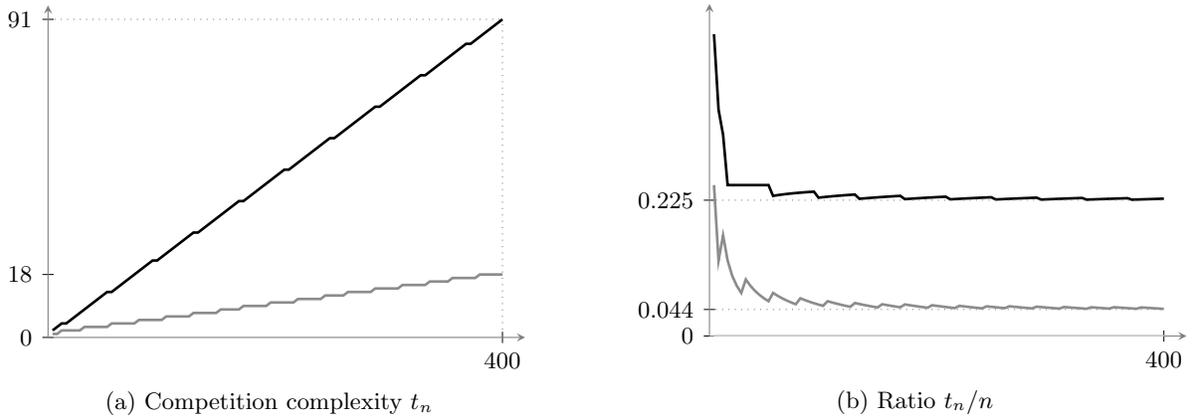

\Cref{fig:regular distribution:general market} illustrates both the growth of the competition complexity $\addBuyers$ as $n$ increases and the convergence of the ratio $\addBuyers/n$ under regular distributions when the approximation ratio is $\ccapproxratio=1$.
We let $n$ range from $4$ to $400$ with step size $4$, and set $\imbalanceratio \in \{0.5, 0.75, 1\}$.
The results show that as $n$ increases, the competition complexity $\addBuyers$ grows linearly, and the ratio $\addBuyers/n$ matches the asymptotic competition complexity $\plus{\ccapproxratio - 1 + \imbalanceratio}$ throughout.

\begin{figure}
\centering
\begin{minipage}{0.48\textwidth}
\centering
\begin{tikzpicture}
  \begin{axis}[
width=\textwidth,
height=6cm,
xmin=0,xmax=420,
ymin=0,ymax=420,
xtick={0, 400},   
xticklabels = {0, 400},
tick label style={font=\footnotesize}, scaled y ticks=false,                 yticklabel style={/pgf/number format/fixed}, ytick={0, 100,200,300,400}, 
yticklabels = {0, 100,200,300,400},
minor tick num=0,
axis line style=gray,
axis x line=middle,
axis y line=left,
tick style={black},
legend style={draw=none,fill=white,font=\small,at={(0.98,0.98)},anchor=north east},
]

\pgfplotsset{
  myLine/.style={line width=1pt}
}

    \addplot[black!20, myLine ] coordinates {(4,1) (8,2) (12,3) (16,4) (20,5) (24,6) (28,7) (32,8) (36,9) (40,10) (44,11) (48,12) (52,13) (56,14) (60,15) (64,16) (68,17) (72,18) (76,19) (80,20) (84,21) (88,22) (92,23) (96,24) (100,25) (104,26) (108,27) (112,28) (116,29) (120,30) (124,31) (128,32) (132,33) (136,34) (140,35) (144,36) (148,37) (152,38) (156,39) (160,40) (164,41) (168,42) (172,43) (176,44) (180,45) (184,46) (188,47) (192,48) (196,49) (200,50) (204,51) (208,52) (212,53) (216,54) (220,55) (224,56) (228,57) (232,58) (236,59) (240,60) (244,61) (248,62) (252,63) (256,64) (260,65) (264,66) (268,67) (272,68) (276,69) (280,70) (284,71) (288,72) (292,73) (296,74) (300,75) (304,76) (308,77) (312,78) (316,79) (320,80) (324,81) (328,82) (332,83) (336,84) (340,85) (344,86) (348,87) (352,88) (356,89) (360,90) (364,91) (368,92) (372,93) (376,94) (380,95) (384,96) (388,97) (392,98) (396,99) (400,100)};
\addplot[black!40, myLine] coordinates {(4,2) (8,4) (12,6) (16,8) (20,10) (24,12) (28,14) (32,16) (36,18) (40,20) (44,22) (48,24) (52,26) (56,28) (60,30) (64,32) (68,34) (72,36) (76,38) (80,40) (84,42) (88,44) (92,46) (96,48) (100,50) (104,52) (108,54) (112,56) (116,58) (120,60) (124,62) (128,64) (132,66) (136,68) (140,70) (144,72) (148,74) (152,76) (156,78) (160,80) (164,82) (168,84) (172,86) (176,88) (180,90) (184,92) (188,94) (192,96) (196,98) (200,100) (204,102) (208,104) (212,106) (216,108) (220,110) (224,112) (228,114) (232,116) (236,118) (240,120) (244,122) (248,124) (252,126) (256,128) (260,130) (264,132) (268,134) (272,136) (276,138) (280,140) (284,142) (288,144) (292,146) (296,148) (300,150) (304,152) (308,154) (312,156) (316,158) (320,160) (324,162) (328,164) (332,166) (336,168) (340,170) (344,172) (348,174) (352,176) (356,178) (360,180) (364,182) (368,184) (372,186) (376,188) (380,190) (384,192) (388,194) (392,196) (396,198) (400,200)};
\addplot[black!60,myLine] coordinates {(4,3) (8,6) (12,9) (16,12) (20,15) (24,18) (28,21) (32,24) (36,27) (40,30) (44,33) (48,36) (52,39) (56,42) (60,45) (64,48) (68,51) (72,54) (76,57) (80,60) (84,63) (88,66) (92,69) (96,72) (100,75) (104,78) (108,81) (112,84) (116,87) (120,90) (124,93) (128,96) (132,99) (136,102) (140,105) (144,108) (148,111) (152,114) (156,117) (160,120) (164,123) (168,126) (172,129) (176,132) (180,135) (184,138) (188,141) (192,144) (196,147) (200,150) (204,153) (208,156) (212,159) (216,162) (220,165) (224,168) (228,171) (232,174) (236,177) (240,180) (244,183) (248,186) (252,189) (256,192) (260,195) (264,198) (268,201) (272,204) (276,207) (280,210) (284,213) (288,216) (292,219) (296,222) (300,225) (304,228) (308,231) (312,234) (316,237) (320,240) (324,243) (328,246) (332,249) (336,252) (340,255) (344,258) (348,261) (352,264) (356,267) (360,270) (364,273) (368,276) (372,279) (376,282) (380,285) (384,288) (388,291) (392,294) (396,297) (400,300)};
\addplot[black, myLine] coordinates {(4,4) (8,8) (12,12) (16,16) (20,20) (24,24) (28,28) (32,32) (36,36) (40,40) (44,44) (48,48) (52,52) (56,56) (60,60) (64,64) (68,68) (72,72) (76,76) (80,80) (84,84) (88,88) (92,92) (96,96) (100,100) (104,104) (108,108) (112,112) (116,116) (120,120) (124,124) (128,128) (132,132) (136,136) (140,140) (144,144) (148,148) (152,152) (156,156) (160,160) (164,164) (168,168) (172,172) (176,176) (180,180) (184,184) (188,188) (192,192) (196,196) (200,200) (204,204) (208,208) (212,212) (216,216) (220,220) (224,224) (228,228) (232,232) (236,236) (240,240) (244,244) (248,248) (252,252) (256,256) (260,260) (264,264) (268,268) (272,272) (276,276) (280,280) (284,284) (288,288) (292,292) (296,296) (300,300) (304,304) (308,308) (312,312) (316,316) (320,320) (324,324) (328,328) (332,332) (336,336) (340,340) (344,344) (348,348) (352,352) (356,356) (360,360) (364,364) (368,368) (372,372) (376,376) (380,380) (384,384) (388,388) (392,392) (396,396) (400,400)};
    
    \addplot[gray, dotted] coordinates{(400,0)(400,400)};

    \addplot[gray, dotted] coordinates{(0,200)(400,200)};

    \addplot[gray, dotted] coordinates{(0,300)(400,300)};

    \addplot[gray, dotted] coordinates{(0,400)(400,400)};
\end{axis}
\end{tikzpicture} \end{minipage}
\hfill
\begin{minipage}{0.48\textwidth}
\centering
\begin{tikzpicture}
  \begin{axis}[
width=\textwidth,
height=6cm,
xmin=0,xmax=420,
ymin=0,ymax=1.2,
xtick={0, 400},   
xticklabels = {0, 400},
tick label style={font=\footnotesize}, scaled y ticks=false,                 yticklabel style={/pgf/number format/fixed}, ytick={0, 0.25,0.5,0.75,1}, 
yticklabels = {0, 0.25,0.5,0.75,1},
minor tick num=0,
axis line style=gray,
axis x line=middle,
axis y line=left,
tick style={black},
legend style={draw=none,fill=white,font=\small,at={(0.98,0.98)},anchor=north east},
]

\pgfplotsset{
  myLine/.style={line width=1pt}
}

\addplot[black!40,myLine] coordinates {(4,0.5) (8,0.5) (12,0.5) (16,0.5) (20,0.5) (24,0.5) (28,0.5) (32,0.5) (36,0.5) (40,0.5) (44,0.5) (48,0.5) (52,0.5) (56,0.5) (60,0.5) (64,0.5) (68,0.5) (72,0.5) (76,0.5) (80,0.5) (84,0.5) (88,0.5) (92,0.5) (96,0.5) (100,0.5) (104,0.5) (108,0.5) (112,0.5) (116,0.5) (120,0.5) (124,0.5) (128,0.5) (132,0.5) (136,0.5) (140,0.5) (144,0.5) (148,0.5) (152,0.5) (156,0.5) (160,0.5) (164,0.5) (168,0.5) (172,0.5) (176,0.5) (180,0.5) (184,0.5) (188,0.5) (192,0.5) (196,0.5) (200,0.5) (204,0.5) (208,0.5) (212,0.5) (216,0.5) (220,0.5) (224,0.5) (228,0.5) (232,0.5) (236,0.5) (240,0.5) (244,0.5) (248,0.5) (252,0.5) (256,0.5) (260,0.5) (264,0.5) (268,0.5) (272,0.5) (276,0.5) (280,0.5) (284,0.5) (288,0.5) (292,0.5) (296,0.5) (300,0.5) (304,0.5) (308,0.5) (312,0.5) (316,0.5) (320,0.5) (324,0.5) (328,0.5) (332,0.5) (336,0.5) (340,0.5) (344,0.5) (348,0.5) (352,0.5) (356,0.5) (360,0.5) (364,0.5) (368,0.5) (372,0.5) (376,0.5) (380,0.5) (384,0.5) (388,0.5) (392,0.5) (396,0.5) (400,0.5)};
\addplot[black!65,myLine] coordinates {(4,0.75) (8,0.75) (12,0.75) (16,0.75) (20,0.75) (24,0.75) (28,0.75) (32,0.75) (36,0.75) (40,0.75) (44,0.75) (48,0.75) (52,0.75) (56,0.75) (60,0.75) (64,0.75) (68,0.75) (72,0.75) (76,0.75) (80,0.75) (84,0.75) (88,0.75) (92,0.75) (96,0.75) (100,0.75) (104,0.75) (108,0.75) (112,0.75) (116,0.75) (120,0.75) (124,0.75) (128,0.75) (132,0.75) (136,0.75) (140,0.75) (144,0.75) (148,0.75) (152,0.75) (156,0.75) (160,0.75) (164,0.75) (168,0.75) (172,0.75) (176,0.75) (180,0.75) (184,0.75) (188,0.75) (192,0.75) (196,0.75) (200,0.75) (204,0.75) (208,0.75) (212,0.75) (216,0.75) (220,0.75) (224,0.75) (228,0.75) (232,0.75) (236,0.75) (240,0.75) (244,0.75) (248,0.75) (252,0.75) (256,0.75) (260,0.75) (264,0.75) (268,0.75) (272,0.75) (276,0.75) (280,0.75) (284,0.75) (288,0.75) (292,0.75) (296,0.75) (300,0.75) (304,0.75) (308,0.75) (312,0.75) (316,0.75) (320,0.75) (324,0.75) (328,0.75) (332,0.75) (336,0.75) (340,0.75) (344,0.75) (348,0.75) (352,0.75) (356,0.75) (360,0.75) (364,0.75) (368,0.75) (372,0.75) (376,0.75) (380,0.75) (384,0.75) (388,0.75) (392,0.75) (396,0.75) (400,0.75)};
\addplot[black, myLine] coordinates {(4,1.0) (8,1.0) (12,1.0) (16,1.0) (20,1.0) (24,1.0) (28,1.0) (32,1.0) (36,1.0) (40,1.0) (44,1.0) (48,1.0) (52,1.0) (56,1.0) (60,1.0) (64,1.0) (68,1.0) (72,1.0) (76,1.0) (80,1.0) (84,1.0) (88,1.0) (92,1.0) (96,1.0) (100,1.0) (104,1.0) (108,1.0) (112,1.0) (116,1.0) (120,1.0) (124,1.0) (128,1.0) (132,1.0) (136,1.0) (140,1.0) (144,1.0) (148,1.0) (152,1.0) (156,1.0) (160,1.0) (164,1.0) (168,1.0) (172,1.0) (176,1.0) (180,1.0) (184,1.0) (188,1.0) (192,1.0) (196,1.0) (200,1.0) (204,1.0) (208,1.0) (212,1.0) (216,1.0) (220,1.0) (224,1.0) (228,1.0) (232,1.0) (236,1.0) (240,1.0) (244,1.0) (248,1.0) (252,1.0) (256,1.0) (260,1.0) (264,1.0) (268,1.0) (272,1.0) (276,1.0) (280,1.0) (284,1.0) (288,1.0) (292,1.0) (296,1.0) (300,1.0) (304,1.0) (308,1.0) (312,1.0) (316,1.0) (320,1.0) (324,1.0) (328,1.0) (332,1.0) (336,1.0) (340,1.0) (344,1.0) (348,1.0) (352,1.0) (356,1.0) (360,1.0) (364,1.0) (368,1.0) (372,1.0) (376,1.0) (380,1.0) (384,1.0) (388,1.0) (392,1.0) (396,1.0) (400,1.0)};
\end{axis}
\end{tikzpicture} \end{minipage}\caption{Competition complexity for general markets under regular distributions (to achieve $\ccapproxratio=1$). 
The light-gray, dark-gray, black, dark-gray curves correspond to the supply-to-demand ratios $\imbalanceratio=0.5$, $\imbalanceratio=0.75$ and $\imbalanceratio=1$, respectively. 
In Figure~(a), the curves plot the competition complexity $\addBuyers$ as a function of the market size $n$. 
In Figure~(b), the curves plot the ratio of the competition complexity $\addBuyers$
to $n$, as a function of $n$, which coincide with the corresponding asymptotic competition complexities.}
\label{fig:regular distribution:general market}
\end{figure}
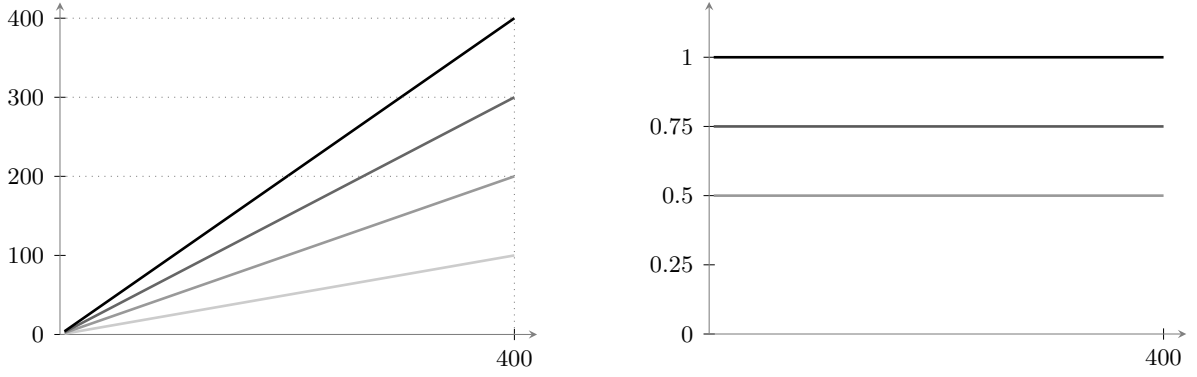

\section{Omitted Proofs}

This section contains all proofs omitted from the main text.

\subsection{Proof of \texorpdfstring{\Cref{thm:bk vcg:mhr finite market}}{Theorem 3.2}}
\label{apx:mhr finite market proof}

In this section, we prove \Cref{thm:bk vcg:mhr finite market}. 

\thmMHRLUB*

With a slight abuse of notation, we let $\RevVCG_{n:n+\addBuyers}(\optquant)$ denote the revenue of the {\VCGAuction} in the $n$-unit, $(n+\addBuyers)$-buyer market when buyer values are drawn from the $\frac{1}{\optquant}$-truncated exponential distribution.
Throughout \Cref{apx:mhr finite market proof}, we use $\RevVCG_{n:n+\addBuyers}(\optquant)$ and  $\RevVCG_{n:n+\addBuyers}(\reserve)$ interchangeably, where $\optquant$ is between $[1/e, 1]$ and $\reserve$ is between $[1,e]$.

We first establish a sequence of technical lemmas needed for our final analysis.

\begin{lemma}\label{lem:VCG revenue derivative}
Define a random variable $X\sim\BetaDist{n+1}{\addBuyers}$ with density
$p\left(\variablex\right)\triangleq \frac{1}{\BetaFun{n+1}{\addBuyers}}\cdot \variablex^{n}\cdot\left(1-\variablex\right)^{\addBuyers-1}$.
Let $
\rprob \triangleq \prob{\left(X\ge \optquant\right)}
= \int_{\optquant}^{1} p\left(\variablex\right)\,\dd \variablex$
and
$\lnexpect \triangleq \expect{\left(\left(-\ln X\right)\cdot \IF_{\left\{X\ge \optquant\right\}}\right)}
= \int_{\optquant}^{1} \left(-\ln \variablex\right)\cdot p\left(\variablex\right)\,\dd \variablex$.
Then,
\[
\frac{\dd \ln\left(\RevVCG_{n:n+\addBuyers}\left(\optquant\right)\right)}{\dd \optquant}
=
-\frac{1}{\optquant}
+
\frac{1}{\optquant}\cdot
\frac{\lnexpect}{
\left(\ln \optquant\right)^2\cdot
\left(1-\rprob-\frac{\lnexpect}{\ln \optquant}\right)
}.
\]
\end{lemma}
\begin{proof}[Proof of \Cref{lem:VCG revenue derivative}]

Define the quantity $\coefA \triangleq (n+\addBuyers)\cdot(n+\addBuyers-1)\cdot\binom{n+\addBuyers-2}{n-1}$, which does not depend on $\optquant$.
The expression for $\RevVCG_{n:n+\addBuyers}\left(\optquant\right)$ is as follows (recall $\osdensity_{a:b}$ is the PDF of $a$-th smallest order statistics of $b$ i.i.d.\ samples from uniform distribution $U[0,1]$):
\begin{align*}
    \RevVCG_{n:n+\addBuyers}(\optquant)={}&
    \left(n+\addBuyers\right)\cdot \displaystyle\int_0^1 \revcurve_{\optreserve}\left(\quant\right)
    \cdot \osdensity_{n:n+\addBuyers - 1}\left(\quant\right)\dd  \quant
    \\={}&(n+\addBuyers)\cdot\left[\displaystyle\int_0^{\optquant} \optreserve\cdot\quant\cdot\osdensity_{n:n+\addBuyers - 1}\dd  \quant
    +\displaystyle\int_{\optquant}^1 \frac{\optreserve}{\ln \optreserve}\cdot\quant\cdot\ln
    \frac{1}{\quant}\cdot\osdensity_{n:n+\addBuyers - 1}\dd  \quant\right]
    \\={}&  
     \frac{\coefA}{\optquant}\cdot
    \left(\int_0^{\optquant} (1-\quant)^{\addBuyers-1}\cdot \quant^{n}\dd \quant -
   \frac{1}{\ln \optquant}\cdot\displaystyle\int_{\optquant}^1 -\ln \quant\cdot(1-\quant)^{\addBuyers-1}\cdot \quant^{n}\dd \quant\right)
\end{align*}
By the definition of the Beta function, we know that $
\BetaFun{n+1}{\addBuyers}=\int_0^1 \quant^{n}\cdot\left(1-\quant\right)^{\addBuyers-1}\dd \quant$. 
Then we have:
\begin{align*}
    &\int_0^{\optquant} (1-\quant)^{\addBuyers-1}\cdot \quant^{n}\dd \quant
=\BetaFun{n+1}{\addBuyers}\cdot
(1-\rprob)
\\&\int_{\optquant}^1 -\ln \quant\cdot(1-\quant)^{\addBuyers-1}\cdot \quant^{n}\dd \quant
=\BetaFun{n+1}{\addBuyers}\cdot
\lnexpect
\end{align*}
These probabilistic representations allow us to rewrite the key terms in a compact form:
\begin{align*}
     \RevVCG_{n:n+\addBuyers}(\optquant)=\frac{\coefA}{\optquant}\cdot\BetaFun{n+1}{\addBuyers}\cdot\left(1-\rprob-\frac{\lnexpect}{\ln \optquant}\right)
\end{align*}
We will analyze the monotonicity of $ \RevVCG_{n:n+\addBuyers}(\optquant)$ by studying the sign of its logarithmic derivative.
Since $ \RevVCG_{n:n+\addBuyers}(\optquant) > 0$ for all $\optquant$, it suffices to study the sign of
$\dd \ln( \RevVCG_{n:n+\addBuyers}(\optquant))/{\dd \optquant} = {\RevVCG_{n:n+\addBuyers}(\optquant)^\prime}/{ \RevVCG_{n:n+\addBuyers}(\optquant)}$, 
where
\begin{align}\label{eq:multi-unit bk:mhr:balance market:ln VCG}
    \ln( \RevVCG_{n:n+\addBuyers}(\optquant))=\ln \frac{\coefA\cdot \BetaFun{n+1}{\addBuyers}}{\optquant} + \ln \left(1-\rprob-\frac{\lnexpect}{\ln \optquant}\right)
\end{align}
Since $\coefA\cdot \BetaFun{n+1}{\addBuyers}$ is independent of $\optquant$, in \Cref{eq:multi-unit bk:mhr:balance market:ln VCG}, the derivative of the first term with respect to $\optquant$ equals $-1/\optquant$.

Let $
\VCGrpart\triangleq1-\rprob-{\lnexpect}/{\ln \optquant}$ denote the argument of the logarithm in the second term of \Cref{eq:multi-unit bk:mhr:balance market:ln VCG}. 
Differentiating $\VCGrpart$ with respect to $\optquant$, we obtain
\[
\dVCGrpart=-\drprob+\frac{\dlnexpect}{\abs{\ln \optquant}}+\frac{\lnexpect}{\optquant\cdot (\ln \optquant)^2}=\frac{\lnexpect}{\optquant\cdot (\ln \optquant)^2}
\]
The last equality follows from Leibniz' rule, which yields $\drprob=-p(\optquant)$ and $\dlnexpect=-(-\ln \optquant)\,p(\optquant)$.
Therefore, 
\begin{align*}
    \frac{\dd \ln( \RevVCG_{n:n+\addBuyers}(\optquant))}{\dd \optquant}
    ={}-\frac{1}{\optquant}+\frac{\dVCGrpart}{\VCGrpart}
    ={}-\frac{1}{\optquant}+\frac{1}{\optquant}\cdot\frac{\lnexpect}{\left(\ln \optquant\right)^2\cdot\left(1-\rprob-\frac{\lnexpect}{\ln \optquant}\right)}
\end{align*}
This complete the proof of \Cref{lem:VCG revenue derivative}
\end{proof}

\begin{lemma}
\label{lem:VCG revenue small reserve}
    For any $n \geq 594$ and $\reserve \in [1, 5/4]$, let 
    $\addBuyers \triangleq \lceil \left(e^{1/e} - 1\right) n + 1.05 \ln n \rceil$.
    Then,
    $\RevVCG_{n:n+\addBuyers}(\reserve) \geq \RevVCG_{n:n+\addBuyers}(1) = n$.
\end{lemma}

\begin{proof}[Proof of \Cref{lem:VCG revenue small reserve}]
In this proof, we will always consider on $\RevVCG_{n:n+t_n}(\optquant)$, where $\optquant$ is between $[0.8, 1]$. 

Let $\ExpectX\triangleq \expect{X}=(n+1)/(n+\addBuyers+1)$, which is approximately $0.69$ when $n\ge 594$. By Hoeffding’s inequality, for $\threshold=0.1$ we have
\begin{align*}
\prob{\abs{X-\ExpectX}\ge \threshold}\le 2\cdot e^{-(n+\addBuyers+1)\cdot \threshold^2}= 2\cdot  e^{-0.01\cdot (n+\addBuyers+1)}
\end{align*}
when $\optquant\in[0.8, 1]$, we know that $\optquant\ge\ExpectX+\threshold$.

We will prove that ${\dd \ln( \RevVCG_{n:n+\addBuyers}(\optquant))}/{\dd \optquant}<0$. According to \Cref{lem:VCG revenue derivative}, 
it is equivalent to
\begin{align}\label{eq:multi-unit bk:MHR:balance market:increasing condition}
    \frac{\lnexpect}{(\ln\optquant)^2}\cdot(1+\ln\optquant)<1-\rprob.
\end{align}

\xhdr{Case (i)}
When $\optquant\to 1^-$, by l'Hôpital's rule, 
\begin{align*}
\lim_{\optquant\to 1^-}\frac{\lnexpect}{(\ln \optquant)^2}
= \lim_{\optquant\to 1^-}
\frac{\dfrac{\dd}{\dd\optquant}\lnexpect}{\dfrac{\dd}{\dd\optquant}(\ln \optquant)^2}
= \lim_{\optquant\to 1^-}
\frac{(\ln \optquant)\cdot p(\optquant)}{2\ln \optquant / \optquant}
= \lim_{\optquant\to 1^-}
\frac{\optquant\cdot p(\optquant)}{2}
=0
\end{align*}
The last equality holds since $p(\optquant)=0$ when $\optquant=1$.
Therefore, condition~\ref{eq:multi-unit bk:MHR:balance market:increasing condition} holds when $\optquant=1$.

\xhdr{Case (ii)}
When $\optquant\in[0.99,1]$, 
 we apply the substitutions $\optquant= 1-\variableu$ and $\variablex=1-\variablet$. Accordingly, $\lnexpect$ and $p(\optquant)$ are rewritten as $h(1-\variableu)$ and $p(1-\variablet)$, respectively. We note that $u\in[0,0.01]$. Then
\begin{align*}
h\left(1-\variableu\right)=\int_{0}^{\variableu}-\ln\left(1-\variablet\right)\cdot\pdf{1-\variablet}\dd \variablet
\,\,\,\text{and}\,\,\,
\pdf{1-\variablet}=\frac{1}{\BetaFun{n+1}{\addBuyers}}\cdot\left(1-\variablet\right)^{n}
\cdot\variablet^{\addBuyers-1}
\end{align*}
Since $0\le \variablet\le \variableu\le0.01$, using the Taylor expansion of $-\ln(1-\variablet)$, we have
\begin{align*}
-\ln\left(1-\variablet\right)=
\variablet+\frac{\variablet^{2}}{2}+\sum_{\indexj=3}^{\infty}\frac{\variablet^{\indexj}}{\indexj}
\le \variablet+\frac{\variablet^{2}}{2}+\frac{1}{3}\sum_{\indexj=3}^{\infty}\variablet^{\indexj}
=\variablet+\frac{\variablet^{2}}{2}+\frac{\variablet^{3}}{3\left(1-\variablet\right)}
\le \variablet+\frac{\variablet^{2}}{2}+\frac{\variablet^{3}}{3\left(1-\variableu\right)}
\end{align*}
Moreover, since $\left(1-\variablet\right)^{n}\le1$, it follows that
\begin{align*}
h\left(1-\variableu\right)
&=\frac{1}{\BetaFun{n+1}{\addBuyers}}\cdot\int_{0}^{\variableu}-\ln\left(1-\variablet\right)\cdot\left(1-\variablet\right)^{n}\cdot\variablet^{\addBuyers-1}\dd \variablet
\\&\le
\frac{1}{\BetaFun{n+1}{\addBuyers}}\cdot\int_{0}^{\variableu}\left(\variablet+\frac{\variablet^{2}}{2}+\frac{\variablet^{3}}{3\cdot\left(1-\variableu\right)}\right)\cdot\variablet^{\addBuyers-1}\dd \variablet\\
&=\frac{1}{\BetaFun{n+1}{\addBuyers}}\cdot\left(\frac{\variableu^{\addBuyers+1}}{\addBuyers+1}+\frac{\variableu^{\addBuyers+2}}{2\cdot\left(\addBuyers+2\right)}+\frac{\variableu^{\addBuyers+3}}{3\cdot\left(1-\variableu\right)\cdot\left(\addBuyers+3\right)}\right)
\end{align*}
Therefore, we conclude that for all $\optquant\in\left[0.99,1\right]$,
\begin{align*}
\frac{h\left(\optquant\right)}{\left(\ln \optquant\right)^{2}}
={}\frac{h\left(1-\variableu\right)}{\left(\ln\left(1-\variableu\right)\right)^{2}}
\le{}&\frac{1}{\BetaFun{n+1}{\addBuyers}}\cdot\left(\frac{\variableu^{\addBuyers-1}}{\addBuyers+1}+\frac{\variableu^{\addBuyers}}{2\cdot\left(\addBuyers+2\right)}+\frac{\variableu^{\addBuyers+1}}{3\cdot0.99\cdot\left(\addBuyers+3\right)}\right)
\\\le{}&\frac{(0.01)^{\addBuyers-1}}{\BetaFun{n+1}{\addBuyers}}\cdot
\left(\frac{1}{\addBuyers+1}+\frac{0.01}{2\cdot\left(\addBuyers+2\right)}+\frac{(0.01)^2}{3\cdot0.99\cdot\left(\addBuyers+3\right)}\right),
\end{align*}
the first inequality holds since  the inequality $-\ln\left(1-u\right)\ge u$ implies
$\left(\ln\left(1-\variableu\right)\right)^{2}\ge \variableu^{2}$, and $\frac{1}{1-\variableu}\le\frac{1}{0.99}$ for all $u\in\left[0,0.01\right]$.
The second inequality holds since
for fixed $n$, $\variableu^{\addBuyers-1}$, $\variableu^{\addBuyers}$ and $\variableu^{\addBuyers+1}$ are all monotonically increasing in $\variableu$, and hence attain their maximum at $\variableu=0.01$. 
Using Stirling's approximation for the Beta function, when $n\ge594$, let $\kratio\triangleq0.444\le \frac{\addBuyers}{n}$, we have 
\begin{align*}
\frac{1}{\BetaFun{n+1}{\addBuyers}}\le 0.32\cdot\sqrt{n}\cdot e^{\left(\left(1+\kratio\right)
\cdot\ln\left(1+\kratio\right)-\kratio\cdot\ln\kratio\right)\cdot n}
\le 0.32\cdot\sqrt{n}\cdot e^n
\end{align*}
Moreover, $(0.01)^{\addBuyers-1}\le(0.01)^{\kratio
\cdot n-1}=100\cdot e^{-\kratio\cdot n\cdot\ln100}\le 100\cdot e^{-2\cdot n}$.
Combining the two bounds, we obtain
\begin{align*}
\frac{h\left(\optquant\right)}{\left(\ln \optquant\right)^{2}}&\le 32
\cdot\sqrt{n}\cdot e^{-n}\cdot \left(\frac{1}{\addBuyers+1}+\frac{0.01}{2\cdot\left(\addBuyers+2\right)}+\frac{(0.01)^2}{3\cdot0.99\cdot\left(\addBuyers+3\right)}\right)
\\&\le 32
\cdot\sqrt{594}\cdot e^{-594}\cdot \left(\frac{1}{271+1}+\frac{0.01}{2\cdot\left(271+2\right)}+\frac{(0.01)^2}{3\cdot0.99\cdot\left(271+3\right)}\right)\approx 3.08\cdot10^{-258},
\end{align*}
the second inequality holds since $\addBuyers \triangleq \lceil \left(e^{1/e} - 1\right) n + 1.05 \ln n \rceil$ grows linearly in $n$, the fuctions ${1}/(\addBuyers+1)+{0.01}/(2\cdot\left(\addBuyers+2\right))+{(0.01)^2}/(3\cdot0.99\cdot\left(\addBuyers+3\right))$ and $n\cdot e^{-n}$ are both decreasing in $n$, and hence attain their maximum at $n=594$ and $\addBuyers=271$.

Since $1-\rprob\ge 0.5$ and $1+\ln\optquant\le 1$ when $\optquant\in[0.8,1]$, we complete the proof of condition~\ref{eq:multi-unit bk:MHR:balance market:increasing condition}, which is $({\lnexpect}/{(\ln\optquant)^2})\cdot(1+\ln\optquant)<1-\rprob$.

\xhdr{Case (iii)} When $\optquant\in[0.8,0.99]$, it holds that $\ln \optquant\in[\ln0.99,\ln0.8]$, then
\begin{align*}
    \lnexpect={}&\int_{\optquant}^1 -\ln \variablex\cdot p(\variablex)\dd \variablex
    \le\int_{0.8}^{1} -\ln \variablex\cdot p(\variablex)\dd \variablex
    \\\le{}&-\ln0.8\cdot \prob{0.8\le X\le 1}
    \le-\ln 0.8\cdot2\cdot e^{-0.01\cdot(n+\addBuyers+1)}
\end{align*}
Additionally, we know that $(\ln\optquant)^2>(\ln0.99)^2$, $1+\ln\optquant<1$ and $1-\rprob=\prob{X\le \optquant}\ge 1-2\cdot  e^{-0.01\cdot (n+\addBuyers+1)}$.
Therefore, condition~\ref{eq:multi-unit bk:MHR:balance market:increasing condition} holds when $n\ge594$:
\begin{align*}
    \frac{\lnexpect}{(\ln\optquant)^2}\cdot(1+\ln\optquant)
    \le\frac{-\ln 0.8\cdot2\cdot e^{-0.01\cdot(n+\addBuyers+1)}}{(\ln0.99)^2}\cdot1
    <1-2\cdot  e^{-0.01\cdot (n+\addBuyers+1)}
    \le 1-\rprob
\end{align*}

Finally, combining the conclusions from the three cases establishes that $ \RevVCG_{n:n+\addBuyers}(\optquant)$ is monotonically decreasing in 
$\optquant$ when $\optquant\in[0.8,1]$. Moreover, since $ \RevVCG_{n:n+\addBuyers}(1)=n$, we conclude the proof of \Cref{lem:VCG revenue small reserve}.
\end{proof} 

\begin{lemma}
\label{lem:VCG revenue middle reserve}
    For any $n \geq 594$ and $\reserve \in [5/4, 2]$, let 
    $\addBuyers \triangleq \lceil \left(e^{1/e} - 1\right) n + 1.05 \ln n \rceil$.
    Then,
    $\RevVCG_{n:n+\addBuyers}(\reserve) \geq n$.
\end{lemma}

\begin{proof}[Proof of \Cref{lem:VCG revenue middle reserve}]

Define 
$
\CriticalQuantile \triangleq \frac{n}{n+\addBuyers},\,  
\threshold \triangleq \sqrt{\frac{\log n}{2\cdot(n-1)}}$, and $ 
\revLB \triangleq \min_{\quant:\shiftCriticalQuantile\le \threshold}
\revcurve_{\optreserve}\left(\quant\right)$.
Then, $\RevVCG_{n:n+\addBuyers}\left(\reserve\right)$ 
can be lower bounded as follows:
\begin{align*}
    &\RevVCG_{n:n+\addBuyers}\left(\reserve\right)
    = \left(n+\addBuyers\right)\cdot \displaystyle\int_0^1 \revcurve_{\optreserve}\left(\quant\right)
        \cdot \osdensity_{n:n+\addBuyers - 1}\left(\quant\right)\dd \quant
    =\left(n+\addBuyers\right)\cdot\expect{\revcurve_{\optreserve}\left(\quant\right)}
    \\={}&\left(n+\addBuyers\right)\cdot\left(\prob{\shiftCriticalQuantile\le \threshold}\cdot\expect{\revcurve_{\optreserve}\left(\quant\right)\condition\shiftCriticalQuantile\le \threshold}
    +\prob{\shiftCriticalQuantile> \threshold}\cdot\expect{\revcurve_{\optreserve}\left(\quant\right)\condition\shiftCriticalQuantile> \threshold}\right)
    \\ \ge{} & \left(n+\addBuyers\right)\cdot\left(\prob{\shiftCriticalQuantile\le \threshold}\cdot\revLB+
    \prob{\shiftCriticalQuantile> \threshold}\cdot 0\right)
    =\left(n+\addBuyers\right)\cdot\prob{\shiftCriticalQuantile\le \threshold}\cdot\revLB,
    \end{align*}
where all the expectations and probabilities are taken over the random quantile $\quant$, assuming that $\quant$ is the $n$-th smallest order statistic of $n+k_n-1$ samples drawn from the uniform distribution $U[0, 1]$.
By Hoeffding’s inequality, the following probability bound holds:
\begin{align*}
\prob{\shiftCriticalQuantile\le \threshold}
&\ge 1-2\cdot e^{-2\cdot\left(n+\addBuyers-1\right)\cdot\left(\threshold+
\frac{n}{\left(n+\addBuyers-1\right)\cdot\left(n+\addBuyers\right)}\right)^2}
\ge 1-2\cdot e^{-2\cdot \left(n-1\right)\cdot\threshold^2}
\\&= 1-2\cdot e^{-2\cdot \left(n-1\right)\cdot\frac{\log n}{n-1}}
=1-\frac{2}{n}.
\end{align*}
Therefore, $\left(n+\addBuyers\right)\cdot\left(1-\frac{2}{n}\right)\cdot \revLB$ 
serves as a lower bound for 
$\RevVCG_{n:n+\addBuyers}\left(\buyerdist_{\optreserve}\right)$. We can prove \Cref{lem:VCG revenue middle reserve} by showing that $\addBuyers = \lceil \left(e^{1/e} - 1\right) n + 1.05 \ln n \rceil$ satisfies the following condition:
\begin{align*}
\left(n+\addBuyers\right)\cdot\left(1-\frac{2}{n}\right)\cdot \revLB
\ge n
\end{align*}
Define the left endpoint $\qLeft\triangleq\max(0,\CriticalQuantile-\threshold)$, and the right endpoint $\qRight\triangleq\min(1,\CriticalQuantile+\threshold)$.
The distribution is MHR and therefore regular, implying that the revenue curve is concave. 
Consequently, the revenue curve is single-peaked, and 
the value of $\revLB$ is either $\revcurve_{\optreserve}\left(\qLeft\right)$ or $\revcurve_{\optreserve}\left(\qRight\right)$.

\xhdr{Case (i)} When $\revcurve_{\optreserve}\left(\qLeft\right) \le \revcurve_{\optreserve}\left(\qRight\right)$ and $\optreserve \in [5/4, {1}/{\qLeft}]$. 
By construction, $\revLB = \revcurve_{\optreserve}\left(\qLeft\right) = {\optreserve} \cdot \qLeft$.
Then, the following condition must hold:
\begin{align}\label{eq:multi-unit bk:MHR:balance market:upper bound condition}
    \left(n+\addBuyers\right)\cdot{\optreserve}\cdot\qLeft\cdot\left(1-\frac{2}{n}\right)\ge  n.
\end{align}
When $n \ge 594$ and $\optreserve \in [5/4,1/\qLeft]$, straightforward calculation shows that 
condition~\ref{eq:multi-unit bk:MHR:balance market:upper bound condition} 
is satisfied for $\addBuyers = \lceil \left(e^{1/e} - 1\right) n + 1.05 \ln n \rceil$:
\begin{align*}
\left(n+\addBuyers\right)\cdot{\optreserve}\cdot\qLeft\cdot\left(1-\frac{2}{n}\right)
&\ge \left(n - 2\cdot n \cdot \threshold\right)\cdot\frac{5}{4}\cdot\left(1-\frac{2}{n}\right)
\\&\ge n\cdot\left(1 - 2\cdot\sqrt{\frac{\log 594}{2\cdot593}}\right)\cdot\frac{5}{4}\cdot\left(1-\frac{2}{594}\right)
\ge n,
\end{align*}
the first inequality holds since $
\qLeft={n}/(n+\addBuyers)-\epsilon$ and $\addBuyers\le n$, and the second inquality holds since $n\ge594$ and $\epsilon=\sqrt{\log n/(2\cdot(n-1))}$ and $2/n$ are both decreasing in $n$.

\xhdr{Case (ii)} 
When $\revcurve_{\optreserve}\left(\qLeft\right)\le\revcurve_{\optreserve}\left(\qRight\right)$ and $\optreserve\in[{1}/{\qLeft},2]$, we have
$\revLB=\revcurve_{\optreserve}\left(\qLeft\right)=({\optreserve}/\ln \optreserve)\cdot\qLeft\cdot\ln({{1}/{\qLeft}})
\ge (2/\ln 2)\cdot\qLeft\cdot\ln{({1}/{\qLeft})}$,
where the inequality holds because $\optreserve/\ln\optreserve$ is decreasing in $\optreserve$, and $\max \optreserve=2$.  
Therefore, we need to prove
\begin{align*}
    &\left(n+\addBuyers\right)\cdot \frac{2}{\ln2}\cdot\qLeft\cdot\ln{\frac{1}{\qLeft}}\cdot\left(1-\frac{2}{n}\right)
    =n\cdot\left(1-\left(1+\frac{\addBuyers}{n}\right)\cdot\threshold\right)\cdot \frac{2}{\ln2}\cdot \ln\left(\frac{1}{\qLeft}\right)\cdot\left(1-\frac{2}{n}\right)
    \\\ge{}&n\cdot\left(1-\left(1+\frac{ \left(e^{1/e} - 1\right) \cdot594 + 1.05 \ln594 +1}{594}\right)\cdot\sqrt{\frac{\log 594}{2\cdot593}}\right)\cdot \frac{2}{\ln2}\cdot 1\cdot\left(1-\frac{2}{594}\right) 
    \ge n,
\end{align*}
the first equality holds since $\qLeft={n}/(n+\addBuyers)-\threshold$, and the first inequality holds since $n\ge 594$, and the quantities $\addBuyers/n$, $\threshold=\sqrt{\log n/(2\cdot(n-1))}$, and $2/n$ are all decreasing in $n$; hence the expression is overall monotonically increasing in $n$.

\xhdr{Case (iii)} 
When $\revcurve_{\optreserve}\left(\qLeft\right)>\revcurve_{\optreserve}\left(\qRight\right)$ and $\optreserve\in[5/4,{1}/{\qRight}]$, it holds that
$
\revLB=\revcurve_{\optreserve}\left(\qRight\right)={\optreserve}\cdot\qRight
$. Therefore, it holds that
\begin{align*}
    \left(n+\addBuyers\right)\cdot{\optreserve}\cdot\qRight\cdot\left(1-\frac{2}{n}\right)= \left(n+\left(n+\addBuyers\right)\cdot\threshold\right)\cdot{\optreserve}\cdot\left(1-\frac{2}{n}\right)\ge n\cdot \frac{5}{4}\cdot\left(1-\frac{2}{594}\right)\ge n
\end{align*}
the first equality holds since $\qRight={n}/(n+\addBuyers)+\epsilon$, and the first inequality holds since $n\ge  594$ and $1-2/n$ is increasing in $n$.

\xhdr{Case (iv)} 
When $\revcurve_{\optreserve}\left(\qLeft\right)>\revcurve_{\optreserve}\left(\qRight\right)$ and $\optreserve\in[{1}/{\qRight},2]$, it follows that $    \revLB=\revcurve_{\optreserve}\left(\qRight\right)=({\optreserve}/\ln \optreserve)\cdot\qRight\cdot\ln({{1}/{\qRight}})
\ge (2/\ln 2)\cdot\qRight\cdot\ln{({1}/{\qRight})}$
where the inequality holds because $\optreserve/\ln\left(\optreserve\right)$ is decreasing in $\optreserve$, and $\max \optreserve=2$.
Therefore, it holds that
\begin{align*}
    \left(n+\addBuyers\right)\cdot \frac{2}{\ln 2}\cdot\qRight\cdot\ln{\frac{1}{\qRight}}\cdot\left(1-\frac{2}{n}\right)
    =\left(n+\left(n+\addBuyers\right)\cdot\threshold\right)\cdot{\frac{2}{\ln2}}\cdot\left(1-\frac{2}{n}\right)\ge n\cdot \frac{2}{\ln2}\cdot\left(1-\frac{2}{594}\right)\ge n
\end{align*}
the first equality holds since $\qRight=\frac{n}{n+\addBuyers}+\epsilon$, and the first inequality holds since $n\ge  594$ and $1-2/n$ is increasing in $n$.

Combining the four cases completes the proof of \Cref{lem:VCG revenue middle reserve}.
\end{proof}

\begin{lemma}
\label{lem:VCG revenue high reserve}

    For any $n \geq 594$ and $\reserve \in [2, e/(1+1/n)]$, let 
    $\addBuyers \triangleq \lceil \left(e^{1/e} - 1\right) n + 1.05 \ln n \rceil$.
    Then,
    $\RevVCG_{n:n+\addBuyers}(\reserve) \geq \RevVCG_{n:n+\addBuyers}(e/(1+1/n))$.
\end{lemma}

\begin{proof}[Proof of \Cref{lem:VCG revenue high reserve}]
In this proof, we will always consider on $\RevVCG_{n:n+t_n}(\optquant)$, where $\optquant$ is between $[({1}/{e})\cdot(1+{1}/{n}), 0.5]$. 

Let $\ExpectX\triangleq \expect{X}=(n+1)/(n+\addBuyers+1)$, which is approximately $0.69$ when $n\ge 594$. Recall that by Hoeffding’s inequality, for $\threshold=0.1$ we have
\begin{align*}
\prob{\abs{X-\ExpectX}\ge \threshold}\le 2\cdot e^{-(n+\addBuyers+1)\cdot \threshold^2}= 2\cdot  e^{-0.01\cdot (n+\addBuyers+1)}
\end{align*}
when $\optquant\in[({1}/{e})\cdot(1+{1}/{n}), 0.5]$, we know that $\optquant<\ExpectX(\optquant)-\threshold$, and we will prove that
\begin{align*}
    {\dd \ln( \RevVCG_{n:n+\addBuyers}(\optquant))}/{\dd \optquant}>0
\end{align*}
According to \Cref{lem:VCG revenue derivative}, 
it is equivalent to
\[(1+\ln\optquant)\cdot\lnexpect>(\ln\optquant)^2\cdot(1-\rprob)\]
Since $\optquant\in[({1}/{e})\cdot(1+{1}/{n}), 0.5]$, we know that $1+\ln\optquant\ge \ln(1+{1}/{n})$ and $(\ln\optquant)^2\le 1$.
Since $\optquant<\ExpectX(\optquant)-\threshold$, we know that
\begin{align*}
    \lnexpect={}&\int_{\optquant}^1 -\ln \variablex\cdot p(\variablex)\, \dd \variablex
    \ge\int_{\optquant}^{\ExpectX+\threshold} -\ln \variablex\cdot p(\variablex)\, \dd \variablex
    \\\ge{}&-\ln(\ExpectX+\threshold)\cdot \prob{0.5\le X\le \ExpectX+\threshold}
    \ge-\ln 0.8\cdot\left(1-2\cdot e^{-0.01\cdot(n+\addBuyers+1)}\right)
\end{align*}
and $1-\rprob=\prob{X\le \optquant}\le 2\cdot  e^{-0.01\cdot (n+\addBuyers+1)}$.
Therefore, we need to prove that when $n\ge 594$, it holds that
\begin{align*}
    \ln(1+\frac{1}{n})\cdot(-\ln 0.8)\cdot\left(1-2\cdot e^{-0.01\cdot(n+\addBuyers+1)}\right)>2\cdot  e^{-0.01\cdot (n+\addBuyers+1)}.
\end{align*}

As $n\to\infty$, the leading term on the left satisfies $\ln\left(1+\frac{1}{n}\right)\sim \frac{1}{n}$, which decays only polynomially, whereas the right-hand side $2e^{-0.01(n+\addBuyers+1)}$ decays exponentially to $0$; hence, for sufficiently large $n$, the exponential terms become negligible and the inequality must hold.
By calculation, we know that when $n\ge 594$, the VCG revenue $ \RevVCG_{n:n+\addBuyers}(\optquant)$ is increasing when $\optquant\in[({1}/{e})\cdot(1+{1}/{n}), 0.5]$.

This completes the proof of \Cref{lem:VCG revenue high reserve}.
\end{proof}

\begin{lemma}
\label{lem:VCG revenue super high reserve}
    For any $n \geq 594$ and $\reserve \in [e/(1+1/n), e]$, let 
    $\addBuyers \triangleq \lceil \left(e^{1/e} - 1\right) n + 1.05 \ln n \rceil$.
    Then,
    $\RevVCG_{n:n+\addBuyers}(\reserve) \geq \RevVCG_{n:n+\addBuyers}(e) - 8$.
\end{lemma}

\begin{proof}[Proof of \Cref{lem:VCG revenue super high reserve}]
In this proof, we will always consider on $\RevVCG_{n:n+t_n}(\optquant)$, where $\optquant$ is between $[1/e, ({1}/{e})\cdot(1+{1}/{n})]$.

We first prove that when $\reserve\in[e/(1+1/n),e]$, or equivalently when $\optquant\in[1/e,(1+1/n)\cdot 1/e]$, let $\addBuyers\triangleq \lceil \left(e^{1/e} - 1\right) n + 1.05 \ln n \rceil$, for every integer $n\ge 594$, the derivative of the revenue with respect to 
$\optquant$ satisfies 
$\abs{\dd \RevVCG_{n:n+\addBuyers}(\optquant)/\dd\optquant}\le 0.5\cdot n$.

Let $\ExpectX\triangleq \expect{X}=(n+1)/(n+\addBuyers+1)$, which is approximately $0.69$ when $n\ge 594$. To bound $\dd \ln( \RevVCG_{n:n+\addBuyers}(\optquant))/\dd \optquant$, we will repeatedly use the following inequality: by Hoeffding’s inequality, for $\threshold=0.1$ we have
\begin{align}\label{eq:multi-unit bk:mhr:balance market:concentrate bound}
    \prob{\abs{X-\ExpectX}\ge \threshold}\le 2\cdot e^{-(n+\addBuyers+1)\cdot \threshold^2}= 2\cdot  e^{-0.01\cdot (n+\addBuyers+1)}
\end{align}
We first provide upper and lower bounds on 
$\lnexpect$. The upper bound of $\lnexpect$ is
\begin{align*}
    \lnexpect={}&\int_{\optquant}^1 -\ln \variablex\cdot p(\variablex)\dd \variablex
    \ge \int_{\optquant}^{0.8} -\ln \variablex\cdot p(\variablex)\dd \variablex
    \\\ge{}&
    -\ln(0.8)\cdot \prob{\frac{1}{e}+\frac{1}{e}\cdot\frac{1}{n}\le X\le 1}
    \ge0.22\cdot\left(1-2\cdot e^{-0.01\cdot(n+\addBuyers+1)}\right),
\end{align*}
 where the second inequality holds since $\frac{1}{e}\cdot\frac{1}{n}< 0.01$ when $n\ge 594$. The lower bound of $\lnexpect$ is
\begin{align*}
    \lnexpect=\int_{\optquant}^1 -\ln \variablex\cdot p(\variablex)\dd \variablex
    \le-\ln\frac{1}{e}\cdot \prob{\frac{1}{e}\le X\le 1}
    \le1
\end{align*}
Define $\funcC(\optquant)\triangleq{\lnexpect}/\left({(\ln \optquant)^2\cdot\left(1-\rprob-{\lnexpect}/{\ln \optquant}\right)}\right)$. Since $-1\le\ln\optquant\le-0.97$ and $0\le1-\rprob\le 2\cdot e^{-0.01\cdot(n+\addBuyers+1)}$, the upper and lower bounds of $\funcC(\optquant)$ are
\begin{align*}
    \frac{0.22\cdot\left(1-2\cdot e^{-0.01\cdot(n+\addBuyers+1)}\right)}{2\cdot e^{-0.01\cdot(n+\addBuyers+1)}+\frac{1}{0.97}}\le \funcC(\optquant)\le \frac{1}{(0.97)^2\cdot 0.22\cdot\left(1-2\cdot e^{-0.01\cdot(n+\addBuyers+1)}\right)}
\end{align*}
When $n\ge 594$, we have $0.2\le\funcC(\optquant)\le 5$. Therefore, we have 
\begin{align*}
    \abs{\frac{\dd \ln(\RevVCG_{n:n+\addBuyers}(\optquant))}{\dd \optquant}}\le \abs{-\frac{1}{\optquant}+\frac{1}{\optquant}\cdot5}
    \le4\cdot e
\end{align*}
Since $ \RevVCG_{n:n+\addBuyers}(\optquant)\le n+\addBuyers$, we know that $\abs{{\dd \RevVCG_{n:n+\addBuyers}(\optquant)}/{\dd\optquant}}\le 4\cdot e\cdot(n+\addBuyers)\le 20\cdot n$. Therefore, when monopoly quantile $\optquant \in [{1}/{e}, {1}/{e}\cdot\left(1+{1}/{n}\right)]$ and $n\ge594$, for every $\addBuyers\in[0.444\cdot n, 0.446\cdot n]$, we have \begin{align*}
        \RevVCG_{n:n+\addBuyers}(\optquant)\ge{} \RevVCG_{n:n+\addBuyers}(\frac{1}{e})-20\cdot n\cdot\frac{1}{e}\cdot\frac{1}{n}
        \ge \RevVCG_{n:n+\addBuyers}(\frac{1}{e})-8
    \end{align*}
    This completes the proof of \Cref{lem:VCG revenue super high reserve}.
\end{proof}

\begin{lemma}
\label{lem:VCG revenue e reserve}
    For any $n \geq 594$, let 
    $\addBuyers \triangleq \lceil \left(e^{1/e} - 1\right) n + 1.05 \ln n \rceil$.
    Then, $\RevVCG_{n:n+\addBuyers}(e) \geq n + 8$.
\end{lemma}

\begin{proof}[Proof of \Cref{lem:VCG revenue e reserve}]
Recall from \Cref{lem:VCG revenue derivative} that $
\coefA \triangleq (n+\addBuyers)\cdot(n+\addBuyers-1)\cdot\binom{n+\addBuyers-2}{n-1}$ and
\begin{align*}
    \RevVCG_{n:n+\addBuyers}\left(\reserve\right) 
    ={}& \coefA \cdot \optreserve\cdot\int_0^{\optquant} (1-\quant)^{\addBuyers-1}\cdot \quant^{n}\dd \quant +
    \coefA\cdot\frac{\optreserve}{\ln \optreserve}\cdot\displaystyle\int_{\optquant}^1 -\ln \quant\cdot(1-\quant)^{\addBuyers-1}\cdot \quant^{n}\dd \quant
    \\={} &\coefA\cdot\optreserve\cdot\sum_{\indexj=0}^{\addBuyers-1}\binom{\addBuyers-1}{\indexj}\cdot{(-1)}^\indexj\cdot\frac{{\optquant}^{n+\indexj+1}}{n+\indexj+1}
    \\&+
    \coefA\cdot\frac{\optreserve}{\ln \optreserve}\cdot\sum_{\indexj=0}^{\addBuyers-1}\binom{\addBuyers-1}{\indexj}\cdot{(-1)}^\indexj\cdot\left[\frac{1}{(n+\indexj+1)^2}+\frac{{\optquant}^{n+\indexj+1}\cdot \ln \optquant}{n+\indexj+1}-\frac{{\optquant}^{n+\indexj+1}}{(n+\indexj+1)^2}\right]
\end{align*}
According to \Cref{lem:VCG revenue e reserve calculation}, when $\reserve=e$ and $\optquant={1}/{e}$, we can obtain a lower bound for $\RevVCG_{n:n+\addBuyers}(\reserve)$ as follows:
\begin{align*}
\RevVCG_{n:n+\addBuyers}\left(\reserve\right) 
    ={}& \coefA\cdot e\cdot\sum_{\indexj=0}^{\addBuyers-1}\binom{\addBuyers-1}{\indexj}\cdot{(-1)}^\indexj\cdot\frac{1-{e}^{-(n+\indexj+1)}}{(n+\indexj+1)^2}\notag
    \\\ge{}& \coefA\cdot e\cdot(1-{e}^{-\sqrt{n}})\cdot\sum_{\indexj=0}^{\addBuyers-1}\binom{\addBuyers-1}{\indexj}\cdot{(-1)}^\indexj\cdot\frac{1}{(n+\indexj+1)^2}
\end{align*}
Denote $\AltBinSum\triangleq \sum_{\indexj=0}^{\addBuyers-1}\binom{\addBuyers-1}{\indexj}\cdot(-1)^{\indexj}\cdot\frac{1}{(n+\indexj+1)^2}$. 
Using the standard identity, we obtain $
\frac{1}{(n+\indexj+1)^2}
= \int_0^1 -\ln \quant\cdot {\quant}^{n+\indexj}\dd \quant.$
Substituting this representation into the sum gives
\begin{align*}
\AltBinSum={}&\sum_{\indexj=0}^{\addBuyers-1}\binom{\addBuyers-1}{\indexj}\cdot(-1)^{\indexj}\cdot
   \int_0^1-\ln \quant\cdot {\quant}^{n+\indexj}\dd {\quant} 
   \\={}& \int_0^1 -\ln \quant\cdot {\quant}^{n}\cdot
   \left[\sum_{\indexj=0}^{\addBuyers-1}\binom{\addBuyers-1}{\indexj}\cdot(-1)^{\indexj}\cdot{\quant}^{\indexj}\right]\dd {\quant}
   =\int_0^1 -\ln \quant\cdot(1-\quant)^{\addBuyers-1}\cdot{\quant}^{n}\dd \quant
\end{align*}
The last equality holds since, by the binomial theorem, we know that 
$\sum_{\indexj=0}^{\addBuyers-1}\binom{\addBuyers-1}{\indexj}\cdot(-1)^{\indexj}\cdot {\quant}^{\indexj}
= (1-{\quant})^{\addBuyers-1}$.

Next we relate this integral to the Beta function. Recall that for $\alpha,\beta>0$,
\[
\BetaFun{\alpha}{\beta}\triangleq\displaystyle\int_0^1 \quant^{\alpha-1}\cdot(1-\quant)^{\beta-1}\mathrm{d}\quant
\,\,\text{and}\,\,
\int_0^1 -\ln \quant\cdot\quant^{\alpha-1}\cdot(1-\quant)^{\beta-1}\dd\quant
= -\frac{\partial}{\partial \alpha}\BetaFun{\alpha}{\beta}
\]
According to the standard derivative formula for the Beta function, let $\digammaf$ denote the digamma function. For positive integers $n$, the digamma 
function satisfies
$\digammaf(n+1) = \Harm{n} - \gamma$,
where $\gamma$ is the Euler--Mascheroni constant. Therefore,
\begin{align}\label{eq:bk vcg:mhr:beta function}
    \coefA\cdot  \AltBinSum={}& -\coefA\cdot \frac{\partial}{\partial (n+1)}\BetaFun{n+1}{\addBuyers}
    = a\cdot  \BetaFun{n+1}{\addBuyers}\cdot\left[\digammaf(n+\addBuyers+1)-\digammaf(n+1)\right]\notag
    \\={}&\coefA\cdot \frac{H_{n+\addBuyers}-H_n}{(n+\addBuyers)\cdot\binom{n+\addBuyers-1}{\addBuyers-1}}
    = n\cdot (\Harm{n+\addBuyers}-\Harm{n})
\end{align}
In our case, the value of $\RevVCG_{n:n+\addBuyers}\left(\reserve\right)$ becomes $\coefA\cdot e\cdot(1-{e}^{-\sqrt{n}})\cdot \AltBinSum=e\cdot(1-{e}^{-\sqrt{n}})\cdot n\cdot (\Harm{n+\addBuyers}-\Harm{n})$.
Finally, we need to prove that
\begin{align}\label{eq:multi-unit bk:MHR:revenue}
    e\cdot(1-{e}^{-\sqrt{n}})\cdot n\cdot (\Harm{n+\addBuyers}-\Harm{n})\ge n+8
\end{align}
Since $
\Harm{n+\addBuyers}-\Harm{n} = \sum_{\indexj=n+1}^{n+\addBuyers}{1}/{\indexj}$,
and the function ${1}/{x}$ is decreasing in $x$, we have the lower bound
$
\sum_{\indexj=n+1}^{n+\addBuyers}{1}/{\indexj}
\ge
\int_{n+1}^{n+\addBuyers+1}{1}/{x}\,\dd x
= \ln\left(1+{\addBuyers}/(n+1)\right).
$
Therefore, a sufficient condition for \Cref{eq:multi-unit bk:MHR:revenue} is
$\ln\left(1+\frac{\addBuyers}{n+1}\right)\ge({n+8})/({e\cdot(1-{e}^{-\sqrt{n}})})$,
or equivalently,
$\addBuyers=\left\lceil \left(e^{1/e}-1\right)\cdot n + 1.05\cdot\ln n \right\rceil \ge (n+1)\cdot\left(\exp\left(({1+\frac{8}{n}})/({e\cdot(1-e^{-\sqrt{n}})})\right)-1\right)$.

By \Cref{lem:VCG revenue e reserve calculation}, this inequality holds, and hence \Cref{eq:multi-unit bk:MHR:revenue} follows.
This completes the proof of \Cref{lem:VCG revenue e reserve}.
\end{proof}

\begin{lemma}\label{lem:VCG revenue e reserve calculation}
For any $n \geq 594$, it holds that $\addBuyers=\left\lceil \left(e^{1/e}-1\right)\cdot n + 1.05\cdot\ln n \right\rceil \ge (n+1)\cdot(e^{\frac{1+\frac{8}{n}}{e\cdot(1-e^{-\sqrt{n}})}}-1)$.
\end{lemma}
\begin{proof}[Proof of \Cref{lem:VCG revenue e reserve calculation}]
We will first prove the inequality $(n+1)\cdot e^{\frac{1+\frac{8}{n}}{e\cdot\left(1-e^{-\sqrt{n}}\right)}}
- (n+1)\cdot e^{\frac{1+\frac{8}{n}}{e}}
\le1$ for every integer $n$:
\begin{align*}
(n+1)\cdot e^{\frac{1+\frac{8}{n}}{e\cdot\left(1-e^{-\sqrt{n}}\right)}}
- (n+1)\cdot e^\frac{1+\frac{8}{n}}{e} 
={} (n+1) \cdot e^{\frac{1+\frac{8}{n}}{e}}\cdot\left(
e^{\frac{e^{-\sqrt{n}-1}}{1-e^{-\sqrt{n}}}}-1
\right)
\end{align*}
Since $e^{-\sqrt{n}}\le e^{-1}\le \frac{1}{2}$, we obtain $(e^{-\sqrt{n}-1})/(1-e^{-\sqrt{n}})\le 2\cdot{e^{-\sqrt{n}-1}}\le 1$.
For any $\quant\in[0,1]$ the elementary inequality $e^{\quant}-1\le 2\cdot\quant$ holds, hence $
e^{\frac{e^{-\sqrt{n}-1}}{1-e^{-\sqrt{n}}}}-1 \le e^{2\cdot{e^{-\sqrt{n}-1}}}-1 \le  4\cdot {e^{-\sqrt{n}-1}}$.
Therefore, 
\[
(n+1)\cdot e^{\frac{1+\frac{8}{n}}{e}}\cdot(e^{\frac{1}{e}\cdot\frac{e^{-\sqrt{n}}}{1-e^{-\sqrt{n}}}}-1)
\le  (n+1)\cdot e^{\frac{1+\frac{8}{n}}{e}}\cdot 4\cdot {e^{-\sqrt{n}-1}}
= 4\cdot (n+1)\cdot e^{\frac{1+\frac{8}{n}}{e}-1}\cdot {e^{-\sqrt{n}}}
\]
By differentiating the function, we observe that $
4\cdot (n+1)\cdot e^{\frac{1+\frac{8}{n}}{e}-1}\cdot e^{-\sqrt{n}} \le 1$
for all $n\ge 594$. Hence,
\[
(n+1)\cdot\left(e^{\frac{1+\frac{8}{n}}{e\cdot(1-{e}^{-\sqrt{n}})}}-1\right) \le (n+1)\cdot (e^{\frac{1+\frac{8}{n}}{e}}-1) + 1, \text{ }
\left\lceil (n+1)\cdot\left(e^{\frac{1+\frac{8}{n}}{e\cdot(1-{e}^{-\sqrt{n}})}}-1\right)\right\rceil \le (n+1)\cdot (e^{\frac{1+\frac{8}{n}}{e}}-1) + 2.
\]
Second, we will prove that $(n+1)\cdot (e^{\frac{1+\frac{8}{n}}{e}}-1)-n\cdot (e^{\frac{1}{e}}-1)\le 1$. Since $e^{\frac{1+\frac{8}{n}}{e}}-1\le 1$ when $n\ge 594$, we have
$(n+1)\cdot (e^{\frac{1+\frac{8}{n}}{e}}-1)-n\cdot (e^{\frac{1}{e}}-1)\le n\cdot e^{\frac{1}{e}}\cdot(e^{\frac{1}{e}\cdot\frac{8}{n}}-1)+1$,
define $\Auxfunc(n)\triangleq n\cdot e^{\frac{1}{e}}\cdot(e^{\frac{1}{e}\cdot\frac{8}{n}}-1)$, we know that 
\begin{align*}
\Auxfunc^\prime(n)=e^{\frac{1}{e}}\cdot\left(e^{\frac{1}{e}\cdot\frac{8}{n}}\cdot\left(1-\frac{1}{e}\cdot\frac{8}{n}\right)-1\right)
\le e^{\frac{1}{e}}\cdot\left(e^{\frac{1}{e}\cdot\frac{8}{n}}\cdot e^{-\frac{1}{e}\cdot\frac{8}{n}}-1\right)
=0
\end{align*} 
Inequality holds by the convexity inequality $1-\variablex\le e^{-\variablex}$ for all $\variablex\ge 0$.
Therefore, $\Auxfunc(n)$ is decreasing in $n$.
When $n\ge 594$, $\Auxfunc(n)\le 1$, therefore, when $n\ge 594$, we have 
\[
\left\lceil n\cdot\left(e^{\frac{1+\frac{1}{n}}{e\cdot(1-{e}^{-\sqrt{n}})}}-1\right)\right\rceil \le n\cdot (e^{\frac{1}{e}}-1) + 4\le \addBuyers.
\]
This completes the proof of \Cref{lem:VCG revenue e reserve calculation}.
\end{proof}

\begin{lemma}\label{lem:VCG revenue e reserve inequality}
For any $n \geq 594$, let 
    $\addBuyers \triangleq \lceil \left(e^{1/e} - 1\right) n + 1.05 \ln n \rceil$.
    Then, 
    \begin{align*}
     \sum_{\indexj=0}^{\addBuyers-1}\binom{\addBuyers-1}{\indexj}\cdot{(-1)}^\indexj\cdot\frac{1-{e}^{-(n+\indexj+1)}}{(n+\indexj+1)^2}\notag
    \ge{} (1-{e}^{-\sqrt{n}})\cdot\sum_{\indexj=0}^{\addBuyers-1}\binom{\addBuyers-1}{\indexj}\cdot{(-1)}^\indexj\cdot\frac{1}{(n+\indexj+1)^2}.
    \end{align*}
\end{lemma}
\begin{proof}[Proof of \Cref{lem:VCG revenue e reserve inequality}] Let $n\ge594$ and $\addBuyers \triangleq \lceil \left(e^{1/e} - 1\right) n + 1.05 \ln n \rceil$. Define
\[
\coefOld
\triangleq
\sum_{\indexj=0}^{\addBuyers-1}
\binom{\addBuyers-1}{\indexj} \cdot (-1)^{\indexj}
\cdot
\frac{e^{-(n+\indexj+1)}}{(n+\indexj+1)^2}
\,\,\text{and}\,\,
\coefNew
\triangleq
\sum_{\indexj=0}^{\addBuyers-1}
\binom{\addBuyers-1}{\indexj} \cdot (-1)^{\indexj}
\cdot
\frac{e^{-\sqrt{n}}}{(n+\indexj+1)^2}
\]
Then we will show that $\coefNew>\coefOld>0$, which directly ensures \Cref{lem:VCG revenue e reserve inequality} holds.

Define $\intOld(n,\addBuyers)\triangleq\int_0^1
\quant^n
\cdot
(1-e^{-1}\cdot\quant)^{\addBuyers-1}
\cdot
(-\ln\quant)\, \dd\quant$ and $\intNew(n,\addBuyers)\triangleq\int_0^1
\quant^n
\cdot
(1-\quant)^{\addBuyers-1}
\cdot
(-\ln\quant)\, \dd\quant$, both $\intOld$ and $\intNew$ are strictly positive because every factor in their integrands is positive.
Using
\begin{align*}
\frac{1}{(n+\indexj+1)^2}
=
\int_0^1
\quant^{n+\indexj} \cdot (-\ln \quant)\, \dd\quant
\,\,\,\text{and}\,\,\,
\sum_{\indexj=0}^{\addBuyers-1}
\binom{\addBuyers-1}{\indexj} \cdot (-1)^{\indexj}
\cdot
\quant^{\indexj}
=
(1-\quant)^{\addBuyers-1},
\end{align*}
 we obtain
\begin{align*}
\coefOld
&=
e^{-(n+1)}
\cdot
\int_0^1
\quant^n
\cdot
(1-e^{-1}\cdot\quant)^{\addBuyers-1}
\cdot
(-\ln\quant)\, \dd\quant=
e^{-(n+1)} \cdot \intOld(n,\addBuyers),
\\
\coefNew
&=
e^{-\sqrt{n}}
\cdot
\int_0^1
\quant^n
\cdot
(1-\quant)^{\addBuyers-1}
\cdot
(-\ln\quant)\, \dd\quant
=
e^{-\sqrt{n}} \cdot \intNew(n,\addBuyers)
\end{align*}
Therefore, \Cref{lem:VCG revenue e reserve inequality} holds if and only if
\begin{align*}
\frac{\coefOld}{\coefNew}
=
e^{-(n+1-\sqrt{n})}
\cdot
\frac{\intOld(n,\addBuyers)}{\intNew(n,\addBuyers)}\le 1
\end{align*}

Since $(1-e^{-1}\quant)^{\addBuyers-1}\le 1$, we obtain the upper bound $
\intOld(n,\addBuyers)
\le
\int_0^1
\quant^n \cdot (-\ln \quant)\, \dd\quant
=
\frac{1}{(n+1)^2}$.
According to \Cref{eq:bk vcg:mhr:beta function}, we know that
\begin{align*}
\intNew(n,\addBuyers)
= \frac{\partial}{\partial (n+1)}\BetaFun{n+1}{\addBuyers}=
\frac{n! \cdot (\addBuyers-1)!}{(n+\addBuyers)!}
\cdot
(H_{n+\addBuyers}-H_m)
\ge
\frac{n! \cdot (\addBuyers-1)!}{(n+\addBuyers)!}
\cdot
\frac{\addBuyers}{n+\addBuyers},
\end{align*}
the last inequality holds since $H_{n+\addBuyers}-H_n
=
\sum_{i=n+1}^{n+\addBuyers}
{1}/{i}
\ge
{\addBuyers}/(n+\addBuyers)$. Therefore,
\begin{align*}
\frac{\intOld}{\intNew}
\le
\frac{n+\addBuyers}{(n+1)^2}
\cdot
\binom{n+\addBuyers}{\addBuyers}
\end{align*}
Since when $n\ge 594$,  $\addBuyers= \lceil \left(e^{1/e} - 1\right) n + 1.05 \ln n \rceil\le \frac{n}{2}$, the standard entropy bound yields
$\binom{n+\addBuyers}{\addBuyers}
\le
\exp\left(
\tfrac{3}{2}\cdot n \cdot \binEntropy(1/3)
\right)$,
where $\binEntropy(p)$ is the binary entropy function defined by
$
\binEntropy(\quant)\triangleq
-\quant \cdot\ln \quant
-
(1-\quant)\cdot\ln(1-\quant)$ for $\quant\in(0,1)$.
In particular, $\binEntropy(1/3)\approx 0.6365$. Therefore, 
\begin{align*}
\frac{\coefOld}{\coefNew}
\le
\exp\left(
-(n+1-\sqrt{n})
+
\frac{3}{2} \cdot n \cdot \binEntropy(1/3)
+
\ln \frac{n+\addBuyers}{(n+1)^2}
\right)
\le
\exp\left(
-0.04 \cdot n + \sqrt{n} + \ln \frac{3}{2\cdot n}
\right)
\le1
\end{align*}
holds when $n>594$.
Therefore $\coefNew > \coefOld$ when $n\ge 594$ and $\addBuyers\le \frac{n}{2}$, as claimed, which directly ensures \Cref{lem:VCG revenue e reserve inequality} holds.
\end{proof}

We are now ready to prove \Cref{thm:bk vcg:mhr finite market}.

\begin{proof}[Proof of \Cref{thm:bk vcg:mhr finite market}]
    For $n\leq 593$, the numerical computation in \Cref{sec:numerical experiment} proves both the lower and upper bounds. It remains to prove the statement for $n\geq 594$.

    \xhdr{Upper bound.}
    We next prove the upper bound in the theorem statement for $n\geq 594$. Invoking \Cref{prop:bk vcg:mhr worst case}, it suffices to show that by adding $\addBuyers=\lceil(e^{1/e}-1)n+1.05\ln n\rceil$ additional buyers, 
    \begin{align*}
        \RevVCG_{n:n+\addBuyers}(\buyerdist_{(0, \reserve)}) \geq \RevOPT_{n:n}(\buyerdist_{(0, \reserve)})
    \end{align*}
    for all $\reserve\in[1, e]$. 
    By construction, we know that 
    \begin{align*}
        \RevOPT_{n:n}(\buyerdist_{(0, \reserve)}) = n 
        \;\;
        \text{for all $\reserve\in[1,e]$}
    \end{align*}
    Invoking \Cref{lem:VCG revenue small reserve}, we have that for all $\reserve\in[1,5/4]$,
    \begin{align*}
        \RevVCG_{n:n+\addBuyers}(\buyerdist_{(0, \reserve)})
        \geq 
        \RevVCG_{n:n+\addBuyers}(\buyerdist_{(0, 1)})
        = n = \RevOPT_{n:n}(\buyerdist_{(0, \reserve)})
    \end{align*}
    Invoking \Cref{lem:VCG revenue middle reserve}, we have that for all $\reserve\in[5/4,2]$,
    \begin{align*}
        \RevVCG_{n:n+\addBuyers}(\buyerdist_{(0, \reserve)})
        \geq 
        n = \RevOPT_{n:n}(\buyerdist_{(0, \reserve)})
    \end{align*}
    Invoking \Cref{lem:VCG revenue high reserve,lem:VCG revenue super high reserve,lem:VCG revenue e reserve},  we have that for all $\reserve\in[2,e/(1+1/n)]$,
    \begin{align*}
        \RevVCG_{n:n+\addBuyers}(\buyerdist_{(0, \reserve)})
        &\geq 
        \RevVCG_{n:n+\addBuyers}(\buyerdist_{(0, e/(1+1/n))})
        \\
        &\geq 
        \RevVCG_{n:n+\addBuyers}(\buyerdist_{(0, e)}) - 1
        \geq 
        n + 1 - 1 = \RevOPT_{n:n}(\buyerdist_{(0, \reserve)})
    \end{align*}
    and for all $\reserve\in[e/(1+1/n),e]$,
    \begin{align*}
        \RevVCG_{n:n+\addBuyers}(\buyerdist_{(0, \reserve)})
        \geq 
        \RevVCG_{n:n+\addBuyers}(\buyerdist_{(0, e)}) - 1
        \geq 
        n + 1 - 1 = \RevOPT_{n:n}(\buyerdist_{(0, \reserve)})
    \end{align*}
    Combining all the pieces, we prove the upper bound stated in the theorem as desired.

    \xhdr{Lower bound.} For the lower bound, we directly consider the $e$-truncated exponential distribution $\buyerdist_{(0, e)}$. Invoking \Cref{lem:prelim:revenue expression in uniform order statistic}, the revenue of the {\VCGAuction} with $k$ additional buyers can be expressed as 
    \begin{align*}
        \RevVCG_{n:n+k}(\buyerdist_{(0, e)}) &= 
        \left(n+k\right)\cdot \displaystyle\int_0^1 \revcurve_{(0,e)}\left(\quant\right)
        \cdot \osdensity_{n:n+k - 1}\left(\quant\right)\,\dd \quant \notag
        \\
        &\overset{(a)}{\leq}
         \left(n+k\right)\cdot \revcurve_{(0,e)}\left(\frac{n}{n+k}\right)
         \overset{(b)}{=}e\cdot n\cdot \ln\left(\frac{n+k}{n}\right)
    \end{align*}
    where inequality~(a) holds since revenue curve $\revcurve_{(0,e)}$ induced from MHR distribution $\buyerdist_{(0, e)}$ is concave and the Jensen's inequality, and equality~(b) holds due to the construction of revenue curve $\revcurve_{(0,e)}$. By the definition of the competition complexity, it satisfies that 
    \begin{align*}
        \RevVCG_{n:n+k}(\buyerdist_{(0, e)}) \geq 
        \RevOPT_{n:n}(\buyerdist_{(0, e)})
    \end{align*}
    which implies
    \begin{align*}
        e\cdot n\cdot \ln\left(\frac{n+\ccomplexity\left(n,\MHRDists\right)}{n}\right)
        \geq n 
        \;\;
        \Longleftrightarrow
        \;\;
        \ccomplexity\left(n,\MHRDists\right) \geq 
        \left(e^{\frac{1}{e}}-1\right)\cdot n
    \end{align*}
    which completes the analysis of the lower bound and finishes the proof of \Cref{thm:bk vcg:mhr finite market} as desired.
\end{proof}

\subsection{Proof of \texorpdfstring{\Cref{prop:bk supply limiting:worst case}}{Proposition 5.1}}
\label{apx:propVCGSLWorstDists}

\propVCGSLWorstDists*

\begin{proof}[Proof of \Cref{prop:bk supply limiting:worst case}]
    With a slight abuse of notation, we define 
    \begin{align*}
        \ccomplexitySL\left(\supply, m, n, \DistClass, \ccapproxratio\right) 
        &\triangleq 
        \min\left\{k \in \naturals :
        \forall\,\buyerdist \in \DistClass,\,
        \RevSLVCG_{\supply:n+k}\left(\buyerdist\right) \geq 
        \ccapproxratio \cdot \RevOPT_{m:n}\left(\buyerdist\right)
        \right\}.
    \end{align*}
    It suffices to prove that for all $\supply\in[m]$,
    \begin{align*}
        \ccomplexitySL\left(\supply,m,n,\lambdaDists,\ccapproxratio\right) = \ccomplexitySL\left(\supply,m,n,\truncGPDists,\ccapproxratio\right)
    \end{align*}
    Define auxiliary notation $k \triangleq \ccomplexitySL\left(\supply, m,n,\lambdaDists, \ccapproxratio\right) - 1$.
    By the definition of the competition complexity, there exists an $\Reglevel$-regular distribution $\buyerdist\primed$ such that 
    \begin{align*}
        \RevVCG_{\supply:n+k}\left(\buyerdist\primed\right)
        <
        \ccapproxratio\cdot \RevOPT_{m:n}
        \left(\buyerdist\primed\right)
    \end{align*}
    Without loss of generality, we assume that the monopoly revenue in distribution $\buyerdist\primed$ is one. Invoking Proposition~4 in \citet{SS-19}, since the monopoly quantile of any $\Reglevel$-regular distribution is at least ${(1-\Reglevel)}^{{1}/{\Reglevel}}$, the monopoly reserve $\optreserve$ of distribution $\buyerdist\primed$ satisfies $\optreserve\in[1,{({1-\Reglevel})}^{-{1}/{\Reglevel}}]$. 

    Let $\revcurve\primed$ be the revenue curve of distribution $\buyerdist\primed$.
    Invoking \Cref{lem:prelim:revenue expression in uniform order statistic}, we express the revenue of the {\BayesianOptimalMech} and the {\VCGAuction} as follows (recall $\osdensity_{a:b}$ is the PDF of $a$-th smallest order statistics of $b$ i.i.d.\ samples from uniform distribution $U[0,1]$):
    \begin{align*}
        \RevOPT_{m:n}\left(\buyerdist\primed\right) 
        & = 
        n\cdot \displaystyle\int_0^1 \revcurve\primed\left(\min\{\quant,\optquant\}\right)\cdot \osdensity_{m:n - 1}\left(\quant\right)\,\dd\quant
        \\
        & = 
        n\cdot \left(\displaystyle\int_0^{\optquant} \revcurve\primed\left(\quant\right) \cdot \osdensity_{m:n - 1}\left(\quant\right)\,\dd\quant 
        +
        \int_{\optquant}^1 \revcurve\primed\left(\optquant\right) \cdot \osdensity_{m:n - 1}\left(\quant\right)\,\dd\quant
        \right)
    \end{align*}
    and
    \begin{align*}
        \RevVCG_{\supply:n+k}\left(\buyerdist\primed\right) 
        & = 
        \left(n + k\right)\cdot \displaystyle\int_0^1 \revcurve\primed\left(\quant\right) \cdot \osdensity_{\supply:n + k - 1}\left(\quant\right)\,\dd\quant 
        \\
        & = 
        \left(n + k\right)\cdot \left(\displaystyle\int_0^{\optquant} \revcurve\primed\left(\quant\right) \cdot \osdensity_{\supply:n + k - 1}\left(\quant\right)\,\dd\quant 
        +
        \displaystyle\int_{\optquant}^1 \revcurve\primed\left(\quant\right) \cdot \osdensity_{\supply:n + k - 1}\left(\quant\right)\,\dd\quant 
        \right)
    \end{align*}
    We utilize the following technical lemma about the PDF of order statistics from the uniform distribution.
    \begin{lemma}[Single-crossing property]
    \label{lem:bk supply limiting:osdensity single crossing}
        Given any $m,n,k\in\naturals$ ($m\leq n$), there exists threshold quantile $\quantSC\in[0, 1]$ such that
        \begin{align*}
            \forall \quant \in [0, \quantSC]:&
            \qquad
            \left(n + k\right)\cdot \osdensity_{\supply:n + k - 1}\left(\quant\right)
            \geq 
            \ccapproxratio\cdot n\cdot \osdensity_{m:n - 1}\left(\quant\right)
            \\
            \forall \quant \in [\quantSC, 1]:&
            \qquad
            \left(n + k\right)\cdot \osdensity_{\supply:n + k - 1}\left(\quant\right)
            \leq 
            \ccapproxratio\cdot n\cdot \osdensity_{m:n - 1}\left(\quant\right)
        \end{align*}
    \end{lemma}
    \begin{proof}
        Consider the (normalized) difference between the left-hand side and right-hand side of the inequality in the lemma statement:
        \begin{align*}
            & 
            \frac{\left(n + k\right)}{\ccapproxratio\cdot n}\cdot \osdensity_{\supply:n + k - 1}\left(\quant\right)
            - 
            \osdensity_{m:n - 1}\left(\quant\right)
            \\
            ={}&
            \frac{\left(n + k\right)}{\ccapproxratio\cdot n}
            \cdot 
            \left(n + k - 1\right)\cdot \binom{n + k - 2}{\supply - 1}\cdot
            \left(1 - \quant\right)^{n + k - 1 - \supply}\cdot \quant^{\supply - 1}
            -
            \left(n - 1\right)\cdot \binom{n - 2}{m - 1}
            \cdot 
            \left(1 - \quant\right)^{n - 1 - m}\cdot \quant^{m - 1}
            \\
            ={} &
            \left(1 - \quant\right)^{n - 1 - m}\cdot \quant^{\supply - 1} \cdot 
            \left(
            \frac{\left(n + k\right)}{\ccapproxratio\cdot n}
            \cdot 
            \left(n + k - 1\right)\cdot \binom{n + k - 2}{\supply - 1}\cdot
            \left(1 - \quant\right)^{m+k-\supply}
            -
            \left(n - 1\right)\cdot \binom{n - 2}{m - 1}\cdot{\quant}^{m-\supply}
            \right)
        \end{align*}
        Note that the difference is zero for boundary quantiles $\quant \in \{0, 1\}$. 
        In addition,
        \begin{align*}
            \frac{\left(n + k\right)}{\ccapproxratio\cdot n}
            \cdot 
            \left(n + k - 1\right)\cdot \binom{n + k - 2}{\supply - 1}\cdot
            \left(1 - \quant\right)^{m+k-\supply}
            -
            \left(n - 1\right)\cdot \binom{n - 2}{m - 1}\cdot{\quant}^{m-\supply}
        \end{align*}
        is decreasing in $\quant\in[0, 1]$. Therefore, there exists at most one quantile $\quantSC \in \left(0, 1\right)$ such that the difference equals zero. Moreover, the difference is weakly positive (resp.\ weakly negative) for all quantiles smaller than (resp.\ larger than) $\quantSC$. This completes the proof of \Cref{lem:bk supply limiting:osdensity single crossing}.
    \end{proof}
    Invoking \Cref{lem:bk supply limiting:osdensity single crossing}, there exists threshold quantile $\quantSC\in[0, 1]$ such that
    \begin{align*}
            \forall \quant \in [0, \quantSC]:&
            \qquad
            \left(n + k\right)\cdot \osdensity_{\supply:n + k - 1}\left(\quant\right)
            \geq 
            \ccapproxratio\cdot n\cdot \osdensity_{m:n - 1}\left(\quant\right)
            \\
            \forall \quant \in [\quantSC, 1]:&
            \qquad
            \left(n + k\right)\cdot \osdensity_{\supply:n + k - 1}\left(\quant\right)
            \leq 
            \ccapproxratio\cdot n\cdot \osdensity_{m:n - 1}\left(\quant\right)
        \end{align*}
    Define auxiliary function $\auxfunc:[0, 1]\rightarrow\reals$ as follows:
    \begin{align*}
        \auxfunc\left(\quant\right) \triangleq \ccapproxratio\cdot n\cdot \osdensity_{m:n - 1}\left(\quant\right)
         - \left(n + k\right)\cdot \osdensity_{\supply:n + k - 1}\left(\quant\right)
    \end{align*}
    By construction, $\auxfunc\left(\quant\right) \leq 0$ for all  $\quant\in[0, \quantSC]$ and $\auxfunc\left(\quant\right) \geq 0$ for all $\quant\in[\quantSC,1]$.
    Consider the following two cases:

    \xhdr{Case (i) [Quantiles $\optquant\leq \quantSC$]:} In this case, 
    consider $\optreserve$-truncated $\Reglevel$-generalized Pareto distribution~$\WorstRegDist$. By construction, $\WorstRegDist$ has the same monopoly reserve, monopoly quantile, and monopoly revenue as distribution $\buyerdist\primed$. In addition, $\distLambdaReserve$ is first order stochastically dominated by $\buyerdist\primed$, and thus the induced revenue curves $\WorstRevcurve\left(\quant\right) \leq \revcurve\primed\left(\quant\right)$ for all quantiles $\quant\in[0, 1]$. Note that
    \begin{align*}
        &\ccapproxratio\cdot \RevOPT_{m:n}\left(\WorstRegDist\right)
        -
        \RevVCG_{\supply:n + k}\left(\WorstRegDist\right) 
        \\
        ={} & 
        \ccapproxratio\cdot n\cdot \left(\displaystyle\int_0^{\optquant} \WorstRevcurve\left(\quant\right) \cdot \osdensity_{m:n - 1}\left(\quant\right)\,\dd\quant 
        +
        \int_{\optquant}^1 \WorstRevcurve\left(\optquant\right) \cdot \osdensity_{m:n - 1}\left(\quant\right)\,\dd\quant
        \right)
        \\
        & \qquad 
        - 
        \left(n + k\right)\cdot \left(\displaystyle\int_0^{\optquant} \WorstRevcurve\left(\quant\right) \cdot \osdensity_{\supply:n + k - 1}\left(\quant\right)\,\dd\quant 
        +
        \displaystyle\int_{\optquant}^1 \WorstRevcurve\left(\quant\right) \cdot \osdensity_{\supply:n + k - 1}\left(\quant\right)\,\dd\quant 
        \right)
        \\
        ={} & 
        \displaystyle\int_0^{\optquant}
        \WorstRevcurve\left(\quant\right)
        \cdot 
        \auxfunc\left(\quant\right)
        \,
        \dd \quant
+
        \ccapproxratio\cdot n\cdot \int_{\optquant}^1 \WorstRevcurve\left(\optquant\right) \cdot \osdensity_{m:n - 1}\left(\quant\right)\,\dd\quant
       \\
       &\qquad 
       -
        \left(n + k\right)\cdot \displaystyle\int_{\optquant}^1 \WorstRevcurve\left(\quant\right) \cdot \osdensity_{\supply:n + k - 1}\left(\quant\right)\,\dd\quant 
        \\
        \geq{} & 
        \displaystyle\int_0^{\optquant}
        \revcurve\primed\left(\quant\right)
        \cdot 
        \auxfunc\left(\quant\right)
        \,
        \dd \quant
+
        \ccapproxratio\cdot n\cdot \int_{\optquant}^1 \revcurve\primed\left(\optquant\right) \cdot \osdensity_{m:n - 1}\left(\quant\right)\,\dd\quant
        -
        \left(n + k\right)\cdot \displaystyle\int_{\optquant}^1 \revcurve\primed\left(\quant\right) \cdot \osdensity_{\supply:n + k - 1}\left(\quant\right)\,\dd\quant 
        \\
        ={} & 
        \ccapproxratio\cdot \RevOPT_{m:n}\left(\WorstRegDist\primed\right)
        -
        \RevVCG_{\supply:n + k}\left(\buyerdist\primed\right) 
        \\
        \geq {}& 0
    \end{align*}
    where the first inequality holds since (i) revenue curves satisfy $\WorstRevcurve\left(\quant\right)\leq \revcurve\primed\left(\quant\right)$ for all quantiles $\quant\in[0, 1]$, $\WorstRevcurve\left(\optquant\right) = \revcurve\primed\left(\optquant\right)$, and (ii) auxiliary function $\auxfunc\left(\quant\right)
    \leq 0$ for all quantiles $\quant\in[0,\optquant]\subseteq[0,\quantSC]$ (\Cref{lem:bk supply limiting:osdensity single crossing}). This completes the proof of the proposition statement for Case (i).

    \xhdr{Case (ii) [Quantiles $\optquant\geq \quantSC$]:}
    In this case, 
    consider $\optreserve$-truncated $\Reglevel$-generalized Pareto distribution~$\WorstRegDist$. By construction, $\WorstRegDist$ has the same monopoly reserve, monopoly quantile, and monopoly revenue as distribution $\buyerdist\primed$. In addition, $\WorstRegDist$ is first order stochastically dominated by $\buyerdist\primed$.
    Let $\valSC$ be the value such that $1 - \buyerdist\primed\left(\valSC\right) = \quantSC$. The revenue curve $\WorstRevcurve$ induced by distribution $\WorstRegDist$ satisfies 
    \begin{align*}
        \forall \quant\in[0,\quantSC]:&\qquad 
        \WorstRevcurve\left(\quant\right) \leq \frac{\optreserve}{\valSC}\cdot  \revcurve\primed\left(\quant\right) \\
        \forall \quant\in[\quantSC,\optquant]:& \qquad 
        \WorstRevcurve\left(\quant\right) \geq 
        \frac{\optreserve}{\valSC}\cdot \revcurve\primed\left(\quant\right)\\
        \forall \quant\in[\optquant, 1]:& \qquad 
        \WorstRevcurve\left(\quant\right) \leq \revcurve\primed\left(\quant\right)
    \end{align*}
    and $\WorstRevcurve\left(\optquant\right) = \revcurve\primed\left(\optquant\right)$,
    where the first and second inequalities holds since distribution $\buyerdist\primed$ is $\Reglevel$-regular and thus the revenue curve $\revcurve\primed$ is concave. Note that 
    \begin{align*}
        &\ccapproxratio\cdot \RevOPT_{m:n}\left(\WorstRegDist\right)
        -
        \RevVCG_{\supply:n + k}\left(\WorstRegDist\right) 
        \\    
        ={} & 
        \ccapproxratio\cdot n\cdot \left(\displaystyle\int_0^{\optquant} \WorstRevcurve\left(\quant\right) \cdot \osdensity_{m:n - 1}\left(\quant\right)\,\dd\quant 
        +
        \int_{\optquant}^1 \WorstRevcurve\left(\optquant\right) \cdot \osdensity_{m:n - 1}\left(\quant\right)\,\dd\quant
        \right)
        \\
        & \qquad 
        - 
        \left(n + k\right)\cdot \left(\displaystyle\int_0^{\optquant} \WorstRevcurve\left(\quant\right) \cdot \osdensity_{\supply:n + k - 1}\left(\quant\right)\,\dd\quant 
        +
        \displaystyle\int_{\optquant}^1 \WorstRevcurve\left(\quant\right) \cdot \osdensity_{\supply:n + k - 1}\left(\quant\right)\,\dd\quant 
        \right)
        \\
        ={} & 
        \ccapproxratio\cdot n\cdot \left(\displaystyle\int_0^{\quantSC} \WorstRevcurve\left(\quant\right) \cdot \osdensity_{m:n - 1}\left(\quant\right)\,\dd\quant 
        +
        \displaystyle\int_{\quantSC}^{\optquant} \WorstRevcurve\left(\quant\right) \cdot \osdensity_{m:n - 1}\left(\quant\right)\,\dd\quant 
        +
        \int_{\optquant}^1 \WorstRevcurve\left(\optquant\right) \cdot \osdensity_{m:n - 1}\left(\quant\right)\,\dd\quant
        \right)
        \\
        & \qquad 
        - 
        \left(n + k\right)\cdot \left(\displaystyle\int_0^{\quantSC} \WorstRevcurve\left(\quant\right) \cdot \osdensity_{
        \supply:n + k - 1}\left(\quant\right)\,\dd\quant 
        +
        \displaystyle\int_{\quantSC}^{\optquant} \WorstRevcurve\left(\quant\right) \cdot \osdensity_{\supply:n + k - 1}\left(\quant\right)\,\dd\quant \right.
        \\
        &
        \hspace{10cm}
        \left.
        +
        \displaystyle\int_{\optquant}^1 \WorstRevcurve\left(\quant\right) \cdot \osdensity_{\supply:n + k - 1}\left(\quant\right)\,\dd\quant 
        \right)
        \\
        ={} & 
        \displaystyle\int_0^{\quantSC}
        \WorstRevcurve\left(\quant\right)
        \cdot 
        \auxfunc\left(\quant\right)
        \,
        \dd \quant
+
        \displaystyle\int_{\quantSC}^{\optquant}
        \WorstRevcurve\left(\quant\right)
        \cdot 
        \auxfunc\left(\quant\right)
        \,
        \dd \quant
        \\
        & \qquad
        +
        \ccapproxratio\cdot n\cdot \int_{\optquant}^1 \WorstRevcurve\left(\optquant\right) \cdot \osdensity_{m:n - 1}\left(\quant\right)\,\dd\quant
        -
        \left(n + k\right)\cdot \displaystyle\int_{\optquant}^1 \WorstRevcurve\left(\quant\right) \cdot \osdensity_{\supply:n + k - 1}\left(\quant\right)\,\dd\quant 
        \\
        = {} & 
        \frac{\optreserve}{\valSC}
        \cdot 
        \left(
        \frac{\valSC}{\optreserve}
        \cdot 
        \displaystyle\int_0^{\quantSC}
        \WorstRevcurve\left(\quant\right)
        \cdot 
        \auxfunc\left(\quant\right)
        \,
        \dd \quant
        \right.
+
        \frac{\valSC}{\optreserve}\cdot 
        \displaystyle\int_{\quantSC}^{\optquant}
        \WorstRevcurve\left(\quant\right)
        \cdot 
        \auxfunc\left(\quant\right)
        \,
        \dd \quant
        \\
        & \left.\qquad
        +
        \frac{\valSC}{\optreserve}\cdot 
        \ccapproxratio\cdot n\cdot \int_{\optquant}^1 \WorstRevcurve\left(\optquant\right) \cdot \osdensity_{m:n - 1}\left(\quant\right)\,\dd\quant
        -
        \frac{\valSC}{\optreserve}\cdot 
        \left(n + k\right)\cdot \displaystyle\int_{\optquant}^1 \WorstRevcurve\left(\quant\right) \cdot \osdensity_{\supply:n + k - 1}\left(\quant\right)\,\dd\quant 
        \right)
        \\  
        \geq {} & 
        \frac{\optreserve}{\valSC}
        \cdot 
        \left(
        \displaystyle\int_0^{\quantSC}
        \revcurve\primed\left(\quant\right)
        \cdot 
        \auxfunc\left(\quant\right)
        \,
        \dd \quant
        \right.
+
        \displaystyle\int_{\quantSC}^{\optquant}
        \revcurve\primed\left(\quant\right)
        \cdot 
        \auxfunc\left(\quant\right)
        \,
        \dd \quant
        \\
        & \left.\qquad
        + {}~
        \ccapproxratio\cdot n\cdot \int_{\optquant}^1 \revcurve\primed\left(\optquant\right) \cdot \osdensity_{m:n - 1}\left(\quant\right)\,\dd\quant
        -
        \left(n + k\right)\cdot \displaystyle\int_{\optquant}^1 \revcurve\primed\left(\quant\right) \cdot \osdensity_{\supply:n + k - 1}\left(\quant\right)\,\dd\quant 
        \right)
        \\
        ={} & \frac{\optreserve}{\valSC}
        \cdot\left(\ccapproxratio\cdot \RevOPT_{m:n}\left(\buyerdist\primed\right)
        -
        \RevVCG_{\supply:n + k}\left(\buyerdist\primed\right) \right)
        \\
        \geq {} & 0
    \end{align*}
    Here, the third equality holds due to the construction of auxiliary function $\auxfunc$, and the first inequality holds since 
    \begin{align*}
        &\frac{\valSC}{\optreserve}
        \cdot 
        \displaystyle\int_0^{\quantSC}
        \WorstRevcurve\left(\quant\right)
        \cdot 
        \auxfunc\left(\quant\right)
        \,
        \dd \quant
        \overset{\left(a\right)}{\geq} 
        \displaystyle\int_0^{\quantSC}
        \revcurve\primed\left(\quant\right)
        \cdot 
        \auxfunc\left(\quant\right)
        \,
        \dd \quant
        \\
        &\frac{\valSC}{\optreserve}
        \cdot 
        \displaystyle\int_{\quantSC}^{\optquant}
        \WorstRevcurve\left(\quant\right)
        \cdot 
        \auxfunc\left(\quant\right)
        \,
        \dd \quant
        \overset{\left(b\right)}{\geq} 
        \displaystyle\int_{\quantSC}^{\optquant}
        \revcurve\primed\left(\quant\right)
        \cdot 
        \auxfunc\left(\quant\right)
        \,
        \dd \quant
        \\
        &\frac{\valSC}{\optreserve}\cdot 
        \ccapproxratio\cdot n\cdot \int_{\optquant}^1 \WorstRevcurve\left(\optquant\right) \cdot \osdensity_{m:n - 1}\left(\quant\right)\,\dd\quant
        -
        \frac{\valSC}{\optreserve}\cdot 
        \left(n + k\right)\cdot \displaystyle\int_{\optquant}^1 \WorstRevcurve\left(\quant\right) \cdot \osdensity_{\supply:n + k - 1}\left(\quant\right)\,\dd\quant 
        \\
        & \qquad \qquad\overset{\left(c\right)}{\geq}{} 
        \ccapproxratio\cdot n\cdot \int_{\optquant}^1 \WorstRevcurve\left(\optquant\right) \cdot \osdensity_{m:n - 1}\left(\quant\right)\,\dd\quant
        -
        \left(n + k\right)\cdot \displaystyle\int_{\optquant}^1 \WorstRevcurve\left(\quant\right) \cdot \osdensity_{\supply:n + k - 1}\left(\quant\right)\,\dd\quant 
        \\
        & \qquad \qquad\overset{\left(d\right)}{\geq}{} 
        \ccapproxratio\cdot n\cdot \int_{\optquant}^1 \revcurve\primed\left(\optquant\right) \cdot \osdensity_{m:n - 1}\left(\quant\right)\,\dd\quant
        -
        \left(n + k\right)\cdot \displaystyle\int_{\optquant}^1 \revcurve\primed\left(\quant\right) \cdot \osdensity_{\supply:n + k - 1}\left(\quant\right)\,\dd\quant 
    \end{align*}
    where inequality~(a) holds since $\WorstRevcurve\left(\quant\right) \leq \frac{\optreserve}{\valSC}\cdot \revcurve\primed\left(\quant\right)$ and $\auxfunc\left(\quant\right) \leq 0$ for all quantiles $\quant \in [0, \quantSC]$, inequality~(b) holds since $\WorstRevcurve\left(\quant\right) \geq \frac{\optreserve}{\valSC}\cdot \revcurve\primed\left(\quant\right)$ and $\auxfunc\left(\quant\right) \geq 0$ for all quantiles $\quant \in [\quantSC,\optquant]$, and inequality~(c) holds since $\frac{\valSC}{\optreserve} \geq 1$ (implied by the case assumption that $\quantSC \leq \optquant$ and the construction of $\valSC$) and $\ccapproxratio\cdot n\cdot \int_{\optquant}^1 \WorstRevcurve\left(\optquant\right) \cdot \osdensity_{m:n - 1}\left(\quant\right)\,\dd\quant
    -
    \left(n + k\right)\cdot \int_{\optquant}^1 \WorstRevcurve\left(\quant\right) \cdot \osdensity_{\supply:n + k - 1}\left(\quant\right)\,\dd\quant  \geq 0$ (implied by the fact that $\WorstRevcurve\left(\optquant\right) \geq \WorstRevcurve\left(\quant\right)$ for all $\quant\in[\optquant,1]$ and $\ccapproxratio\cdot n\cdot\osdensity_{m:n-1}\left(\quant\right)\,\dd\quant - \left(n+k\right)\cdot \osdensity_{\supply:n+k-1}\left(\quant\right)\geq 0$ for all $\quant\in[\optquant, 1]$), and inequality~(d) holds since $\WorstRevcurve\left(\optquant\right) = \revcurve\primed\left(\optquant\right)$, $\WorstRevcurve\left(\quant\right) \leq \revcurve\primed\left(\quant\right)$ for all $\quant\in[\optquant,1]$. This completes the proof of the proposition statement for Case~(ii). 

    Putting the analysis for the two cases together, we finish the proof of \Cref{prop:bk supply limiting:worst case}.
\end{proof} 
\subsection{Proof of \texorpdfstring{\Cref{thm:bk supply limiting:large market}}{Theorem 5.2}}
\label{apx:thmVCGSLLargeMarket}

\thmVCGSLLargeMarket*

\begin{proof}[Proof of \Cref{thm:bk supply limiting:large market}]
    Fix $\Reglevel\in[0, 1]$, $\imbalanceratio\in(0, 1]$, and $\ccapproxratio\in(0, 1)$. For ease of presentation, define auxiliary notations $m_n$ and $k_n$ as
    \begin{align*}
        m_n \triangleq \lceil\imbalanceratio\cdot n\rceil
        \;\;\mbox{and}\;\;
        k_n \triangleq \ccomplexitySLInfty(\imbalanceratio,\lambdaDists,\ccapproxratio)
    \end{align*}
    Invoking \Cref{prop:bk supply limiting:worst case}, $k_n$ is the smallest non-negative integer such that for all $\reserve\in[1,(1-\Reglevel)^{-1/\Reglevel}]$,
    \begin{align}
    \label{eq:bk supply-limiting:large market:condition}
        \RevVCG_{\supply:n+k_n}\left(\distLambdaReserve\right) \geq 
        \ccapproxratio\cdot 
        \RevOPT_{m_n:n}\left(\distLambdaReserve\right) 
    \end{align}
    where $\distLambdaReserve$ is the $\reserve$-truncated $\Reglevel$-generalized Pareto distribution (\Cref{def:truncated generalized Pareto distribution}). By construction, distribution $\Reglevel$ has monopoly reserve $\reserve$, monopoly quantile $1/\reserve$, and thus monopoly revenue of one. 

    Invoking \Cref{lem:prelim:revenue expression in uniform order statistic}, we express the revenue of the {\VCGAuction} and the {\BayesianOptimalMech} as follows (recall $\osdensity_{a:b}$ is the PDF of $a$-th smallest order statistics of $b$ i.i.d.\ samples from uniform distribution $U[0,1]$):
    \begin{align*}
        \lim_{n\rightarrow \infty}
        \RevVCG_{\supply:n+k_n}\left(\distLambdaReserve\right) 
        &= 
        \lim_{n\rightarrow \infty}
        \left(n+k_n\right)\cdot \displaystyle\int_0^1 \revcurve_{(\Reglevel,\reserve)}\left(\quant\right)
        \cdot \osdensity_{\supply:n+k_n - 1}\left(\quant\right)\,\dd \quant \notag
        \\
        &\overset{(a)}{=}
        \lim_{n\rightarrow \infty} \left(n+k_n\right)\cdot \revcurve_{(\Reglevel,\reserve)}\left(\frac{\supply}{n+k_n}\right)~,
    \end{align*}
    and
    \begin{align*}
        \lim_{n\rightarrow \infty}
        \RevOPT_{m_n:n}\left(\distLambdaReserve\right) 
        &= 
        \lim_{n\rightarrow \infty}
        n\cdot \displaystyle\int_0^1
        \revcurve_{(\Reglevel,\reserve)}\left(\min\left\{\quant,\frac{1}{\reserve}\right\}\right)\cdot \osdensity_{m_n:n-1}\left(\quant\right)\,\dd\quant \notag
        \\
        &\overset{(b)}{=}
        \lim_{n\rightarrow \infty} n\cdot \revcurve_{(\Reglevel,\reserve)}\left(\min\left\{\frac{m_n}{n},\frac{1}{\reserve}\right\}\right)~,
    \end{align*}
    where $\revcurve_{(\Reglevel,\reserve)}$ is the induced revenue curve of distribution $\distLambdaReserve$, and both equalities~(a) and (b) hold due to the concentration of the order statistics ($\expect[\quant\sim\osdensity_{\supply:n+k_n - 1}]{\quant}=\supply/(n+k_n)$ and $\expect[\quant\sim\osdensity_{m_n:n - 1}]{\quant}=m_n/n$) and the Lipschitz continuity (implied by the $\Reglevel$-regularity of $\distLambdaReserve$) of the revenue curve $\revcurve_{(\Reglevel,\reserve)}(\cdot)$ and $\revcurve_{(\Reglevel,\reserve)}(\min\{\cdot,1/\reserve\})$ at $\supply/(n+k_n)$ and $m_n/n$, respectively.

    Below, we consider two cases depending on the value of the supply-to-demand ratio $\imbalanceratio$ separately.

    \xhdr{Case (i) [Supply-to-demand ratio $\imbalanceratio\in[(1-\Reglevel)^{1/\Reglevel}, 1]$]:}
    In this case, we further consider three subcases depending on the value of truncation $\reserve$ (which is also the monopoly reserve by construction) in distribution $\distLambdaReserve$.
    
    \xhdr{Case (i.a) [Truncation $\reserve\in[1,n/m_n]$]:} In this subcase, the revenue curve $\revcurve_{(\Reglevel,\reserve)}$ satisfies 
    \begin{align*}
        \revcurve_{(\Reglevel,\reserve)}\left(\frac{\supply}{n+k_n}\right) = \reserve\cdot \frac{\supply}{n+k_n}
        \;\;
        \mbox{and}
        \;\;
        \revcurve_{(\Reglevel,\reserve)}\left(\min\left\{\frac{m_n}{n},\frac{1}{\reserve}\right\}\right)
        =
        \reserve\cdot \frac{m_n}{n}
    \end{align*}
    Therefore, 
    \begin{align*}
        \lim_{n\rightarrow\infty}~
        \frac{\RevVCG_{\supply:n+k_n}\left(\distLambdaReserve\right)}{\RevOPT_{m_n:n}\left(\distLambdaReserve\right)}
        =
        \lim_{n\rightarrow\infty}~\frac{\left(n+k_n\right)\cdot \reserve \cdot \frac{\supply}{n + k_n}}{n\cdot \reserve\cdot \frac{m_n}{n}}
        =
        \frac{\supply}{m_n}
    \end{align*}
    Let $\supply/m_n \ge \ccapproxratio$, then $\supply \ge \lceil \ccapproxratio \cdot m_n \rceil = \lceil \ccapproxratio \cdot \lceil \imbalanceratio \cdot n \rceil \rceil$. In this case, condition~\ref{eq:bk supply-limiting:large market:condition} holds for any $k_n \ge 0$, in particular for $k_n=0$.

    \xhdr{Case (i.b) [Truncation $\reserve \in(n/m_n, (n+k_n)/\supply]$]:} In this subcase, the revenue curve $\revcurve_{(\Reglevel,\reserve)}$ satisfies
    \begin{align*}
        \revcurve_{(\Reglevel,\reserve)}\left(\frac{\supply}{n+k_n}\right) = \reserve\cdot \frac{\supply}{n+k_n}
        \;\;
        \mbox{and}
        \;\;
        \revcurve_{(\Reglevel,\reserve)}\left(\min\left\{\frac{m_n}{n},\frac{1}{\reserve}\right\}\right)
        =
        1
    \end{align*}
    Therefore, 
    \begin{align*}
        \lim_{n\rightarrow\infty}~
        \frac{\RevVCG_{\supply:n+k_n}\left(\distLambdaReserve\right)}{\RevOPT_{m_n:n}\left(\distLambdaReserve\right)}
        =
        \lim_{n\rightarrow\infty}~\frac{\left(n+k_n\right)\cdot \reserve \cdot \frac{\supply}{n + k_n}}{n\cdot 1}
        =
        \frac{\supply}{n}\cdot \reserve > \frac{\supply}{m_n}
    \end{align*}
    As in the analysis of \xhdr{Case (i.a)}, when $\supply \ge \lceil \ccapproxratio \cdot m_n \rceil = \lceil \ccapproxratio \cdot \lceil \imbalanceratio \cdot n \rceil \rceil$, condition~\ref{eq:bk supply-limiting:large market:condition} holds for any $k_n \ge 0$, in particular for $k_n=0$.
    
    \xhdr{Case (i.c) [Truncation $\reserve \in((n+k_n)/\supply,(1-\Reglevel)^{-1/\Reglevel}]$]:} In this subcase, the revenue curve $\revcurve_{(\Reglevel,\reserve)}$ satisfies
    \begin{align*}
        &\revcurve_{(\Reglevel,\reserve)}\left(\frac{\supply}{n+k_n}\right) = \frac{\reserve}{\reserve^\Reglevel - 1}\cdot \left(\left(\frac{n+k_n}{\supply}\right)^\Reglevel-1\right)\cdot \frac{\supply}{n+k_n}
        \geq 
        \frac{1}{\Reglevel}
        \cdot 
        \left(\frac{1}{1-\Reglevel}\right)^{\frac{1-\Reglevel}{\Reglevel}}\cdot \left(\left(\frac{n+k_n}{\supply}\right)^\Reglevel-1\right)\cdot \frac{\supply}{n+k_n}
        \\
        \mbox{and}
        \;\;
        &\revcurve_{(\Reglevel,\reserve)}\left(\min\left\{\frac{m_n}{n},\frac{1}{\reserve}\right\}\right)
        =
        1 
    \end{align*}
    where the inequality holds since $\reserve/(\reserve^\Reglevel - 1)$ is decreasing in $\reserve\in((n+k_n)/\supply,(1-\Reglevel)^{-1/\Reglevel}]$.
    Therefore, 
    \begin{align*}
        \lim_{n\rightarrow\infty}~
        \frac{\RevVCG_{\supply:n+k_n}\left(\distLambdaReserve\right)}{\RevOPT_{m_n:n}\left(\distLambdaReserve\right)}
        \geq
        \lim_{n\rightarrow\infty}~\frac{(n+k_n)\cdot\frac{1}{\Reglevel}
        \cdot 
        \left(\frac{1}{1-\Reglevel}\right)^{\frac{1-\Reglevel}{\Reglevel}}\cdot \left(\left(\frac{n+k_n}{\supply}\right)^\Reglevel-1\right)\cdot \frac{\supply}{n+k_n}}{n\cdot 1}
        \ge \ccapproxratio
    \end{align*}
    where the last inequality holds when $\supply=\lceil \ccapproxratio \cdot\lceil\imbalanceratio\cdot n \rceil\rceil$ due to the value of $k_n$ provided in the theorem statement.

    \xhdr{Case (ii) [Supply-to-demand ratio $\imbalanceratio\in(0,(1-\Reglevel)^{1/\Reglevel})$]:}
    In this case, since truncation $\reserve \leq (1-\Reglevel)^{-1/\Reglevel}$ by construction (\Cref{def:truncated generalized Pareto distribution}), it guarantees that $\reserve \in [1, n/m_n]$ and thus the same argument as Case (i.a) applies.

    \smallskip
    \noindent
    Combining all the pieces together, we complete the proof of \Cref{thm:bk supply limiting:large market} as desired.
\end{proof}

\end{document}